\documentclass[12pt]{article}
\usepackage{amsmath}
\usepackage{graphicx,psfrag,epsf}
\usepackage{enumerate}
\usepackage{natbib}
\usepackage{url} 

\usepackage[margin = 2.3cm]{geometry}

\usepackage{amssymb, amsthm, ascmac, authblk, bm, booktabs, caption, comment, enumerate, latexsym, lscape, longtable, mathrsfs, mathtools, multirow, natbib, newtxtext, psfrag, setspace, subcaption, thmtools}
\usepackage[stable]{footmisc}

\usepackage{enumitem}

\usepackage{color}
\usepackage{pgfplots}
\pgfplotsset{compat=1.16}
\usepackage{here}
\usepackage{multirow}
\usepackage{caption}
\usepackage[colorlinks=true, citecolor=teal, linkcolor=teal, urlcolor=teal, backref=page]{hyperref}
\usepackage{threeparttable}
\usepackage{comment}
\usepackage{stackengine}
\usepackage{tikz}
\usepackage{chngcntr}
\usetikzlibrary{arrows.meta}
\usetikzlibrary{bending}

\usepackage{lineno}

\usepackage[capitalize]{cleveref}

\theoremstyle{definition}
\newtheorem{lemma}{Lemma}
\newtheorem{proposition}{Proposition}
\newtheorem{assumption}{Assumption}
\newtheorem{definition}{Definition}
\newtheorem*{rmk*}{Remark}
\newtheorem{rmk}{Remark}
\newtheorem{example}{Example}
\newtheorem{theorem}{Theorem}
\newtheorem{rl}{Rule}

\usepackage{etoolbox}

\renewcommand{\qedsymbol}{\openbox}

\newcommand{\addendmark}[2]{%
  \AtBeginEnvironment{#1}{%
    \renewcommand{\qedsymbol}{\ensuremath{#2}}%
    \pushQED{\qed}%
  }%
  \AtEndEnvironment{#1}{\popQED}%
}

\addendmark{theorem}{\blacksquare}
\addendmark{lemma}{\blacksquare}
\addendmark{proposition}{\blacksquare}
\addendmark{corollary}{\blacksquare}
\addendmark{assumption}{\blacksquare}
\addendmark{definition}{\blacksquare}
\addendmark{example}{\blacksquare}
\addendmark{rl}{\lrcorner}
\addendmark{rmk}{\lrcorner}
\addendmark{rmk*}{\lrcorner}

\newcommand{\Cov}{\text{Cov}}
\newcommand{\Var}{\text{Var}}
\newcommand{\Corr}{\text{Corr}}

\newcommand{\ul}{\underline}
\newcommand{\ol}{\overline}

\allowdisplaybreaks

\begin{document}
\baselineskip=18pt

\title{Accuracy of Uniform Inference on Fine Grid Points\thanks{First version: December 21, 2025. This version: \today}}
\author[1]{Shunsuke Imai \thanks{This research was supported by JSPS KAKENHI Grant Number 24KJ1472. 
}}

\affil[1]{Graduate School of Economics, Kyoto University}

\maketitle

\begin{abstract}
Uniform inference for nonparametric function estimation is typically implemented by replacing the supremum over a continuum with the maximum over a finite evaluation grid.
This paper investigates how this discretization affects the coverage accuracy of uniform confidence bands.
Surprisingly, numerical experiments show that coverage can deteriorate as the sample size increases when the number of grid points is fixed.
We derive general upper bounds on coverage probability and coverage error for uniform confidence bands based on multiplier-bootstrap critical values over grid points.
For local empirical processes arising from kernel smoothing, we establish primitive conditions under which coverage converges to zero, and conditions ensuring asymptotically nominal coverage.
These conditions are sharp in terms of the growth rate of the number of evaluation points.
Based on these results, we develop quantitative grid selection rules for a user-specified discretization tolerance level.
Monte Carlo experiments show that the proposed rules select nearly optimal grid sizes asymptotically, although they can be somewhat optimistic in small samples.
Further numerical investigations suggest that studentization error is an important source of this optimism.
Motivated by this finding, we derive a small-sample adjustment that substantially improves this optimism considering the variability of studentization.

\vspace{2mm}

\noindent \textit{Keywords}: 
bootstrap;
empirical process;
high-dimensional central limit theorem; 
kernel smoothing;
nonparametric inference;
uniform confidence bands
\end{abstract}

\newpage

\section{Introduction} \label{sec:introduction}

In many statistical applications, the object of interest is an unknown function.
Uniform confidence bands are designed to cover the true function simultaneously over a specified region at a prescribed confidence level, thereby quantifying sampling uncertainty in the estimation.
This simultaneous coverage allows researchers, for example, to assess whether a conditional treatment effect is positive throughout a range of covariate values.
The study of uniform inference in nonparametric estimation has a long history, including the early work of \citet{bickel1973some} on kernel density estimation.

Nowadays, by virtue of the seminal studies by \cite{chernozhukov2014anti,chernozhukov2014gaussian} (hereafter CCK), their extensions like \cite{chernozhukov2016empirical,ChKa20} and other important approximation frameworks such as those developed by, for example, \cite{rio1994local,koltchinskii1994komlos,gine2010confidence,cattaneo2025strong}, we can readily validate the uniform Gaussian approximation and the implementation of uniform critical value for a variety of possibly non-Donsker nonparametric estimators.  
Specifically, CCK established approximation result between suprema of empirical processes and suprema of Gaussian processes under fairly general setting.
Therefore, using the quantile of supremum of approximating Gaussian process as the critical value results in the confidence band which attains asymptotically correct coverage.
More importantly for practical implementation, CCK also show that multiplier bootstrap-based critical values can be used as a feasible critical value.

Owing to its generality and ease of application, CCK-type sup-approximation has been widely used in statistics and related fields to justify both analytic and bootstrap-based critical values for uniform inference.
Applications include conditional average treatment effects \citep{lee2017doubly,fan2022estimation,imai2026catt}, marginal treatment effects \citep{tsuda2025uniform}, effects of continuous treatments \citep{colangelo2026double}, unconditional quantile partial effects \citep{sasaki2022unconditional}, and functions under measurement error \citep{kato2018uniform,kato2019eiv,hou2026quantile}.
Further applications include uniform inference in econometric models, such as nonseparable triangular models \citep{ma2023ite}, first-price auction models \citep{ma2019inference}, and nonparametric instrumental-variable models \citep{babii2020honest}.
Of course, these examples are not exhaustive.

However, in practice, neither supremum of Gaussian process nor supremum of bootstrap statistic over continuous interval is computationally feasible.
Accordingly, implementations replace supremum over continuous interval by maximum over a finite evaluation grid and compute critical values from the corresponding maximum of bootstrap statistics over finite grid points.
This raises the question of whether the error introduced by this replacement can be ignored, and how fine the evaluation grid must be to ensure valid uniform inference while keeping computation feasible. 
In this study, we provides a quantitative analysis of the error introduced by the implementation on the finite grid points and address the problem.

In \cref{subsec:numerical_local}, we first conduct numerical simulations of uniform confidence bands for kernel density estimators.
The simulations show that computing bootstrap critical values on coarse grids can substantially reduce uniform coverage.
Surprisingly, the coverage loss induced by discretization can increase with sample size when the number of grid points is fixed, even as the coverage of the sup-bootstrap confidence band approaches the nominal level.

In \cref{sec:general_lower}, we then derive a general upper bound on the uniform coverage probability of confidence bands constructed using  critical values based on multiplier-bootstrap over finite grid points.
The bound distinguishes the contribution of discretization from that of the remaining approximation error arising from Gaussian approximation, the multiplier bootstrap, and studentization on the grid points.
The result is stated for general empirical processes and gives conditions under which coverage converges to zero.
Then, for local empirical processes (a class of empirical processes arising from nonparametric estimators based on kernel smoothing), we translate the general result into a more interpretable upper bound on the coverage probability and find that coverage probability converges to zero for any deterministic grid satisfying $ph_n/\sqrt{\log(h_n^{-1})}\to0$, where $p$ is the number of grid points and $h_n$ denotes the bandwidth.
Notably, this condition allows the number of evaluation points to diverge and indicates that increasing the number of evaluation points does not by itself prevent coverage failure.

Next, in \cref{sec:general_upper_bound}, we derive a general upper bound on coverage error, defined as the absolute difference between the coverage probability and the nominal level.
The bound separates discretization error (with an explicitly identified constant) from the remaining approximation errors arising from Gaussian approximation, the multiplier bootstrap, and studentization.
Then, for local empirical processes, we further derive an explicit leading constant in the discretization error bound.
In particular, we show that, under a vanishing high-dimensional bootstrap approximation error, asymptotically nominal coverage is attained for equally spaced grids satisfying $ph_n/\sqrt{\log(h_n^{-1})}\to\infty$.

The conditions for coverage failure and valid uniform inference are sharp in terms of the growth rate of the number of grid points. 
Specifically, the two results identify the same threshold, $h_n^{-1}\sqrt{\log(h_n^{-1})}$, separating the regimes of coverage failure and asymptotically valid uniform inference; coverage probability converges to zero for any deterministic grid with $p\ll h_n^{-1}\sqrt{\log(h_n^{-1})}$, whereas asymptotically nominal coverage is attained on an equally spaced grid when $h_n^{-1}\sqrt{\log(h_n^{-1})}\ll p$. 

Beyond this asymptotic characterization, we develop quantitative grid selection rules for a fixed user-chosen tolerance level of discretization error, which can be readily implemented by searching over candidate grid sizes. 

In \cref{sec:numerical_experiments}, we assess the performance of the proposed quantitative grid selection rules.
In \cref{subsec:gaussian_bound_numerical}, we compare the grid sizes selected by the proposed rules with those obtained from Monte Carlo estimates of the discretization error of the Gaussian process.
Our proposed rules and Monte Carlo estimates select nearly identical grid sizes for a given discretization tolerance level.
This suggests that our discretization bound is remarkably sharp when the remaining approximation errors arising from Gaussian approximation, the multiplier bootstrap, and studentization are negligible.
In \cref{subsec:finite_sample_grid_numerical}, we evaluate the finite-sample performance of these rules under the same setting as \cref{subsec:numerical_local}.
The results show that the rules are asymptotically optimal but can be somewhat optimistic in small samples.
In \cref{subsec:sources_small_sample_optimism}, we investigate the sources of this optimism by decomposing the discretization error and find that studentization error accounts for most of the discrepancy between the discretization error of Gaussian process and the actual discretization error.

Motivated by this finding, in \cref{subsec:student_t_grid_selection}, we propose a small-sample adjustment to the grid selection rules that accounts for the variability of variance estimation, drawing on the insights of \citet{Welch1951,BellMcCaffrey2002,ImbensKolesar2016}.
Specifically, we approximate the distribution of the variance estimator by a scaled chi-square distribution through moment matching and construct a Student-$t$ approximating process.
We then derive a discretization error bound for this Student-$t$ process and propose a grid selection rule based on the resulting bound.
Numerical experiments in \cref{subsec:experiment_student_t} show that the proposed adjustment substantially mitigates the small-sample optimism without introducing excessive conservativeness.

\paragraph{Related Literature:}
As a related problem, the effect of discretization of Gaussian process has been studied in the literature on extreme-value theory. 
In particular, \citet{piterbarg2004discrete} characterized the joint limiting distribution of suitably normalized maxima of a stationary Gaussian process over an expanding interval and those over a uniform grid.
The result relies on some assumptions on the structure of covariance function, and characterizes discretization through the joint extreme-value limit of the two maxima.
Related results have also been established for locally stationary Gaussian processes \citep{husler2004dependence}, multivariate stationary Gaussian processes \citep{tan2014multivariate,hashorva2015grids}, and homogeneous Gaussian random fields \citep{tan2015homogeneous}, while retaining structural assumptions on the correlation functions.
In a different direction, \citet{azais2020discretization} studied the discretization error for the maximum of a smooth Gaussian field over a fixed square.
Under non-degeneracy, smoothness and stationality conditions, they derived asymptotic bounds for the normalized difference between the continuous maximum and the maximum over a grid, using the Hessian at the location of the maximum.

More directly related to uniform inference on finite grid points, \citet{hall1988confidence} constructed confidence bands for unknown density and regression functions by first obtaining simultaneous confidence intervals for local averages over subintervals and then interpolating their upper and lower endpoints.
To account for interpolation error, they enlarge the confidence bands using upper bounds on derivatives of the target functions, which should be specified in practice.
\citet{horowitz2012uniform} constructed uniform confidence bands for series-based nonparametric instrumental variable estimators over finite grid points.
Their analysis relies on the multivariate Berry--Esseen bound of \citet{bentkus2003dependence}, which yields Gaussian approximation error bound that depends polynomially on the number of grid points.
They also impose a condition on the smallest eigenvalue of the covariance matrix of the approximating Gaussian vector, which requires the number of grid points not to exceed the sieve dimension.
Also, they noted that a quantitative rule for choosing the number of grid points is out of scope of their study.

Our analysis builds on the Gaussian approximation framework for maxima of high-dimensional sample means and suprema of empirical processes \citep{CCK13,chernozhukov2014anti,chernozhukov2014gaussian} and its subsequent developments. In particular, we make extensive use of the high-dimensional Gaussian and bootstrap approximation results and accompanying technical inequalities of \citet{CCKK22} to control the errors arising from Gaussian approximation, the multiplier bootstrap, and studentization on finite grids.
These approximation bounds depend logarithmically on the number of grid points and remain applicable even when the covariance matrix of the approximating Gaussian vector is singular.
We separately quantify the additional error caused by replacing the continuum supremum with the grid maximum, using \cite{Plackett54}'s identity for bivariate Gaussian distribution functions.
Combining the discretization bound with the high-dimensional approximation results allows us to bound the coverage error of finite-grid bootstrap confidence bands over the continuum and develop quantitative grid selection rules, without requiring stationarity, specific structure of covariance function of Gaussian process or convergence of Gaussian maxima to an extreme-value distribution.

Very recently, \citet{CCKK26} substantially improved the high-dimensional Gaussian approximation bounds of \citet{CCKK22}, achieving $n^{-1/2}$ rate up to logarithmic factors under sub-exponential tail conditions without restrictions on the covariance matrix.
In our study, since the high-dimensional approximation errors enter our coverage bounds only as asymptotically negligible terms, these improvements do not affect the proposed quantitative grid selection rules.

\paragraph{Organization:}
The paper is organized as follows.
\Cref{sec:setup} presents the setup.
\Cref{subsec:numerical_local} illustrates failure of uniform inference on coarse-grid  through numerical simulations.
\Cref{subsec:general_lower} derives a general upper bound on uniform coverage probability, and \cref{subsec:kernel_failure_theory} gives primitive conditions under which coverage converges to zero for a class of empirical processes.
\Cref{subsec:general_upper} derives a general upper bound on coverage error and develops asymptotically valid grid selection rules.
\cref{subsec:uuper_bound_local} gives primitive conditions for valid uniform inference and implementable grid selection rules with a class of empirical processes.
\Cref{sec:numerical_experiments} assesses the performance of the grid selection rules.
\Cref{subsec:student_t_grid_selection} develops a small-sample adjustment to the asymptotically valid grid selection rule and evaluates its performance.
\Cref{sec:conclusion} summarizes the main findings.
All proofs are collected in the Supplement.

\section{Setup} \label{sec:setup}
Suppose that we are interested in constructing a uniform confidence band for some function $f$ over the compact region $\mathcal{X} \coloneqq[\ul{x}, \ol{x}] \subset \mathbb{R}$ and we have an estimator for $f$ of the form $\hat{f}_{h_n}(x) \coloneqq n^{-1}\sum_{i=1}^n\psi_{h_n}(X_i,x) $ and
\begin{align}
    \hat{f}_{h_n}(x) - \mathbb{E}[\hat f_{h_n}(x)] = \frac{1}{n}\sum_{i=1}^n \{\psi_{h_n}(X_i,x) - \mathbb{E}[\psi_{h_n}(X_i,x)]\}, \label{eq:linear}
\end{align}
where $\{X_i\}_{i=1}^n$ is a sample from i.i.d.~observations and $h_n$ denotes a tuning parameter possibly depending on $n$.
\begin{rmk}
To keep the paper concise and focus on the error caused by the discretization, we take $\mathbb{E}[\hat{f}_{h_n}(x)]$ as the baseline target.
However, in practice, the parameter of interest is typically $f(x)$ rather than the pseudo-true target $\mathbb{E}[\hat{f}_{h_n}(x)]$.
In that case, one may write
\begin{align*}
    \hat{f}_{h_n}(x) - f(x) = \frac{1}{n}\sum_{i=1}^n \{\psi_{h_n}(X_i,x) - \mathbb{E}[\psi_{h_n}(X_i,x)]\} + r_n(x),
\end{align*}
where $r_n(x)$ represents the linearization remainder, which includes deterministic bias and stochastic approximation errors. 
Nevertheless, the results below continue to hold provided $r_n(x)$ satisfies a uniform negligibility condition over $\mathcal{X}$ relative to the leading term (cf. Lemma 1 in \citealp{chernozhukov2023high}). 
\end{rmk}
Define $\sigma_n^2(x)\coloneqq n^{-1}\Var[ \psi_{h_n}(X_i,x)]$ and let $\hat{\sigma}^2_n(x)$ be its estimator:
\begin{align}
    \hat{\sigma}^2_n(x) \coloneqq \frac{1}{n} \left[ \frac{1}{n}\sum_{i=1}^n \psi_{h_n}^2(X_i,x) - \left( \frac{1}{n}\sum_{i=1}^n \psi_{h_n}(X_i,x)\right)^2 \right]. \label{eq:variance_estimator}
\end{align}
Define the standardized and studentized processes by
\begin{align*}
    T_n(x) \coloneqq \frac{\hat{f}_{h_n}(x) - \mathbb{E}[\hat f_{h_n}(x)]}{\sigma_n(x)}, \quad \hat{T}_n(x) \coloneqq \frac{\hat{f}_{h_n}(x) - \mathbb{E}[\hat f_{h_n}(x)]}{\hat\sigma_n(x)},
\end{align*}
respectively.
For each bootstrap iteration, we generate a set of i.i.d.~bootstrap weights $\{w_i^{\star}\}_{i=1}^n$ independently of the original data and compute
\begin{align*}
    \hat T_n^{\star}(x) \coloneqq \frac{1}{n\hat\sigma_n(x)}\sum_{i=1}^n w_i^{\star}  \{\psi_{h_n}(X_i,x) - \hat{f}_{h_n}(x)\}.
\end{align*}
For simplicity, we take the bootstrap weights to be independent standard normal variables.
Ideally, we would construct the critical value based on the conditional $(1-\alpha)$-quantile of the supremum of the bootstrap process over the continuum,
\begin{align*}
    \hat{c}_{n,1-\alpha}^{\star, \texttt{infeasible}} \coloneqq \inf\left\{  t\in\mathbb{R} : \mathbb{P}^\star\left( \sup_{x\in\mathcal{X}}|\hat T_n^\star(x)| \le t \mid X_{1:n} \right) \ge 1-\alpha\right\},
\end{align*}
where $\mathbb{P}^\star(\cdot \mid X_{1:n})$  denotes probability with respect to the bootstrap weights conditional on the original sample.
However, evaluating $\sup_{x\in\mathcal{X}}|\hat T_n^\star(x)|$ is not computationally feasible in general.
Therefore, we approximate the continuum supremum by a maximum over a finite grid.

To proceed, we introduce some notation about discrete grids.
Let $\{x_j\}_{j=1}^p\subset\mathcal X$ be a deterministic evaluation grid
with $p\ge2$ satisfying
\begin{align*}
    \ul x=x_1<x_2<\cdots<x_p=\ol x.
\end{align*}
In this study, we consider the following feasible bootstrap critical value;
\begin{align*}
     \hat{c}_{n,1-\alpha}^{\star} \coloneqq \inf\left\{  t\in\mathbb{R} : \mathbb{P}^\star\left( \max_{1\le j \le p}|\hat T_n^\star(x_j)| \le t \mid X_{1:n} \right) \ge 1-\alpha\right\},
\end{align*}
as an approximation of the critical value based on studentized statistic $\sup_{x\in\mathcal{X}}|\hat T_n(x)|$
and investigate the quality of this approximation to provide conditions under which the resulting critical value delivers valid uniform inference.

\subsection{Local Empirical Process} \label{subsec:setup_local}

We introduce a class of empirical processes referred to as local empirical processes by \citet{chernozhukov2014gaussian}.
This class arises in a wide range of nonparametric estimation problems based on kernel smoothing, including kernel density estimation, local polynomial regression, and more complex estimators built upon these methods.
By considering this class, we translate the general theoretical results developed in later sections into more interpretable conditions and provide practical guidance on how to choose the number of evaluation grid points when conducting uniform inference.

\begin{definition}[local empirical process]
With a slight abuse of notation, throughout the sections on local empirical process, we write $(Y_i,X_i)$ for the observation and use $x\in\mathcal X$ as the evaluation point.
Let $K:\mathbb R\to\mathbb R$ be a kernel function and let $h_n>0$ be a bandwidth tending to zero.
For a fixed measurable function $g$ satisfying $\mathbb E[g(Y_i)^2]<\infty$, consider
\begin{align*}
    \hat f_{h_n}(x)
    \coloneqq \frac1n\sum_{i=1}^n K_{i,h_n}(x),
    ~~\text{ with } ~~K_{i,h_n}(x)
    \coloneqq \frac{g(Y_i)}{h_n} K\left(\frac{X_i-x}{h_n}\right).
\end{align*}
In the general framework, this corresponds to taking $\psi_{h_n}((Y_i,X_i),x)=K_{i,h_n}(x)$.
The associated local empirical process is obtained by centering $\hat f_{h_n}(x)$ at its expectation.
Additional regularity conditions on the kernel and the distribution of the observations are specified below.
\end{definition}

\subsection{Series Empirical Process}\label{subsec:setup_series}

We introduce a class of empirical processes referred to as series empirical processes by \citet{chernozhukov2014gaussian}.
This class arises in a wide range of nonparametric estimation problems based on series basis functions, including polynomial, Fourier, spline, and wavelet bases.
Although this paper take local empirical processes as an example in the subsequent analysis, the conditions for general theory (\cref{as:DGP_lower} and \cref{ass:upper_general}) do not exclude series empirical processes.
Thus, analyses parallel to those for local empirical processes are possible under suitable regularity conditions.

\begin{definition}[Series empirical process]
Let $(\eta_i,X_i)_{i=1}^n$ be i.i.d. random variables and let $x\in\mathcal X$ be the evaluation point.
Let $b_{K_n}(x)=(b_{1,n}(x),\ldots,b_{K_n,n}(x))^\top$ be a vector of basis functions, where $K_n$ may depend on $n$, and let $A_n$ be a deterministic $K_n\times K_n$ matrix.
For a fixed measurable function $g$ satisfying $\mathbb E[g(\eta_i)^2]<\infty$ and $\mathbb E[g(\eta_i)\mid X_i]=0$, consider
\begin{align*}
    \hat f_{K_n}(x)&\coloneqq\frac1n\sum_{i=1}^n S_{i,K_n}(x),~~\text{ with } ~~S_{i,K_n}(x)\coloneqq b_{K_n}(x)^\top A_n b_{K_n}(X_i)g(\eta_i).
\end{align*}
In the general framework, this corresponds to taking the summand $\psi_{h_n}((\eta_i,X_i),x)=S_{i,K_n}(x)$, with the series dimension $K_n$ playing the role of the tuning parameter.
The associated series empirical process is obtained by centering $\hat f_{K_n}(x)$ at its expectation.
\end{definition}

\section{Numerical Illustration of A Failure on Coarse Grids} \label{subsec:numerical_local}

To assess the practical relevance of discretization error, we conduct a Monte Carlo simulation study of uniform inference with kernel density estimators (KDE) following \cite{chernozhukov2014anti}, and illustrate that discretization is practically relevant in the sense that it can lead to a substantial loss in the coverage probability of uniform confidence bands.

To minimize arbitrariness in the choice of data-generating processes, we adopt the normal and mixture-normal designs considered in Appendix J of the Supplementary Material for \cite{chernozhukov2014anti}.
Specifically, we generate independent observations $X_1,\ldots,X_n$ from either of the following densities:
\begin{align*}
    f(x) &= \frac{1}{1.5}\phi(x/1.5), &&\text{(Normal)},\\
    f(x) &= \frac{1}{2}\phi(x-1.5)+\frac{1}{2}\phi(x+1.5), &&\text{(Mixture normal)},
\end{align*}
where $\phi$ denotes the standard normal density.
As in \cite{chernozhukov2014anti}, we take the evaluation interval to be $\mathcal X=[-3,3]$.

We estimate $f$ using the kernel density estimator with the Epanechnikov kernel:
\begin{align*}
    \hat f_{h_n}(x)
    &\coloneqq \frac{1}{nh_n}\sum_{i=1}^n K((X_i-x)/h_n), ~\text{with}~ K(u)=\frac{3}{4}(1-u^2)\mathbf{1}\{|u|\leq1\},
\end{align*}
as in \cite{chernozhukov2014anti}.
The bandwidth is given by $h_n=c_hn^{-2/7}$, where
\begin{align*}
    c_h
    &= \left(\frac{\int_{\mathbb R}K^2(u)\,du\int_{\mathcal X}f(x)\,dx}{\{\int_{\mathbb R}u^2K(u)\,du\}^2\int_{\mathcal X}\{f''(x)\}^2\,dx}\right)^{1/5}.
\end{align*}
This constant is obtained by minimizing the leading terms in the asymptotic mean integrated squared error (AMISE) over $\mathcal X$; see, for example, Section 5.4 in \cite{wasserman2006all}.
We then multiply the corresponding AMISE-optimal bandwidth $c_hn^{-1/5}$ by $n^{-3/35}$ to obtain the undersmoothing rate $n^{-2/7}$.
This undersmoothing rate has been used in the literature on uniform inference, including  intersection bounds \citep{chernozhukov2013intersection}, conditional moment inequality \citep{andrews2014nonparametric}, conditional treatment effect functions \citep{lee2017doubly,fan2022estimation,imai2026catt}.
Let $\psi_{h_n}(X_i,x)\coloneqq h_n^{-1}K((X_i-x)/h_n)$, so that $\hat f_{h_n}(x)=n^{-1}\sum_{i=1}^n\psi_{h_n}(X_i,x)$.
The studentized statistic is given by
\begin{align*}
    \hat T_n(x)
    &\coloneqq \frac{\hat f_{h_n}(x)-f(x)}{\hat\sigma_n(x)}, \text{ with } \hat\sigma_n^2(x)
    \coloneqq \frac{1}{n}\left[\frac{1}{n}\sum_{i=1}^n\psi_{h_n}^2(X_i,x)-\left(\frac{1}{n}\sum_{i=1}^n\psi_{h_n}(X_i,x)\right)^2\right]. 
\end{align*}
Note that the smoothing bias $r_n(x)=\mathbb E[\hat f_{h_n}(x)]-f(x)$ is asymptotically negligible relative to the standard error under the undersmoothing bandwidth.
The multiplier-bootstrap process is given by
\begin{align*}
    \hat T_n^*(x)
    &\coloneqq \frac{n^{-1}\sum_{i=1}^nw^\star_i\{\psi_{h_n}(X_i,x)-\hat f_{h_n}(x)\}}{\hat\sigma_n(x)}.
\end{align*}

For each sample size $n\in\{250,500,1000,2000,5000\}$ and each number of grid points $p\in\{5,10,15,20,25,30,40,50,75,100\}$, we use an equally spaced evaluation grid containing both endpoints of $\mathcal X$.
We approximate the supremum over $\mathcal X$ by the maximum over the $3{,}201$ equally spaced points.
We use $5{,}000$ Monte Carlo replications for each data-generating process and sample size, and $999$ multiplier draws for each replication.

We compare the coverage probabilities of two nominal $95\%$ uniform confidence bands constructed using different bootstrap critical values (sup-bootstrap  and max-bootstrap).
The sup-bootstrap critical value is the conditional $0.95$-quantile of $\sup_{x\in\mathcal X}|\hat T_n^*(x)|$, whereas the max-bootstrap critical value is the conditional $0.95$-quantile of $\max_{j\in[p]}|\hat T_n^*(x_j)|$.
Both confidence bands are evaluated by whether they cover the true density $f(x)$ throughout $\mathcal X$.
For each data set, both critical values are computed using the same multiplier draws.
The difference between the two coverage probabilities measures the coverage loss induced by computing the critical value on the evaluation grid.
\cref{fig:coverage} reports the coverage probabilities and their differences, while \cref{tab:kde_coverage_loss} provides numerical values of the discretization errors.

\begin{figure}[!t]
    \centering
    \includegraphics[width=\linewidth]{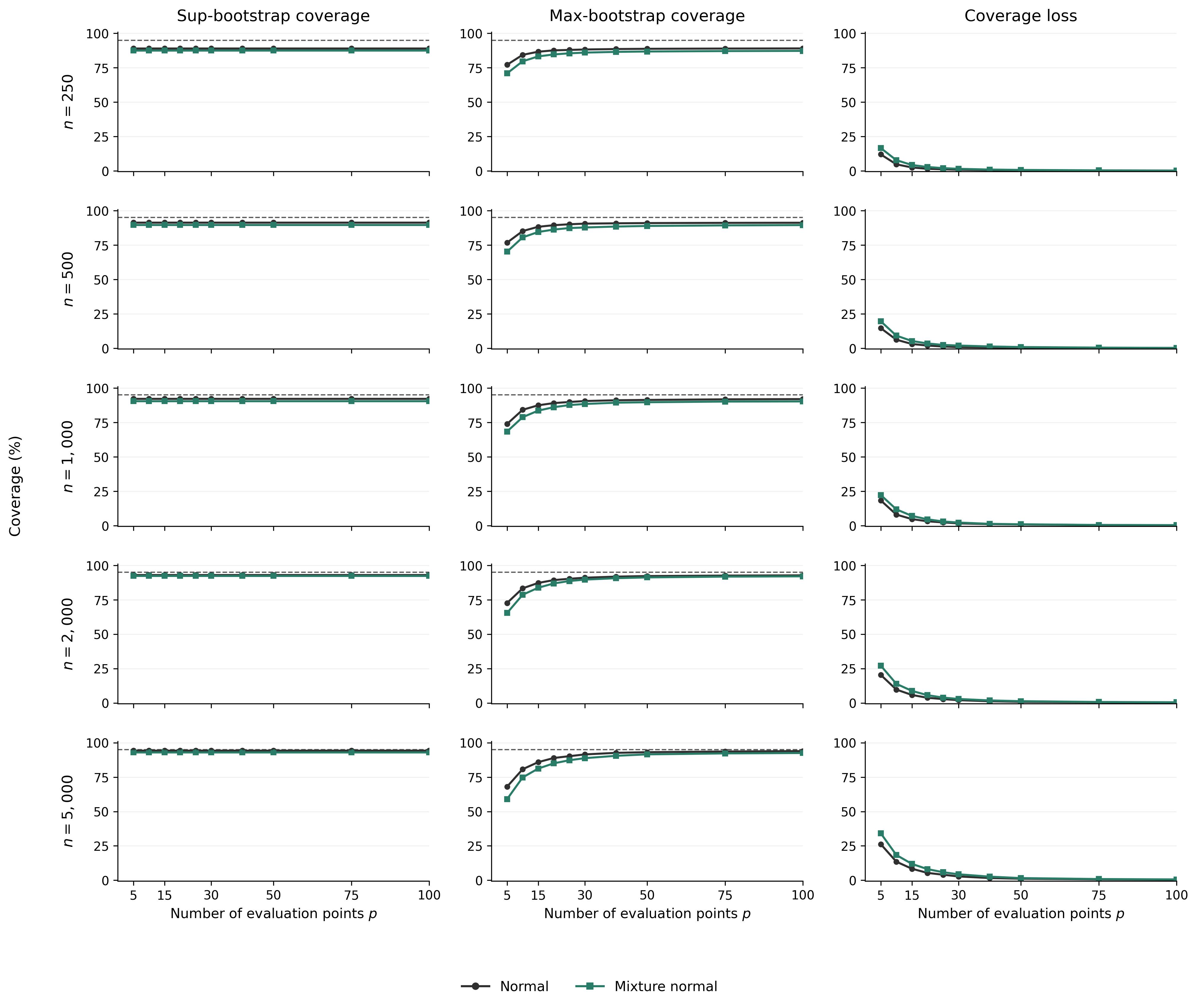}
    \caption{Uniform coverage of $95\%$ confidence bands and discretization-induced coverage loss.
    The first two columns report coverage using the sup-bootstrap and max-bootstrap critical values, respectively; both evaluate coverage of the true density $f(x)$ over the entire interval $\mathcal X$.
    The third column reports sup-bootstrap coverage minus max-bootstrap coverage in percentage points.
    Rows correspond to sample sizes, and colors to data-generating processes.
    The horizontal axis gives the number of evaluation points $p$.
    Dashed horizontal lines in the first two columns indicate the nominal coverage probability of $95\%$.}
    \label{fig:coverage}
\end{figure}

\begin{table}[htbp!]
\centering
\caption{Discretization errors in KDE confidence-band coverage.}
\label{tab:kde_coverage_loss}
\fontsize{10}{11.5}\selectfont
\renewcommand{\arraystretch}{1.08}
\setlength{\tabcolsep}{3pt}

\begin{tabular}{@{}lrrrrr@{\hspace{8pt}}rrrrr@{}}
\toprule
& \multicolumn{5}{c}{Normal}
& \multicolumn{5}{c}{Mixture normal} \\
\cmidrule(lr){2-6}\cmidrule(lr){7-11}
 $p \backslash n$
& $250$ & $500$ & $1,000$ & $2,000$ & $5,000$
& $250$ & $500$ & $1,000$ & $2,000$ & $5,000$ \\
\midrule
$5$   & 12.00 & 14.62 & 18.30 & 20.38 & 26.24 & 16.60 & 19.58 & 22.22 & 27.14 & 34.16 \\
$10$  & 4.78 & 6.20 & 8.00 & 9.72 & 13.52 & 7.82 & 9.18 & 11.82 & 13.88 & 18.38 \\
$15$  & 2.50 & 3.10 & 4.72 & 5.80 & 8.38 & 4.28 & 5.26 & 7.04 & 8.76 & 11.90 \\
$20$  & 1.52 & 1.94 & 3.12 & 3.74 & 5.36 & 2.84 & 3.48 & 4.54 & 5.64 & 8.10 \\
$25$  & 1.20 & 1.26 & 2.32 & 2.78 & 4.04 & 2.02 & 2.40 & 3.00 & 3.88 & 5.86 \\
$30$  & 0.88 & 0.80 & 1.62 & 2.00 & 2.76 & 1.50 & 1.96 & 2.20 & 2.86 & 4.34 \\
$40$  & 0.56 & 0.56 & 1.06 & 1.22 & 1.62 & 0.98 & 1.30 & 1.26 & 1.82 & 2.62 \\
$50$  & 0.38 & 0.40 & 0.82 & 0.80 & 1.14 & 0.70 & 0.86 & 0.94 & 1.26 & 1.58 \\
$75$  & 0.20 & 0.24 & 0.38 & 0.46 & 0.64 & 0.40 & 0.46 & 0.44 & 0.68 & 0.88 \\
$100$ & 0.10 & 0.16 & 0.22 & 0.30 & 0.36 & 0.26 & 0.24 & 0.34 & 0.48 & 0.60 \\
\bottomrule
\end{tabular}

\vspace{4pt}
\begin{minipage}{0.98\linewidth}
\footnotesize
\textit{Notes.}
Discretization errors are reported in percentage points
and are defined as the differences between the coverage
probabilities of the sup-bootstrap and max-bootstrap
confidence bands.
\end{minipage}
\end{table}

\cref{fig:coverage} and \cref{tab:kde_coverage_loss} reveal three patterns.
First, for a fixed sample size, the discretization error depends substantially on the number of evaluation points.
As the grid becomes finer, the coverage loss induced by discretization decreases, and the coverage probability of the max-bootstrap confidence band approaches that of the sup-bootstrap confidence band.
For example, in the mixture-normal design with $n=5000$, the discretization error decreases from $11.90$ percentage points at $p=15$ to $0.60$ percentage points at $p=100$.

Second and most importantly, for a fixed number of grid points, the discretization error generally increases as the sample size increases, even as the sup-bootstrap coverage approaches the nominal level.
An intuitive explanation for this pattern comes from comparing the grid spacing with the bandwidth.
Under $h_n=c_hn^{-2/7}$, a fixed number of equally spaced grid points on the fixed interval $\mathcal X$ implies that the grid spacing $(x_{j+1}-x_j)$ is constant, so $(x_{j+1}-x_j)/h_n$ increases with the sample size.
The evaluation grid therefore becomes progressively coarser relative to the bandwidth.
Consequently, improvements in the accuracy of the bootstrap approximation as the sample size increases may be partly offset by an increase in discretization error.
Thus, increasing the sample size does not necessarily reduce the coverage loss induced by discretization.
We investigate this phenomenon theoretically in a later section, where we establish a coverage failure result for a class of empirical processes associated with kernel-smoothing nonparametric estimators (\cref{thm:local_coverage_failure}).
This contrast illustrates why grid resolution should be chosen jointly with the bandwidth (or equivalently, the sample size).

Third, discretization can induce a substantial loss of uniform coverage even when the evaluation grid is moderately fine.
For example, in the mixture-normal design with $n=5000$, the discretization error is $4.34$ percentage points at $p=30$ and remains $1.58$ percentage points at $p=50$.

Taken together, these simulations demonstrate that discretization can induce a substantial loss of uniform coverage even in simple kernel density estimation settings with smooth densities.
The results highlight the importance of accounting for discretization error when assessing the accuracy of uniform inference.

\section{Coverage Failure on Coarse Grids}\label{sec:general_lower}

In this section, we show that discretization can cause the uniform coverage probability to converge to zero even when the number of evaluation points diverges with the sample size.
We first derive an upper bound on uniform coverage of critical values computed by the multiplier bootstrap on grid points with a general class of empirical processes (\cref{subsec:general_lower}).
We then evaluate the quantities appearing in this bound for local empirical processes and obtain an explicit failure condition in terms of the bandwidth and grid size (\cref{subsec:kernel_failure_theory}).

\subsection{General Results} \label{subsec:general_lower}

We impose the following conditions to derive an upper bound on the uniform coverage probability.
\begin{assumption}
\label{as:DGP_lower}
\mbox{}
\begin{itemize}
    \item[(i)]
    $\{X_i\}_{i=1}^n$ is a sample from i.i.d.~observations.

    \item[(ii)]
    The estimator $\hat{f}_{h_n}(x)$ has a linear form
    in \cref{eq:linear} uniformly in $x\in\mathcal X$. 
    
     \item[(iii)]
    There exists a deterministic sequence $B_n\ge1$ such that
    $\sup_{x\in\mathcal{X}}\|\tilde\psi_{i,h_n}(x)\|_{\psi_1}\le B_n$ and
    $\sup_{x\in\mathcal{X}}\mathbb E[\tilde\psi_{i,h_n}(x)^4]\le B_n^2$
    for all $1\le i\le n$, with
    \begin{align*}
    \tilde\psi_{i,h_n}(x)
    &\coloneqq
    \frac{\psi_{h_n}(X_i,x)-\mathbb E[\psi_{h_n}(X_i,x)]}
         {\sqrt{\Var[\psi_{h_n}(X_i,x)]}},
\end{align*}
and $\|\xi\|_{\psi_1}\coloneqq\inf\{b>0:\mathbb E[\exp(|\xi|/b)]\le2\}$. \qedhere
\end{itemize}
\end{assumption}

\Cref{as:DGP_lower}(iii) imposes the tail and moment bounds corresponding to Conditions E and M of \citet{CCKK22} uniformly over $x\in\mathcal X$.
Although imposing these conditions over the entire interval is stronger than imposing them only on the evaluation grid, the proof of \cref{thm:lower_general} introduces an auxiliary finite set of points in $\mathcal X$.
The uniform conditions allow us to apply the high-dimensional Gaussian approximation and control studentization errors on such auxiliary set.

Before stating the coverage bound in \cref{thm:lower_general}, we define the quantity $M_n(t)$.
Let $\Phi$ denote the standard normal distribution function and write $\bar\Phi(t)\coloneqq1-\Phi(t)$.
For $t>0$, define
\begin{align*}
    M_n(t)^{-1}
    &\coloneqq\frac1{|\mathcal X|^2}\int_{\mathcal X}\int_{\mathcal X}
    \Corr\left(\mathbf1\{|G_n(x)|>t\},\mathbf1\{|G_n(y)|>t\}\right)\,dx\,dy,
\end{align*}
where $\{G_n(x):x\in\mathcal{X}\}$ is a centered Gaussian process with covariance $\mathbb E[G_n(x)G_n(y)]=\mathbb E[\tilde\psi_{1,h_n}(x)\tilde\psi_{1,h_n}(y)]$ for $x,y\in\mathcal X$.

\begin{theorem}[Upper bound on uniform coverage probability]
\label{thm:lower_general}
Under \cref{as:DGP_lower}, there exists a universal constant $C>0$ such that, for every $\alpha\in(0,1)$, letting $t_{\alpha,p}\coloneqq\Phi^{-1}(1-\alpha/(2p))$,
\begin{align*}
    &\mathbb P\left(\sup_{x\in\mathcal X}|\hat T_n(x)|\le\hat c^\star_{n,1-\alpha}\right)\le\frac1{n\alpha}+\frac{p}{\alpha M_n(t_{\alpha,p})}
    +C\left\{\frac{B_n^2\log^5(2np)}{n}\right\}^{1/4}. \qedhere
\end{align*}
\end{theorem}

The key quantity in \cref{thm:lower_general} is $p/M_n(t_{\alpha,p})$, where $M_n(t)^{-1}$ measures the average correlation between the events that the absolute Gaussian process exceeds the threshold $t$ at two locations. The theorem establishes that, if $p/M_n(t_{\alpha,p})\to0$ and the Gaussian approximation error vanishes, the coverage probability of the grid-based uniform confidence band converges to zero.
In the following subsection, we provide primitive sufficient conditions for this result in the case of local empirical processes, expressed in terms of the bandwidth and the number of grid points.

\subsection{Application to Local Empirical Processes}\label{subsec:kernel_failure_theory}

We now specialize \cref{thm:lower_general} to the local empirical process introduced in \cref{subsec:setup_local}.

\begin{assumption}
\label{as:DGP_local_lower}
\mbox{}
\begin{itemize}
    \item[(i)]
    The observations $\{(Y_i,X_i)\}_{i=1}^n$ are i.i.d., and $g$ is a measurable function satisfying $\|g\|_\infty<\infty$.

    \item[(ii)]
    The random variable $X_i$ has a continuously differentiable density $f$ satisfying $\sup_{x\in\mathbb R}\{|f(x)|+|f'(x)|\}<\infty$ and $\inf_{x\in\mathcal X}f(x)>0$.
    Also, $x\mapsto\mathbb E[g(Y_i)^2\mid X_i=x]$ is continuous on an open neighborhood of $\mathcal X$ and bounded away from zero on $\mathcal X$.

    \item[(iii)] The kernel $K:\mathbb R\to\mathbb R$ is Lipschitz continuous, supported on $[-1,1]$, and  $\int_{\mathbb R}K(u)\,du=1$.

    \item[(iv)]
    The bandwidth $h_n$ satisfies $h_n\to0$ and $nh_n/\log(h_n^{-1})\to\infty$. \qedhere
\end{itemize}
\end{assumption}

Under \cref{as:DGP_local_lower}, the following two lemmas evaluate the quantities appearing \cref{thm:lower_general}.
\Cref{lem:local_1} gives moment and tail bounds from which we construct $B_n$.
\Cref{lem:local_lower_M} provides a bound on $1/M_n(t_{\alpha,p})$ in \cref{thm:lower_general}.

\begin{lemma}[Evaluation of $B_n$  for local empirical processes]\label{lem:local_1}
Under \cref{as:DGP_local_lower}, the following statements hold.
\begin{itemize}
    \item[(i)]
    Uniformly over $x\in\mathcal X$,
    \begin{align*}
        \Var[K_{i,h_n}(x)]
        = \frac{f(x)\mathbb E[g(Y_i)^2\mid X_i=x]}{h_n}\int_{\mathbb R}K(u)^2\,du +o(h_n^{-1}) 
        = O(h_n^{-1}).
    \end{align*}

    \item[(ii)]
    For any deterministic evaluation grid
    $\{x_j\}_{j=1}^p\subset\mathcal X$,
    \begin{align*}
        \max_{1\le j\le p} \mathbb E[\tilde \psi_{i,h_n}(x_j)^4]
        &\le \max_{1\le j\le p} \frac{16\|g\|_\infty^4\sup_{x\in\mathbb R }f(x)}{h_n^3\Var[K_{i,h_n}(x_j)]^2} \int_{\mathbb R}K(u)^4\,du,\\
        \max_{1\le j\le p} \|\tilde\psi_{i,h_n}(x_j)\|_{\psi_1}
        &\le \max_{1\le j\le p} \frac1{\log2}
        \frac{2\|g\|_\infty\sup_{u\in\mathbb R}|K(u)|}{h_n\sqrt{\Var[K_{i,h_n}(x_j)]}}. \qedhere
    \end{align*}
\end{itemize}
\end{lemma}
From \cref{lem:local_1}(ii), we can take $B_n$ as
\begin{align*}
    B_n = 1 \vee \max_{1\le j\le p}
        \frac1{\log2}
        \frac{
            2\|g\|_\infty\sup_{u\in\mathbb R}|K(u)|}{h_n\sqrt{\Var[K_{i,h_n}(x_j)]}} \vee \max_{1\le j\le p}\left[
        \frac{ 16\|g\|_\infty^4\sup_{x\in\mathbb R}f(x)}{h_n^3\Var[K_{i,h_n}(x_j)]^2} \int_{\mathbb R}K(u)^4\,du\ \right]^{1/2}.
\end{align*}
so that it satisfies \cref{as:DGP_lower}(ii).
Moreover, this result and \cref{lem:local_1}(i) imply $B_n = O(h_n^{-1/2})$.

\begin{lemma}[Evaluation of $M_n(t_{\alpha,p})$  for local empirical processes] \label{lem:evaluation_M}
\label{lem:local_lower_M}
Under \cref{as:DGP_local_lower}, uniformly over $p\ge2$ and $\alpha\in(0,1)$,
\begin{align*}
    M_n(t_{\alpha,p})^{-1}
    &=O\left(\frac{h_n}{|\mathcal X|\sqrt{\log(2p)}}\right). \qedhere
\end{align*}
\end{lemma}

Combining \cref{lem:local_1,lem:local_lower_M} with \cref{thm:lower_general}, we obtain the following upper bound on the uniform coverage probability.

\begin{theorem}[Coverage failure for local empirical processes]
\label{thm:local_coverage_failure}
Under \cref{as:DGP_local_lower}, there exists a constant $C>0$ such that, for every $\alpha\in(0,1)$ and any deterministic grid $\{x_j\}_{j=1}^{p}\subset\mathcal X$ with $p\ge2$,
\begin{align*}
    &\mathbb P\left(\sup_{x\in\mathcal X}|\hat T_n(x)|\le\hat c^\star_{n,1-\alpha}\right)\le
    \frac{C}{\alpha}\frac{ph_n}{\sqrt{\log(2p)}}
    +C\left\{\frac{\log^5(2np)}{nh_n}\right\}^{1/4} + o(1).
\end{align*}
In particular, since the condition $ph_n/\sqrt{\log(2p)}\to0$ is equivalent to $ph_n/\sqrt{\log(h_n^{-1})}\to0$ (see Appendix \cref{subsec:rmk_equiv}),
it holds that
\begin{align*}
    &\frac{\log^5(2np)}{nh_n}\to0,
    \quad
    \frac{ph_n}{\sqrt{\log(h_n^{-1})}}\to0 \implies \mathbb P\left(\sup_{x\in\mathcal X}|\hat T_n(x)|\le\hat c^\star_{n,1-\alpha}\right)\to0.   \qedhere
\end{align*}
\end{theorem}

Fix $\alpha\in(0,1)$ and suppose that $\log^5(2np)/(nh_n)\to0$, so that the high-dimensional bootstrap approximation error vanishes.
\Cref{thm:local_coverage_failure} shows that the uniform coverage probability can converge to zero even when $p\to\infty$ as $n\to\infty$.
In particular, coverage converges to zero whenever $p \ll h_n^{-1}\sqrt{\log (h_n^{-1})}$, a condition that allows $p\to\infty$.
Thus, increasing the number of evaluation points does not by itself guarantee valid uniform inference.

\section{Coverage Validity on Fine Grids} \label{sec:general_upper_bound}

The preceding section shows that divergence of the number of evaluation points alone does not guarantee valid uniform inference.
We now derive a upper bound on the uniform coverage error and use it to obtain sufficient conditions for asymptotically valid uniform coverage (\cref{subsec:general_upper}).
We then specialize the bound to local empirical processes to obtain sufficient number of grid points (\cref{subsec:uuper_bound_local}).

\subsection{General Results} \label{subsec:general_upper}

To evaluate the coverage error, roughly speaking, we separate the discretization error from the remaining approximation errors arising from Gaussian approximation, the multiplier bootstrap, and studentization.
The latter can be controlled using the high-dimensional approximation results of \citet{CCKK22}.
Thus, up to these approximation errors, it suffices to study the discretization error of a Gaussian comparator.

To this end, we introduce the following quantities associated with the Gaussian comparator.
Let $\{G_n(x):x\in\mathcal X\}$ be a centered Gaussian process with the same covariance function as $\{T_n(x):x\in\mathcal X\}$.
Since $T_n(x)$ is standardized, each $G_n(x)$ has unit variance.
For $x,y\in\mathcal X$, define
\[
    d_n(x,y)
    :=\arccos\{\mathbb E[G_n(x)G_n(y)]\}.
\]
For each $j\in[p-1]$, let
\[
    \mathcal L_{n,j}
    :=\sup_{\substack{x_j=u_0<\cdots<u_k=x_{j+1}\\k\ge1}}
    \sum_{a=1}^k d_n(u_{a-1},u_a).
\]
Also, define the total variation over $\mathcal X$ induced by $d_n$ as
\[
    \mathcal L_n
    :=\sup_{\substack{\ul x=u_0<\cdots<u_k=\ol x\\k\ge1}}
    \sum_{a=1}^k d_n(u_{a-1},u_a).
\]
By the additivity of total variation, $\mathcal L_n=\sum_{j=1}^{p-1}\mathcal L_{n,j}$.
Finally, for $t\ge0$ and $\theta\in[0,\pi]$, define
\[
    H_t(\theta)
    :=\int_0^\theta
        \left[
            1-\exp\left\{
                -\frac{t^2}{2}\tan^2\left(\frac{s}{2}\right)
            \right\}
        \right]ds.
\]

\begin{rmk}[Motivation for the quantities] \label{rmk:motivation_cos}
To evaluate the discretization error of the Gaussian comparator, we bound the probability that $G_n$ exceeds a threshold $t$ between evaluation points $x_j$ and $x_{j+1}$ (without exceeding $t$ at either points  $x_j$ and $x_{j+1}$).
For this purpose, we consider finer partitions $x_j=u_0<\cdots<u_k=x_{j+1}$ and evaluate the probabilities that $G_n(u_{a-1})\le t<G_n(u_a)$ for consecutive points.
In the evaluation, we use a classical identity for the bivariate Gaussian distribution function due to \citet{Plackett54}.
As shown in \cref{lem:gaussian_two_point_crossing} in the Supplement, this identity (in conjunction with some additional arguments) yields
\begin{align*}
    \mathbb P(G_n(x)\le t<G_n(y))=\frac{e^{-t^2/2}}{2\pi}\{d_n(x,y)-H_t(d_n(x,y))\}
\end{align*}
for every $t\ge0$ and $x,y\in\mathcal X$.
This implies that $d_n$ and $H_t$ characterize the probability that $G_n$ is at or below $t$ at one point and above $t$ at another.
These quantities allow us to obtain a quite sharp bound on the discretization error given below.
\end{rmk}

We then impose the following conditions for the general upper bound on the uniform coverage error.

\begin{assumption} \mbox{}
\label{ass:upper_general}
\begin{enumerate} 
    \item[(i)] $\{X_i\}_{i=1}^n$ is an i.i.d. sample.

    \item[(ii)] The estimator $\hat{f}_{h_n}(x)$ has a linear form in \cref{eq:linear} for all $x\in\mathcal X$. 
       
    \item[(iii)]
    The class of standardized summands $\{\tilde\psi_{h_n}(\cdot,x):x\in\mathcal X\}$ is pointwise measurable; that is, there exists a countable subset $\widetilde{\mathcal X}_n\subseteq\mathcal X$ such that, for every $x\in\mathcal X$, there exists a sequence $\{x_m\}_{m\ge1}\subseteq\widetilde{\mathcal X}_n$ satisfying $\tilde\psi_{h_n}(w,x_m)\to\tilde\psi_{h_n}(w,x)$ for every observation $w$.

    \item[(iv)]
    There exists a deterministic constant $B_n\ge1$ such that $\sup_{x\in\mathcal{X}}\|\tilde\psi_{i,h_n}(x)\|_{\psi_1}\le B_n$, and $\sup_{x\in\mathcal{X}}\mathbb E[\tilde\psi_{i,h_n}(x)^4]\le B_n^2$.
    
    \item[(v)] For each $n$, $\mathcal L_n<\infty$ .
     \qedhere
\end{enumerate}
\end{assumption}

\cref{ass:upper_general}(iii)--(v) play roles analogous to Assumptions (A1)--(A3) of \citet{chernozhukov2014gaussian}, respectively.
\cref{ass:upper_general}(iii) is same as Assumption (A1) of \citet{chernozhukov2014gaussian} and imposes so-called pointwise measurability condition. 
See \cite{chernozhukov2014gaussian} and Section~2.3 in \cite{van1996weak} for details.
\cref{ass:upper_general}(iv) imposes the uniform tail and moment bounds of \citet{CCKK22} (Conditions E and M), serving a role analogous to the moment condition (A2) of \citet{chernozhukov2014gaussian}, although the conditions are different.
Finally, \cref{ass:upper_general}(v) requires finite total variation induced by $d_n$, which is stronger than the pre-Gaussianity (the existence of a tight, centered Gaussian process with the same covariance function as the empirical process) condition (A3) of \citet{chernozhukov2014gaussian}.

The finite total variation condition in \cref{ass:upper_general}(v) plays two roles.
First, it allows us to derive an explicit bound on the discretization error of the Gaussian comparator.
As described in \cref{rmk:motivation_cos}, we evaluate this error using two-point probabilities over finer partitions of each interval between consecutive evaluation points. For $t\ge0$ and any partition $x_j=u_0<\cdots<u_k=x_{j+1}$, the identity in \cref{rmk:motivation_cos} and the definition of $\mathcal L_{n,j}$ give
\begin{align*}
    \sum_{a=1}^k\mathbb P\{G_n(u_{a-1})\le t<G_n(u_a)\}&\le\frac{e^{-t^2/2}}{2\pi}\sum_{a=1}^k d_n(u_{a-1},u_a)\le\frac{e^{-t^2/2}}{2\pi}\mathcal L_{n,j}.
\end{align*}
Thus, finite total variation controls these sums uniformly over all such partitions, including arbitrarily fine ones.
Second, \cref{ass:upper_general}(v), together with \cref{ass:upper_general}(iii)--(iv), allows us to approximate the supremum of the studentized process over $\mathcal X$ by its maximum over a deterministic finite set containing the evaluation grid, while controlling both the approximation error and the cardinality of the finite set.
This allows us to apply the finite-dimensional Gaussian approximation results of \citet{CCKK22} without losing control of the approximation error arising from the number of points in the finite set; see \cref{lem:upper_finite_approximation} and its proof in the Supplementary Material.

\begin{rmk}[Scope of \cref{ass:upper_general}(v)]
\Cref{ass:upper_general}(v) requires finite total variation induced by $d_n$ for each $n$. 
For example, standardized Brownian bridges and stationary Ornstein--Uhlenbeck processes do not satisfy \cref{ass:upper_general}(v).
For these processes, it holds that $d_n(x,y)\asymp |x-y|^{1/2}$. 
Thus, for an equally spaced partition $a=u_0<\cdots<u_m=b$ of a fixed interior interval,
\begin{align*}
    \sum_{k=1}^m d_n(u_{k-1},u_k)&\gtrsim\sum_{k=1}^m|u_k-u_{k-1}|^{1/2}=\sqrt{m(b-a)}\to\infty,
\end{align*}
as $m\to\infty$ implying $\mathcal L_n=\infty$.
Consequently, the condition excludes direct applications to, for example, ordinary empirical distribution functions in finite sample.b

Nevertheless, \cref{ass:upper_general}(v) only requires $\mathcal L_n<\infty$ for each $n$, allowing $\mathcal L_n\to\infty$ as $n\to\infty$.
Thus, Gaussian comparators satisfying the condition may converge weakly to standardized Brownian bridges or Ornstein--Uhlenbeck processes.
For example, kernel-smoothed empirical distribution functions can have Gaussian comparators satisfying $\mathcal L_n<\infty$ for each $n$, while their Gaussian limits are standardized Brownian bridges; see \citet[Corollary~2]{gine2008uniform}.

Importantly, \cref{ass:upper_general}(v) accommodates both local and series empirical processes.
For local empirical processes, the condition is satisfied by compactly supported Lipschitz continuous kernels, such as the Epanechnikov and triweight kernels, under the density and moment conditions in \cref{as:DGP_local_lower}.
For series processes, the condition is satisfied when the basis functions have bounded variation, along with standard moment conditions.
This covers polynomial, Fourier, spline, and local polynomial partition bases, as well as continuously differentiable Cohen--Daubechies--Vial wavelet bases.
These correspond to the five univariate basis families discussed in \citet[Examples~3.1--3.5]{belloni2015some}.
Thus, although stronger than pre-Gaussianity, \cref{ass:upper_general}(v) remains applicable to a broad range of nonparametric estimators.
\end{rmk}

We now derive a general upper bound on the uniform coverage error that explicitly quantifies the contribution of discretization.

\begin{theorem}[Upper bound on uniform coverage error]
\label{thm:upper_general}
Suppose that Assumption~\ref{ass:upper_general} holds. Let $q_n(u)$ denote the $u$-quantile of $\max_{1\le j\le p}|G_n(x_j)|$, with $q_n(0)=0$. There exists a constant $C>0$ such that, for every $\alpha\in(0,1)$,
\begin{align}
    &\left|\mathbb P\left(\sup_{x\in\mathcal X}|\hat T_n(x)|\le\hat c_{n,1-\alpha}^{\star}\right)-(1-\alpha)\right| \nonumber\\
    &\le \underbrace{\frac{e^{-t_n^2/2}}{\pi}\sum_{j=1}^{p-1}\left\{\mathcal L_{n,j}-d_n(x_j,x_{j+1})+H_{t_n}(d_n(x_j,x_{j+1}))\right\}}_{\text{Discretization Error}}+\underbrace{C\left\{\frac{B_n^2\log^5\{2np(1+\mathcal L_n)\}}{n}\right\}^{1/4}}_{\text{Remaining Approximation Errors}},
    \label{eq:upper_general}
\end{align}
where
\[
    t_n:=q_n\left(\left[1-\alpha-C\left\{\frac{B_n^2\log^5(2np)}{n}\right\}^{1/4}\right]_+\right),
    \qquad [a]_+:=\max\{a,0\}. \qedhere
\]
\end{theorem}

The right-hand side of \cref{eq:upper_general} separates the discretization error from the remaining approximation errors.
The first term provides an explicit bound on the discretization error that depends on the number of the evaluation grid points.
The second term collects the remaining approximation errors arising from finite-dimensional Gaussian approximation, the multiplier bootstrap, and studentization.
The factor $1+\mathcal L_n$ appears in this term because our proof introduces an auxiliary finite set, distinct from the evaluation grid, to apply the high-dimensional approximation results of \citet{CCKK22} to the supremum over continuum.
The approximation error associated with this auxiliary set is absorbed into the second term.

\begin{rmk}[Computation of the Gaussian quantile]\label{rmk:unknown_quantile}
Recall that $t_n$ in \cref{thm:upper_general} is defined as $t_n=q_n\left(\left[1-\alpha-C\left\{B_n^2\log^5(2np)/n\right\}^{1/4}\right]_+\right)$.
When $B_n^2\log^5(2np)/n\to0$, the adjustment to the quantile level is asymptotically negligible, and hence $t_n=q_n(1-\alpha)+o(1)$ for every fixed $\alpha\in(0,1)$.
Although $q_n(1-\alpha)$ depends on the generally unknown covariance structure of $\{G_n(x_j):j\in[p]\}$, the Gaussian multiplier bootstrap already provides an estimate of the quantile, so no additional covariance estimation is required for this purpose.
\end{rmk}

\begin{rl}[Quantitative grid selection]
\label{rmk:grid_unexpanded}
Let $\tau\in(0,1)$ be a user-specified tolerance level for the discretization error.
Motivated by \cref{thm:upper_general}, we consider a quantitative grid selection rule that retains the original discretization bound, replacing $t_n$ by $q_n(1-\alpha)$ as discussed in \cref{rmk:unknown_quantile}.
Specifically, for a given family of evaluation grids, we select the smallest number of evaluation points $p$ satisfying
\begin{align}
    &\frac{e^{-q_n(1-\alpha)^2/2}}{\pi}\sum_{j=1}^{p-1}\left\{\mathcal L_{n,j}-d_n(x_j,x_{j+1})+H_{q_n(1-\alpha)}(d_n(x_j,x_{j+1}))\right\}\le\tau. \nonumber \qedhere
\end{align}
\end{rl}

\subsection{Application to Local Empirical Process} \label{subsec:uuper_bound_local}

We now translate the general coverage error bound into concrete guidance for choosing the evaluation grid in practice. 
We  consider the implementation of the general grid selection rule in \cref{rmk:grid_unexpanded} for local empirical processes.

\begin{rl}[Implementation of the quantitative grid selection]
\label{rmk:local_grid_unexpanded}
Recall that $K_{1,h_n}(x)=g(Y_1)K((X_1-x)/h_n)/h_n$.
Suppose that $x\mapsto K_{1,h_n}(x)$ is continuously differentiable in mean square and $\Var[K_{1,h_n}(x)]>0$ for every $x\in\mathcal X$.
Writing
\begin{align*}
    \rho_n(x,y)=\frac{\Cov[K_{1,h_n}(x),K_{1,h_n}(y)]}{\{\Var[K_{1,h_n}(x)]\Var[K_{1,h_n}(y)]\}^{1/2}},
\end{align*}
the grid selection criterion in \cref{rmk:grid_unexpanded} can be written\footnote{we use an integral representation of $\mathcal L_{n,j}$ to avoid directly evaluating the supremum over partitions in its definition; see Step 1 of the proof of \cref{lem:local_1_disc}(ii) for the derivation.} as
\begin{align}
    &\frac{e^{-q_n(1-\alpha)^2/2}}{\pi}\sum_{j=1}^{p-1}\Bigg[
    \int_{x_j}^{x_{j+1}}
    \left\{
    \frac{\Var[\partial_x K_{1,h_n}(x)]}{\Var[K_{1,h_n}(x)]}
    -\left(\frac{\Cov[K_{1,h_n}(x),\partial_x K_{1,h_n}(x)]}{\Var[K_{1,h_n}(x)]}\right)^2
    \right\}^{1/2}dx \nonumber\\
    &\qquad-\arccos\{\rho_n(x_j,x_{j+1})\}
    +H_{q_n(1-\alpha)}\left(\arccos\{\rho_n(x_j,x_{j+1})\}\right)
    \Bigg]\le\tau. \nonumber
\end{align}
The smallest admissible integer $p$ can then be found by evaluating the selection criterion over candidate grid sizes.
\end{rl}

This criterion retains the original discretization bound.
However, its implementation requires evaluating the covariance structure and performing numerical integration over the evaluation interval.
To reduce this computational burden, we now impose some conditions on the local empirical process and derive an asymptotic expansion of the discretization bound.
This yields a simpler grid selection rule with an explicit leading constant.

\begin{assumption}\label{as:DGP_local}
\mbox{}
\begin{itemize}
    \item[(i)]
    $\{(Y_i,X_i)\}_{i=1}^n$ is an i.i.d. sample, and $X_i$ has a continuously differentiable density $f$ on $\mathbb R$ such that
    $\sup_{x\in\mathbb R}\{|f(x)|+|f'(x)|\}<\infty$
    and $0<\inf_{x\in\mathcal X}f(x)$.
    The function $g$ is bounded.

    \item[(ii)]
    The function
    $x\mapsto f(x)\mathbb E[g(Y_i)^2\mid X_i=x]$
    is continuously differentiable with a bounded derivative on an open neighborhood of $\mathcal X$, and $\inf_{x\in\mathcal X}f(x)\mathbb E[g(Y_i)^2\mid X_i=x]>0$.

    \item[(iii)]
    The bandwidth $h_n>0$ satisfies
    $h_n\to0$ and $nh_n/\log(h_n^{-1})\to\infty$.

    \item[(iv)]
    The kernel $K:\mathbb R\to\mathbb R$ is supported on $[-1,1]$, satisfies $\int_{\mathbb R}K(u)\,du=1$, is twice continuously differentiable on $\mathbb R$, and has a Lipschitz continuous second derivative. \qedhere

\end{itemize}
\end{assumption}

The smoothness condition on the kernel in \cref{as:DGP_local}(iv) is used to derive the asymptotic expansions in following \cref{lem:local_1_disc}.
Under \cref{as:DGP_local}, the following lemma simplifies quantities in \cref{thm:upper_general} in the case of local empirical process.

\begin{lemma}\label{lem:local_1_disc}
Suppose that \cref{as:DGP_local} holds, the evaluation grid is equally spaced and $(p-1)h_n\to\infty$.
Then, the following expansions hold, where the $o(1)$ terms in (i) and (ii) are uniform over $j\in[p-1]$.
\begin{itemize}
    \item[(i)] For each $x_j$ and $x_{j+1}$, it holds that
    \begin{align*}
        d_n(x_j,x_{j+1}) &= \frac{\ol x-\ul x}{h_n(p-1)}\left\{\left(\frac{\int_{\mathbb R}\{K'(u)\}^2du}{\int_{\mathbb R}K(u)^2du}\right)^{1/2}+o(1)\right\}.
    \end{align*}

    \item[(ii)] The difference $\mathcal L_{n,j}-d_n(x_j,x_{j+1})$ satisfies
    \begin{align*}
        \mathcal L_{n,j}-d_n(x_j,x_{j+1}) &= \left(\frac{\ol x-\ul x}{h_n(p-1)}\right)^3\left\{\frac{\displaystyle\frac{\int_{\mathbb R}\{K''(u)\}^2du}{\int_{\mathbb R}K(u)^2du}-\left(\frac{\int_{\mathbb R}\{K'(u)\}^2du}{\int_{\mathbb R}K(u)^2du}\right)^2}{\displaystyle24\left(\frac{\int_{\mathbb R}\{K'(u)\}^2du}{\int_{\mathbb R}K(u)^2du}\right)^{1/2}}+o(1)\right\}.
    \end{align*}

    \item[(iii)] The quantity $\mathcal L_n$ satisfies
    \begin{align*}
        \mathcal L_n &= \frac{\ol x-\ul x}{h_n}\left\{\left(\frac{\int_{\mathbb R}\{K'(u)\}^2du}{\int_{\mathbb R}K(u)^2du}\right)^{1/2}+o(1)\right\}. \qedhere
    \end{align*}
\end{itemize}
\end{lemma}

We now combine this result with \cref{thm:upper_general} and \cref{lem:local_1} to obtain  an explicit coverage error bound for local empirical processes.

\begin{theorem}[Coverage error bound for local empirical process]
\label{thm:local_grid}
Suppose that \cref{as:DGP_local} holds and $(p-1)h_n\to\infty$.
Suppose further that $t_n/\{(p-1)h_n\}\to0$, where $t_n$ is defined in \cref{thm:upper_general}.
Then, for every fixed $\alpha\in(0,1)$,
\begin{align*}
    &\left|\mathbb P\left(\sup_{x\in\mathcal X}|\hat T_n(x)|\le\hat c^\star_{n,1-\alpha}\right)-(1-\alpha)\right|\\
    &\le\frac{(\ol x-\ul x)^3e^{-t_n^2/2}}{24\pi h_n^3(p-1)^2}\left\{
        \frac{\left(\int_{\mathbb R}K(u)^2du\right)\left(\int_{\mathbb R}\{K''(u)\}^2du\right)+(t_n^2-1)\left(\int_{\mathbb R}\{K'(u)\}^2du\right)^2}{\left(\int_{\mathbb R}K(u)^2du\right)^{3/2}\left(\int_{\mathbb R}\{K'(u)\}^2du\right)^{1/2}}+o(1+t_n^2)
    \right\}\\
    &\quad+C\left\{\frac{\log^5(2np)}{nh_n}\right\}^{1/4},
\end{align*}
where $C>0$ is a constant independent of $n$ and $p$.
\end{theorem}

\Cref{thm:local_grid} provides an explicit coverage error bound for local empirical processes under \cref{as:DGP_local}.
The leading term quantifies the discretization effect, while all remaining approximation errors are absorbed into the high-dimensional bootstrap approximation rate.
This explicit bound provides a basis for choosing the number of evaluation points.

\begin{rl}[Simpler quantitative grid selection]
\label{rmk:local_grid_quantitative}
Let $\tau\in(0,1)$ be a user-specified tolerance level for the discretization error.
Motivated by \cref{thm:local_grid}, we consider a quantitative grid selection rule that neglects the asymptotically negligible remainder.
As discussed in \cref{rmk:unknown_quantile}, we approximate $t_n$ by $q_n(1-\alpha)$.
The resulting selection criterion is
\begin{align}
    &\frac{(\ol x-\ul x)^3e^{-q_n(1-\alpha)^2/2}}{24\pi h_n^3(p-1)^2} \nonumber\\
    & \times
    \left\{
        \frac{\int_{\mathbb R}\{K''(u)\}^2du}{\left(\int_{\mathbb R}K(u)^2du\right)^{1/2}\left(\int_{\mathbb R}\{K'(u)\}^2du\right)^{1/2}}
        +\{q_n(1-\alpha)^2-1\}\left(\frac{\int_{\mathbb R}\{K'(u)\}^2du}{\int_{\mathbb R}K(u)^2du}\right)^{3/2}
    \right\}\le\tau.
    \label{eq:local_grid_quantitative}
\end{align}
The Gaussian quantile $q_n(1-\alpha)$ can be approximated using the Gaussian multiplier bootstrap.
\end{rl}

We next show that the sufficient grid condition for asymptotically valid uniform inference is sharp in terms of the growth rate of the number of grid points, by combining the following lemma with the coverage failure result in \cref{thm:local_coverage_failure}.

\begin{lemma}\label{lem:local_grid_validity}
Suppose that \cref{as:DGP_local} holds and the evaluation grid is equally spaced.
If
\begin{align*}
    \frac{ph_n}{\sqrt{\log(h_n^{-1})}}\to\infty,\qquad \frac{\log^5(2np)}{nh_n}\to0,
\end{align*}
then, for every fixed $\alpha\in(0,1)$,
\begin{align*}
    &\left|\mathbb P\left(\sup_{x\in\mathcal X}|\hat T_n(x)|\le\hat c^\star_{n,1-\alpha}\right)-(1-\alpha)\right|=O\left(\frac{\log(h_n^{-1})}{(ph_n)^2} + \left\{\frac{\log^5(2np)}{nh_n}\right\}^{1/4}\right)=o(1). \qedhere
\end{align*}
\end{lemma}

\begin{rmk}[Sharpness of the grid rate]\label{rmk:local_grid_sharpness}
Fix $\alpha\in(0,1)$ and suppose that \cref{as:DGP_local} holds and $\log^5(2np)/(nh_n)\to0$.
By \cref{thm:local_coverage_failure,lem:local_grid_validity},
\begin{align*}
    \frac{ph_n}{\sqrt{\log(h_n^{-1})}}\to0
    &\implies \mathbb P\left(\sup_{x\in\mathcal X}|\hat T_n(x)|\le\hat c^\star_{n,1-\alpha}\right)\to0,\\
    \frac{ph_n}{\sqrt{\log(h_n^{-1})}}\to\infty
    &\implies \mathbb P\left(\sup_{x\in\mathcal X}|\hat T_n(x)|\le\hat c^\star_{n,1-\alpha}\right)\to1-\alpha.
\end{align*}
Thus, the sufficient grid condition is sharp in terms of the growth rate of $p$. \qedhere
\end{rmk}

\section{Numerical Experiments} \label{sec:numerical_experiments}

In this section, we investigate finite-sample performance of the proposed grid selection rules.
In \cref{subsec:gaussian_bound_numerical}, we assess the sharpness of the Gaussian discretization bound (\cref{thm:upper_general}) by comparing the grid sizes selected by the proposed rules with those obtained from Gaussian Monte Carlo simulations.
In \cref{subsec:finite_sample_grid_numerical}, we evaluate the finite-sample performance of these rules under the same setting as \cref{subsec:numerical_local}.
The results show that the rules can be somewhat optimistic in small samples.
Motivated by this finding, in \cref{subsec:sources_small_sample_optimism}, we investigate the sources of this optimism by decomposing the discretization error.

\subsection{Numerical Experiment I : Sharpness of the Discretization Error Bound}
\label{subsec:gaussian_bound_numerical}

In practice, our goal is to choose the number of evaluation points $p$ such that the discretization error does not exceed a prescribed tolerance level $\tau$;
\begin{align*}
    \mathbb P\left(\max_{j\in[p]}|\hat T_n(x_j)|\le\hat c_{n,1-\alpha}^{\star}<\sup_{x\in\mathcal X}|\hat T_n(x)|\right) \le \tau.
\end{align*}
Since the results of \citet{CCKK22} allow us to control the Gaussian approximation error, bootstrap approximation error and studentization error, \cref{thm:upper_general} separates the Gaussian discretization bound from the remaining approximation errors.
Therefore, using the simulation designs in \cref{subsec:numerical_local}, we numerically assess the sharpness of the Gaussian discretization bound derived in \cref{thm:upper_general}.
For each tolerance level $\tau\in\{1\%,0.5\%,0.25\%\}$, we first numerically calculate the smallest number of equally spaced evaluation points $p$ such that
\begin{align}
    \mathbb P\left(\max_{j\in[p]}|G_n(x_j)|\le q_n(1-\alpha)<\sup_{x\in\mathcal X}|G_n(x)|\right)\le\tau. \label{eq:gauss_benchmark}
\end{align}
We then compare this benchmark with the smallest grid sizes selected by our proposed grid selection rules in \cref{rmk:local_grid_unexpanded} and \cref{rmk:local_grid_quantitative}, respectively, using the same tolerance level $\tau$.

We use the two designs in \cref{subsec:numerical_local}, with the same evaluation interval $\mathcal X=[-3,3]$ and the same undersmoothing rate $h_n=c_hn^{-2/7}$.
Unlike the preceding experiment, we use the triweight kernel $K(u)=\frac{35}{32}(1-u^2)^3\mathbf{1}\{|u|\le1\}$ to satisfy the twice-continuous-differentiability condition imposed in \cref{as:DGP_local}, which the Epanechnikov kernel does not satisfy.
The bandwidth constant $c_h$ is recalculated using the AMISE criterion for the triweight kernel.
We set $\alpha=0.05$, consider $n\in\{250,500,1000,2000,5000\}$, and search over equally spaced grids with $p\in\{3,\ldots,150\}$.

To simulate $\max_{j\in[p]}|G_n(x_j)|$, $\sup_{x\in\mathcal X}|G_n(x)|$, and the quantile $q_n(1-\alpha)$, we need to compute the covariance function of $G_n$.
To this end, we define the first order approximation of $\hat{T}_n(x)$, which is the centered KDE standardized by its population standard deviation:
\begin{align*}
    T_n^\circ(x)
    &\coloneqq \frac{\hat f_{h_n}(x)-\mathbb E[\hat f_{h_n}(x)]}{\{\operatorname{Var}(\hat f_{h_n}(x))\}^{1/2}}=\frac{\sum_{i=1}^n\{\psi_{h_n}(X_i,x)-\mathbb E[\psi_{h_n}(X_i,x)]\}}{\{n\operatorname{Var}(\psi_{h_n}(X_i,x))\}^{1/2}}.
\end{align*}
We define $G_n$ as the centered Gaussian process with the same covariance function as $T_n^\circ$.
This covariance function is given by
\begin{align*}
    \rho_n(x,y)
    &=\frac{\operatorname{Cov}(\psi_{h_n}(X_i,x),\psi_{h_n}(X_i,y))}{\{\operatorname{Var}(\psi_{h_n}(X_i,x))\operatorname{Var}(\psi_{h_n}(X_i,y))\}^{1/2}}.
\end{align*}
We evaluate the expectations in the covariance function using the known density $f$ in each design.
To compute $q_n(1-\alpha)$, for each candidate grid with $p$ evaluation points, we generate $200,000$ independent realizations of the centered Gaussian vector $(G_n(x_j))_{j\in[p]}$ with covariance matrix $(\rho_n(x_j,x_k))_{j,k\in[p]}$.
We estimate $q_n(1-\alpha)$ by the empirical $(1-\alpha)$-quantile of the resulting maxima $\max_{j\in[p]}|G_n(x_j)|$.
To estimate the benchmark in \eqref{eq:gauss_benchmark}, we generate $1,000,000$ independent realizations of $G_n$, independently of the draws used to estimate $q_n(1-\alpha)$, and for each realization, we compute $\max_{j\in[p]}|G_n(x_j)|$ and approximate $\sup_{x\in\mathcal X}|G_n(x)|$ by the maximum over $3,201$ equally spaced points.
We then estimate the probability in \eqref{eq:gauss_benchmark} by the proportion of realizations satisfying $\max_{j\in[p]}|G_n(x_j)|\le q_n(1-\alpha)<\sup_{x\in\mathcal X}|G_n(x)|$.

For each tolerance level $\tau$, we determine the smallest candidate grid size $p$ for which the Gaussian Monte Carlo estimate of the probability in \eqref{eq:gauss_benchmark} does not exceed $\tau$.
We similarly determine the smallest grid sizes satisfying the selection criteria in \cref{rmk:local_grid_unexpanded} and \cref{rmk:local_grid_quantitative}.
\Cref{tab:gaussian_grid_comparison} reports these grid sizes together with the differences between the sizes selected by each rule and the Gaussian Monte Carlo benchmark, allowing us to assess the sharpness of the proposed bounds.

\begin{table}[!tbp]
\centering
\fontsize{9}{10.5}\selectfont
\renewcommand{\arraystretch}{1.05}
\setlength{\tabcolsep}{4pt}
\captionsetup{font=footnotesize,skip=3pt}

\caption{Gaussian Monte Carlo benchmark and grid sizes selected by the proposed rules.}
\label{tab:gaussian_grid_comparison}

\begin{tabular}{@{}ccrrrrrr@{}}
\toprule
& & \multicolumn{2}{c}{Gaussian Monte Carlo}
& \multicolumn{2}{c}{\cref{rmk:local_grid_unexpanded}}
& \multicolumn{2}{c}{\cref{rmk:local_grid_quantitative}} \\
\cmidrule(lr){3-4}\cmidrule(lr){5-6}\cmidrule(lr){7-8}
$\tau$ & $n$
& Normal & Mixture normal
& Normal & Mixture normal
& Normal & Mixture normal \\
\midrule
$1\%$ & 250  & $20\,(0.9888)$ & $25\,(0.7822)$ & $21\,(+1)$ & $26\,(+1)$ & $20\,(0)$ & $25\,(0)$ \\
       & 500  & $24\,(0.8111)$ & $30\,(0.6417)$ & $25\,(+1)$ & $31\,(+1)$ & $24\,(0)$ & $31\,(+1)$ \\
       & 1000 & $29\,(0.6654)$ & $37\,(0.5264)$ & $30\,(+1)$ & $38\,(+1)$ & $30\,(+1)$ & $39\,(+2)$ \\
       & 2000 & $36\,(0.5459)$ & $45\,(0.4318)$ & $37\,(+1)$ & $46\,(+1)$ & $37\,(+1)$ & $48\,(+3)$ \\
       & 5000 & $46\,(0.4201)$ & $59\,(0.3324)$ & $48\,(+2)$ & $61\,(+2)$ & $49\,(+3)$ & $63\,(+4)$ \\
\midrule
$0.5\%$ & 250  & $28\,(0.9888)$ & $34\,(0.7822)$ & $29\,(+1)$ & $36\,(+2)$ & $27\,(-1)$ & $34\,(0)$ \\
         & 500  & $34\,(0.8111)$ & $43\,(0.6417)$ & $35\,(+1)$ & $44\,(+1)$ & $33\,(-1)$ & $42\,(-1)$ \\
         & 1000 & $41\,(0.6654)$ & $51\,(0.5264)$ & $42\,(+1)$ & $53\,(+2)$ & $41\,(0)$ & $53\,(+2)$ \\
         & 2000 & $50\,(0.5459)$ & $63\,(0.4318)$ & $51\,(+1)$ & $65\,(+2)$ & $50\,(0)$ & $65\,(+2)$ \\
         & 5000 & $65\,(0.4201)$ & $82\,(0.3324)$ & $67\,(+2)$ & $85\,(+3)$ & $66\,(+1)$ & $86\,(+4)$ \\
\midrule
$0.25\%$ & 250  & $39\,(0.9888)$ & $48\,(0.7822)$ & $40\,(+1)$ & $50\,(+2)$ & $36\,(-3)$ & $47\,(-1)$ \\
          & 500  & $46\,(0.8111)$ & $60\,(0.6417)$ & $48\,(+2)$ & $61\,(+1)$ & $45\,(-1)$ & $58\,(-2)$ \\
          & 1000 & $57\,(0.6654)$ & $72\,(0.5264)$ & $59\,(+2)$ & $75\,(+3)$ & $56\,(-1)$ & $72\,(0)$ \\
          & 2000 & $70\,(0.5459)$ & $89\,(0.4318)$ & $72\,(+2)$ & $92\,(+3)$ & $69\,(-1)$ & $89\,(0)$ \\
          & 5000 & $91\,(0.4201)$ & $117\,(0.3324)$ & $94\,(+3)$ & $121\,(+4)$ & $92\,(+1)$ & $119\,(+2)$ \\
\bottomrule
\end{tabular}

\vspace{4pt}
\begin{minipage}{0.98\linewidth}
\footnotesize
\textit{Notes.}
Entries report the minimum grid sizes attaining each discretization tolerance.
For the Gaussian Monte Carlo benchmark, bandwidths are reported in parentheses.
For the two selection rules, parentheses report differences from the Gaussian Monte Carlo benchmark.
\end{minipage}
\end{table}

\Cref{tab:gaussian_grid_comparison} demonstrate the sharpness of the proposed grid selection rules.
In particular, \cref{rmk:local_grid_unexpanded} selects between one and four additional evaluation points relative to the Gaussian Monte Carlo benchmark across all $30$ settings, with an average difference of only $1.70$ points.
 \cref{rmk:local_grid_quantitative} is less conservative, selecting exactly the same number of evaluation points as the benchmark in $8$ settings, fewer points in $9$ settings, and more points in $13$ settings.
Its differences from the benchmark range from $-3$ to $+4$ points, with an average absolute difference of only $1.30$ points.

\subsection{Numerical Experiment II: Finite-Sample Performance of Selection Rules}
\label{subsec:finite_sample_grid_numerical}

We next assess the performance of the proposed rules in practice.
Specifically, we estimate
\begin{align*}
    \mathbb P\left(\max_{j\in[p]}|\hat T_n(x_j)|\le\hat c_{n,1-\alpha}^{\star}<\sup_{x\in\mathcal X}|\hat T_n(x)|\right)
\end{align*}
and assess how closely the grid sizes selected by \cref{rmk:local_grid_unexpanded} and \cref{rmk:local_grid_quantitative} match the minimum grid sizes required to attain the prescribed tolerance $\tau$.

We use the same data-generating processes, kernel, bandwidths, and evaluation interval as in \cref{subsec:gaussian_bound_numerical}.
For each sample size $n\in\{250,500,1000,2000,5000\}$, we consider equally spaced evaluation grids with $p\in\{3,\ldots,130\}$ points, including both endpoints of $\mathcal X$.
We approximate the supremum over $\mathcal X$ by the maximum over the union of $3{,}201$ equally spaced points and all evaluation-grid points.
For each data set, we compute the conditional $(1-\alpha)$-quantile $\hat c_{n,1-\alpha}^{\star}$ of $\max_{j\in[p]}|\hat T_n^*(x_j)|$ using $499$ multiplier draws.
We use $20{,}000$ Monte Carlo replications for each data-generating process and sample size.
Within each replication, the same observations and bootstrap multipliers are used for all candidate grid sizes $p$.

\Cref{fig:discretization_error_kde} illustrates the favorable large-sample performance of the proposed grid selection rules.
For $n\ge1000$, the grid sizes selected by both rules are generally close to the minimum grid sizes required to attain the prescribed tolerances in both designs.
These results demonstrate the practical usefulness of the proposed grid selection rules.

However, for smaller sample sizes ($n=250$ and $500$), the proposed rules tend to be somewhat optimistic.
The Monte Carlo estimates of the discretization error at the selected grid sizes sometimes exceed the prescribed tolerances, although this discrepancy becomes less pronounced as the sample size increases.
We investigate the sources of this small-sample optimism in the next subsection.

\begin{figure}[htbp!]
\centering
\includegraphics[width=\textwidth]{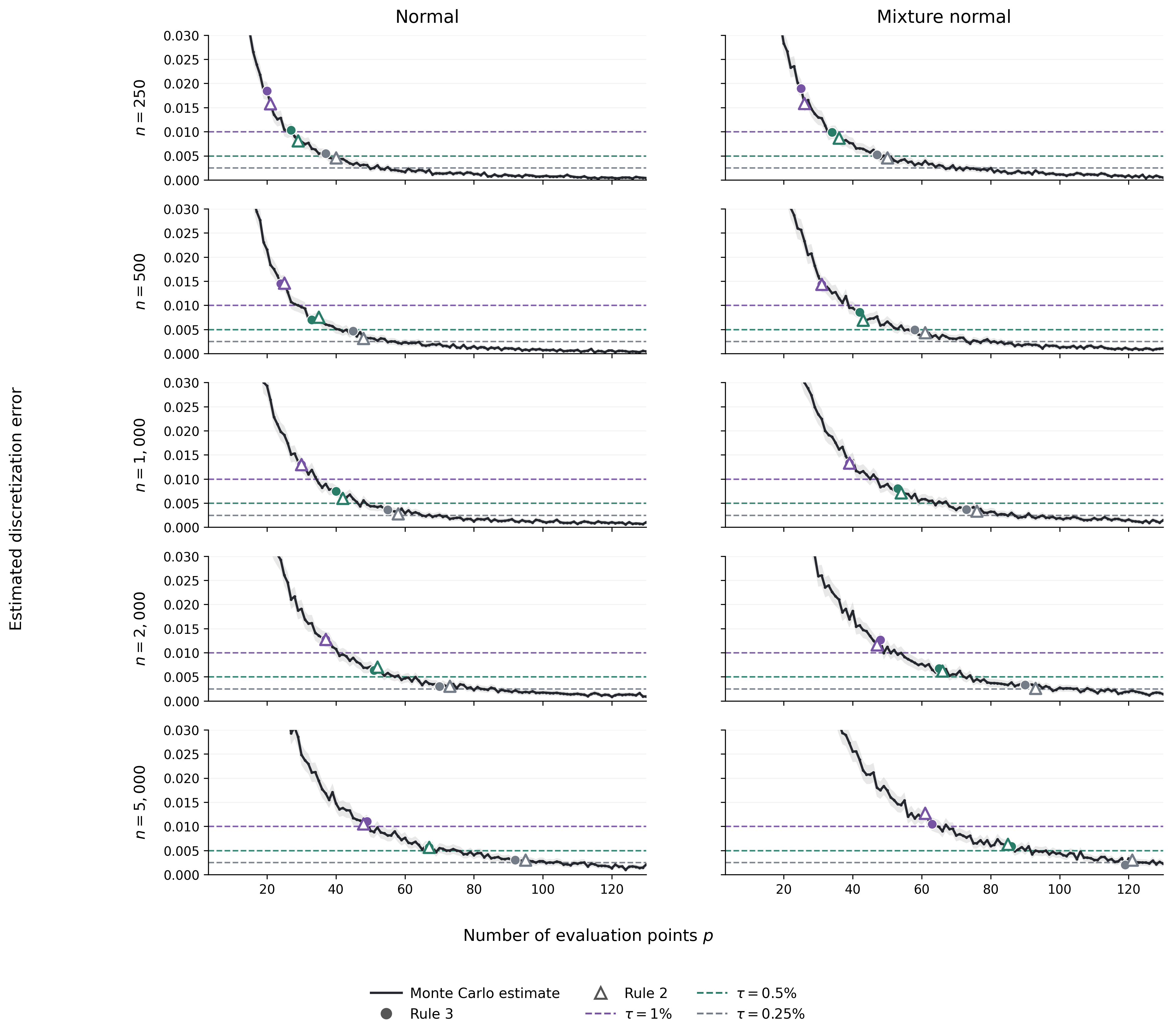}
\caption{
    Rows correspond to sample sizes and columns to data-generating processes.
    The black solid lines report the estimated discretization errors as functions of the number of evaluation points $p$.
    The dashed horizontal lines indicate the prescribed tolerance levels $\tau\in\{1\%,0.5\%,0.25\%\}$.
    The open triangles and filled circles indicate the grid sizes selected by the rules in \cref{rmk:local_grid_unexpanded} and \cref{rmk:local_grid_quantitative}, respectively, with the markers placed at the corresponding Monte Carlo estimates.
    The vertical axes are restricted to $[0,0.03]$.
}
\label{fig:discretization_error_kde}
\end{figure}

\subsection{Numerical Experiments III: Sources of Small-Sample Optimism} \label{subsec:sources_small_sample_optimism}
In the previous section, we separated the discretization error of the Gaussian comparator from the remaining approximation errors arising from Gaussian approximation, the multiplier bootstrap, and studentization, and used this separation to derive a quantitative grid selection rule.
The resulting rule is expected to be nearly sharp when the approximation errors are negligible, as suggested by our numerical results in \cref{subsec:gaussian_bound_numerical}.
However, in finite samples, particularly when $n=250$ or $500$, the rule tends to select small grid points to achieve the prescribed discretization tolerance as in \cref{fig:discretization_error_kde}.

To see why, we first decompose the discretization error into four components. 
For a deterministic  threshold $t_n$, we have
\begin{equation*}
\begin{aligned}
&\mathbb P\left(\max_{1\le j\le p}|\hat T_n(x_j)|\le\hat c^\star_{n,1-\alpha}<\sup_{x\in\mathcal X}|\hat T_n(x)|\right)=\underbrace{\mathbb P\left(\max_{1\le j\le p}|G_n(x_j)|\le t_n<\sup_{x\in\mathcal X}|G_n(x)|\right)}_{\text{Gaussian Discretization Error}}\\
&\quad+\underbrace{\left\{\mathbb P\left(\max_{1\le j\le p}|T_n(x_j)|\le t_n<\sup_{x\in\mathcal X}|T_n(x)|\right)-\mathbb P\left(\max_{1\le j\le p}|G_n(x_j)|\le t_n<\sup_{x\in\mathcal X}|G_n(x)|\right)\right\}}_{\text{Gaussian Approximation Error}}\\
&\quad+\underbrace{\left\{\mathbb P\left(\max_{1\le j\le p}|\hat T_n(x_j)|\le t_n<\sup_{x\in\mathcal X}|\hat T_n(x)|\right)-\mathbb P\left(\max_{1\le j\le p}|T_n(x_j)|\le t_n<\sup_{x\in\mathcal X}|T_n(x)|\right)\right\}}_{\text{Studentization Error}}\\
&\quad+\underbrace{\left\{\mathbb P\left(\max_{1\le j\le p}|\hat T_n(x_j)|\le\hat c^\star_{n,1-\alpha}<\sup_{x\in\mathcal X}|\hat T_n(x)|\right)-\mathbb P\left(\max_{1\le j\le p}|\hat T_n(x_j)|\le t_n<\sup_{x\in\mathcal X}|\hat T_n(x)|\right)\right\}}_{\text{Bootstrap Critical Value Error}} \\
& \eqqcolon \Delta_{\text{GD}} + \Delta_{\text{GA}} + \Delta_{\text{S}} + \Delta_{\text{BC}}.
\end{aligned}
\end{equation*}
In the previous section, we derived a sharp bound for the first component $\Delta_{\text{GD}}$ and bounded the remaining three components $\Delta_{\text{GA}} + \Delta_{\text{S}} + \Delta_{\text{BC}}$ collectively by $C\{B_n^2\log^5(2np)/n\}^{1/4}$ using the high-dimensional approximation results of \citet{CCKK22}.
Since this bound vanishes asymptotically under our assumptions, we omitted it from the quantitative grid selection rule.
However, the omitted approximation errors may not be negligible in finite samples and may account for the small-sample optimism observed in the preceding experiment.
We therefore examine these components separately.

\begin{figure}[htbp!]
\centering
\includegraphics[width=\textwidth]{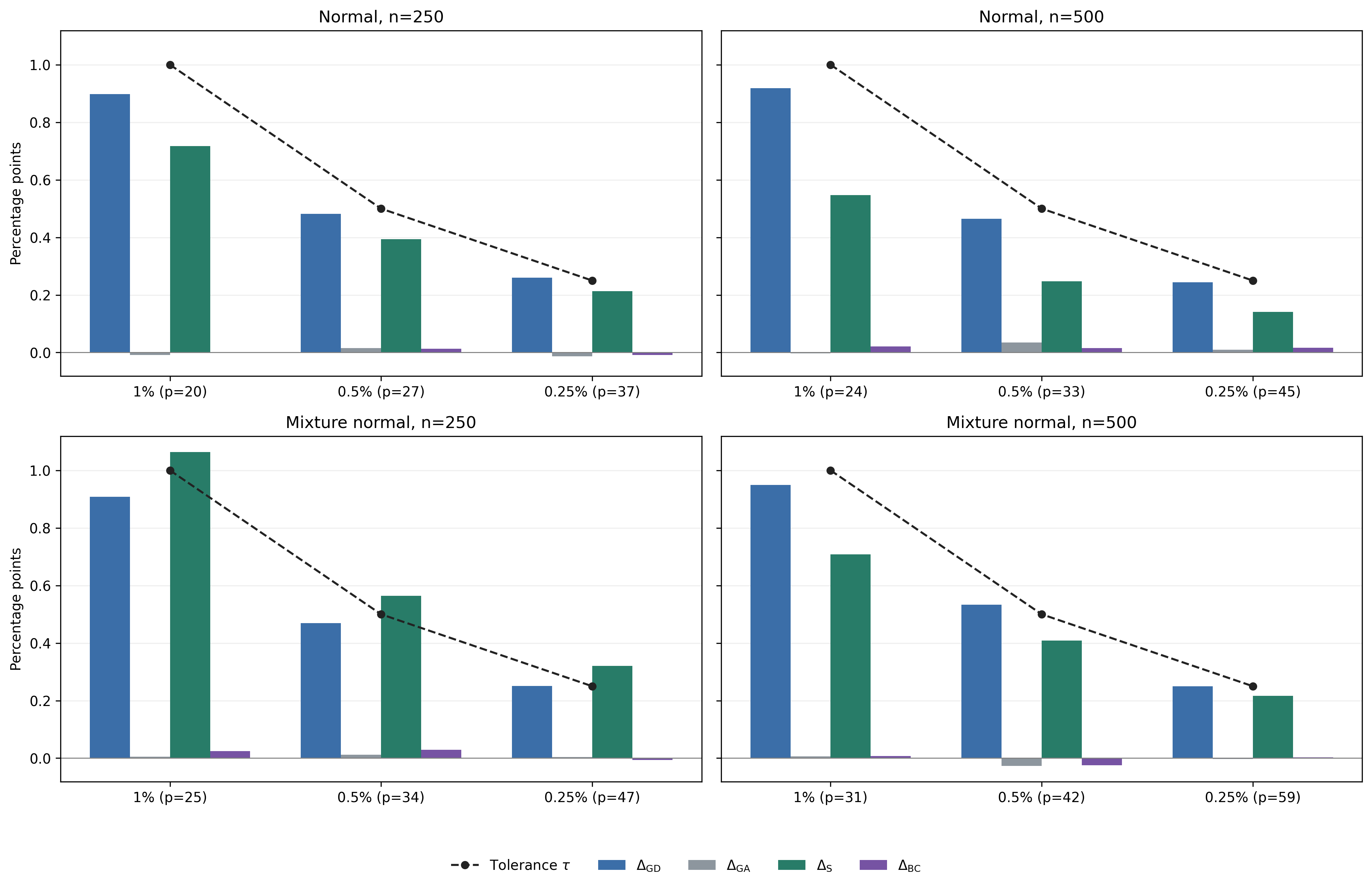}
\caption{The four components are reported in percentage points for the normal and mixture normal designs with $n=250$ and $500$. 
For each tolerance $\tau\in\{0.01,0.005,0.0025\}$, the grid size $p$ is selected using the quantitative grid selection rule in \cref{rmk:local_grid_quantitative}. 
The black dashed line indicates the prescribed tolerance $\tau$. 
The estimates are based on 300,000 Monte Carlo replications and 2,999 multiplier bootstrap draws.}
\label{fig:small_sample_optimism}
\end{figure}

\Cref{fig:small_sample_optimism} reports the four components of the missed-crossing probability for the KDE designs. The Gaussian discretization error $\Delta_{\mathrm{GD}}$ is close to the prescribed tolerance $\tau$ across all settings, confirming the accuracy of the Gaussian-based grid selection rule. 
The Gaussian approximation error $\Delta_{\mathrm{GA}}$ and the bootstrap critical value error $\Delta_{\mathrm{BC}}$ are relatively small and exhibit no systematic sign. 
In contrast, the studentization error $\Delta_{\mathrm S}$ is positive and substantial in all settings. 
In particular, $\Delta_{\mathrm S}$ accounts for most of the discrepancy between the Gaussian discretization error and the actual missed-crossing probability.

\section{Small-Sample Adjustment}
\label{subsec:student_t_grid_selection}

The numerical results in \cref{subsec:sources_small_sample_optimism} suggest that accounting for studentization may improve the accuracy of grid selection in small samples.
Motivated by this observation, we consider an approximating stochastic process that retains the Gaussian approximation for the numerator while incorporating the variability of the estimated variance.
To construct such a process, we approximate the distribution of the estimated variance by a scaled chi-square distribution, following the insights from \citet{Welch1951,BellMcCaffrey2002,ImbensKolesar2016}.
This leads to a Student-$t$ approximating process, for which we derive a discretization error bound and a corresponding quantitative grid selection rule.

Define the relative variance estimator and its moments by
\begin{align*}
    \hat Q_n(x)&:=\frac{\hat\sigma_n^2(x)}{\sigma_n^2(x)}, ~~
    \mu_n:=\mathbb E[\hat Q_n(x)] = \frac{n-1}{n}, ~~
    v_n(x)=\Var[\hat Q_n(x)],
\end{align*}
where $\mu_n = (n-1)/n$ holds, from \cref{eq:variance_estimator} and the i.i.d. assumption.
Following the standard chi-square approximation for ratio of variance and variance estimators used in small-sample inference (see \citealp[Section II.E, pp.~703--704]{ImbensKolesar2016} and references therein), we approximate the distribution of $\hat Q_n(x)$ by a scaled chi-square distribution whose first two moments agree with those of $\hat Q_n(x)$.
\begin{align}
    \hat Q_n(x)\ \dot\sim\
    \mu_n\frac{\chi^2_{\nu_n(x)}}{\nu_n(x)},
    \qquad
    \nu_n(x):=\frac{2\mu_n^2}{v_n(x)}.
    \label{eq:student_t_moment_matching}
\end{align}
This approximation matches the expectation and variance of $\hat Q_n(x)$.
Since $\hat T_n(x)=T_n(x)/\sqrt{\hat Q_n(x)}$, approximating $T_n(x)$ by a standard Gaussian variable independent of the chi-square approximation to $\hat Q_n(x)$ yields
\begin{align}
    \hat T_n(x)\ \dot\sim\ \mu_n^{-1/2}t_{\nu_n(x)},
    \label{eq:student_t_marginal_approximation}
\end{align}
where $t_{\nu_n(x)}$ denotes a Student-$t$ random variable with $\nu_n(x)$ degrees of freedom.

The marginal approximation in \eqref{eq:student_t_marginal_approximation} does not specify the dependence across $\mathcal X$, which is essential for evaluating discretization error.
We therefore construct a Student-$t$ approximating process that retains the covariance structure of the Gaussian comparator $G_n$ in its Gaussian components while incorporating the variability of the  variance estimation through a random denominator.
Let $G_{n,0},G_{n,1},\ldots,G_{n,\nu}$ be independent copies of the Gaussian comparator $G_n$.
For an integer $\nu\ge3$ and $\mu>0$, define
\begin{align}
    S_{n,\nu,\mu}(x)
    :=\frac{1}{\mu^{1/2}}\frac{G_{n,0}(x)}{\left\{\frac{1}{\nu}\sum_{k=1}^{\nu}G_{n,k}(x)^2\right\}^{1/2}},
    \qquad x\in\mathcal X.
    \label{eq:student_t_reference_process}
\end{align}
The marginal distribution of $S_{n,\nu,\mu}(x)$ is $\mu^{-1/2}t_\nu$.

The next proposition provides a discretization error bound for the approximating process in \eqref{eq:student_t_reference_process}.

\begin{proposition}
\label{prop:student_t_discretization_bound}
For $t\ge0$, $\theta\in[0,\pi]$, $\mu>0$, and $\nu>2$, define
\begin{align*}
    H_{t,\nu,\mu}(\theta)
    :=\int_0^\theta
    \left[1-\left\{1+\frac{\mu t^2(\nu-1)}{\nu(\nu-2)}
    \tan^2\left(\frac{s}{2}\right)\right\}^{-(\nu+1)/2}\right]ds.
\end{align*}
Suppose that \cref{ass:upper_general} holds.
Then, for every deterministic evaluation grid, $t\ge0$, $\mu>0$, and integer $\nu\ge3$,
\begin{align}
    &\mathbb P\left(\max_{1\le j\le p}|S_{n,\nu,\mu}(x_j)|\le t<\sup_{x\in\mathcal X}|S_{n,\nu,\mu}(x)|\right)\nonumber\\
    &\quad\le
    \frac{\left(1+\mu t^2/\nu\right)^{-(\nu-1)/2}}{\pi}
    \sum_{j=1}^{p-1}\left\{\mathcal L_{n,j}-d_n(x_j,x_{j+1})
    +H_{t,\nu,\mu}\bigl(d_n(x_j,x_{j+1})\bigr)\right\}. \qedhere
\end{align}
\end{proposition}

The proof is given in the Supplementary Material.
As $\mu\to1$ and $\nu\to\infty$, the bound in \cref{prop:student_t_discretization_bound} reduces to the Gaussian discretization bound in \cref{thm:upper_general} for each fixed $t\ge0$.
This implies that the grid selection rules based on \cref{prop:student_t_discretization_bound} and  \cref{thm:upper_general} are asymptotically equivalent when the variance estimator is consistent and its variance vanishes.

For local empirical processes, the same approximation used in \cref{thm:local_grid} gives a simple counterpart to \cref{rmk:local_grid_quantitative}. 

\begin{rl}[Student-$t$ grid selection rule for local empirical process]
\label{rmk:student_t_grid_selection}
Suppose that \cref{as:DGP_local} holds and the evaluation grid is equally spaced. Given working parameters $\mu>0$ and $\nu>2$, we propose selecting the smallest admissible integer $p$ satisfying
\begin{align}
    &\frac{(\ol x-\ul x)^3}{24\pi h_n^3(p-1)^2}
    \left(1+\frac{\mu q_n(1-\alpha)^2}{\nu}\right)^{-(\nu-1)/2}\Bigg\{
    \frac{\int_{\mathbb R}\{K''(u)\}^2du}
    {\left(\int_{\mathbb R}K(u)^2du\right)^{1/2}
     \left(\int_{\mathbb R}\{K'(u)\}^2du\right)^{1/2}} \nonumber\\
     & \qquad\qquad\qquad\qquad\qquad\qquad+\left\{\frac{\mu q_n(1-\alpha)^2(\nu^2-1)}{\nu(\nu-2)}-1\right\}
    \left(\frac{\int_{\mathbb R}\{K'(u)\}^2du}
    {\int_{\mathbb R}K(u)^2du}\right)^{3/2}\Bigg\}\le\tau.
    \label{eq:student_t_grid_selection}
\end{align}
The criterion reduces to \cref{rmk:local_grid_quantitative} as $\mu\to1$ and $\nu\to\infty$.
\end{rl}

In the numerical implementation, we set $\mu=\mu_n$ and  $\nu=\max\{\min_{j\in[p]}\nu_n(x_j), 15\}$.
The lower bound of $15$ is imposed to avoid numerical instability arising from the factor $\{\nu(\nu-2)\}^{-1}$ in the grid selection criterion when $\nu$ is small.
We use the minimum effective degrees of freedom over the evaluation grid, which corresponds to the maximum variance of the relative variance estimator $\hat Q_n(x)$ over the grid.
Although \cref{prop:student_t_discretization_bound} is stated for integer degrees of freedom, the grid selection rule in \eqref{eq:student_t_grid_selection} can be evaluated for any real $\nu > 2$. 
We therefore apply the rule without rounding $\nu$ when it is not an integer.

\subsection{Numerical Experiment IV : Performance of Small-Sample Adjustment} \label{subsec:experiment_student_t}

We next assess whether the Student-$t$ grid selection rule (\cref{rmk:student_t_grid_selection}) improves the small-sample performance of Gaussian based rule (\cref{rmk:local_grid_quantitative}).
Specifically, we compare the grid sizes selected by the two rules and evaluate the resulting discretization errors for confidence bands for KDE under the same simulation design and computational settings as in \cref{subsec:finite_sample_grid_numerical}.

\Cref{fig:student_t_discretization_error_kde,tab:student_t_grid_experiment}
show that \cref{rmk:student_t_grid_selection} substantially mitigates the small-sample optimism of \cref{rmk:local_grid_quantitative}.
The estimated discretization error does not exceed the prescribed tolerance level in $24$ of the $30$ settings.
Even in the remaining six settings, the excess over the prescribed tolerance level is small, at most $0.05$ percentage points.
Also, the estimated discretization errors at the grid sizes selected by \cref{rmk:student_t_grid_selection} are generally close to the prescribed tolerance levels, suggesting that \cref{rmk:student_t_grid_selection} does not introduce excessive conservativeness.

The effective degrees of freedom $\nu$ are at least $19.6$ across all settings, so the lower bound of $15$ is never binding.
They increase with the sample size, indicating that the Student-$t$ approximation approaches the Gaussian approximation.
This pattern is consistent with the asymptotic equivalence of the two grid selection rules.
Indeed, the difference between the numbers of evaluation points selected by the two rules decreases as the sample size increases.

\begin{figure}[htbp!]
\centering
\includegraphics[width=\textwidth]{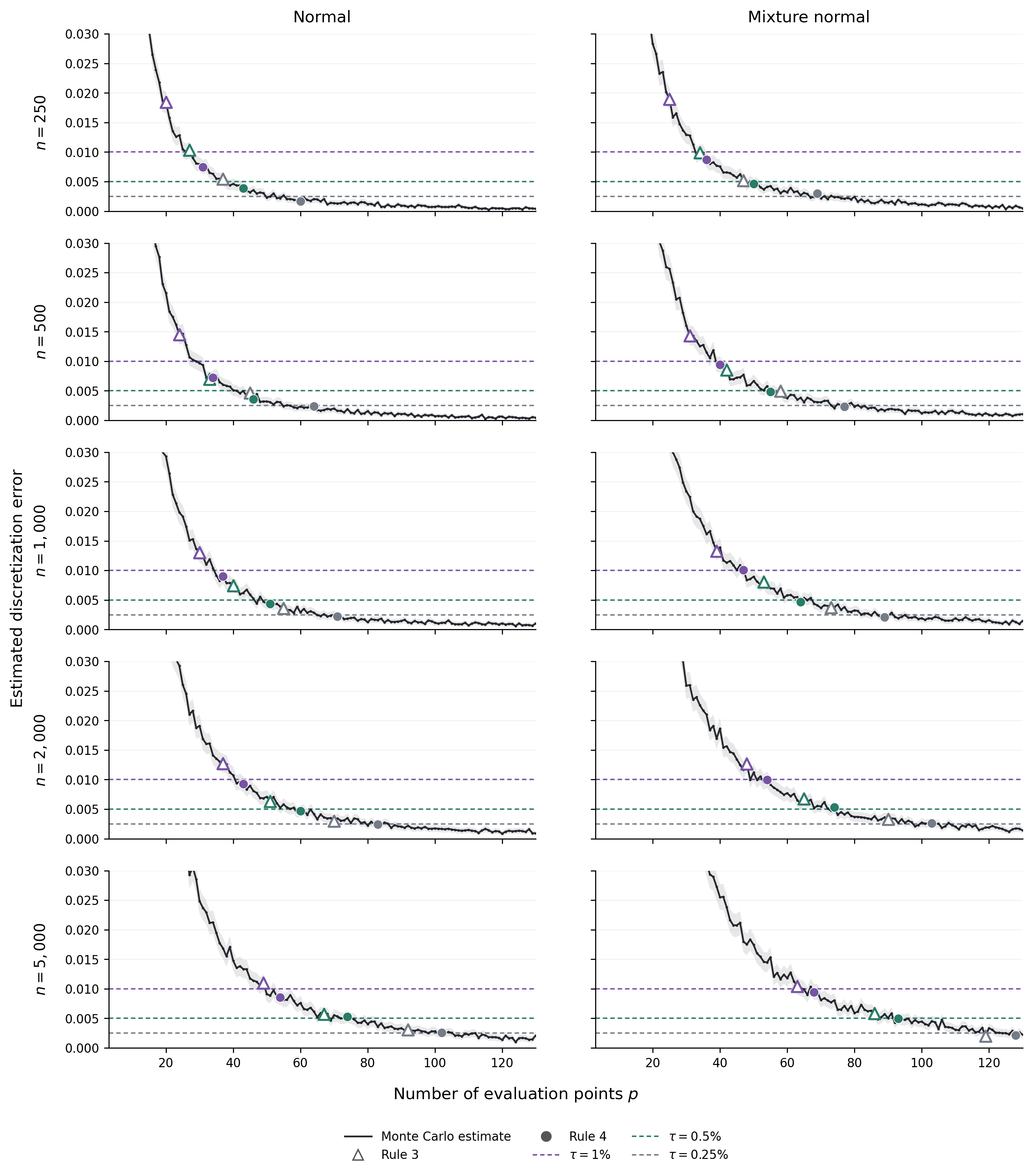}
\caption{
    Rows correspond to sample sizes and columns to data-generating processes.
    The black solid lines report the estimated discretization errors as functions of the number of evaluation points $p$.
    The dashed horizontal lines indicate the prescribed tolerance levels $\tau\in\{1\%,0.5\%,0.25\%\}$.
    The open triangles and filled circles indicate the grid sizes selected by the rules in \cref{rmk:local_grid_quantitative} and \cref{rmk:student_t_grid_selection}, respectively, with the markers placed at the corresponding Monte Carlo estimates.
    The vertical axes are restricted to $[0,0.03]$.
}
\label{fig:student_t_discretization_error_kde}
\end{figure}

\begin{table}[!tbp]
\centering
\fontsize{8}{9.5}\selectfont
\renewcommand{\arraystretch}{1.10}
\setlength{\tabcolsep}{3pt}
\caption{Grid sizes selected by \cref{rmk:local_grid_quantitative} and  \cref{rmk:student_t_grid_selection} and the associated estimated discretization errors. Each entry reports the grid size selected by \cref{rmk:local_grid_quantitative} and its estimated discretization error (in percentage points) and those by  \cref{rmk:student_t_grid_selection}, separated by an arrow.}
\label{tab:student_t_grid_experiment}
\begin{tabular}{@{}llrccc@{}}
\toprule
DGP & $n$ & $\nu_{n,p}^{\mathrm{eff}}$ & $\tau=1\%$ & $\tau=0.5\%$ & $\tau=0.25\%$ \\
\midrule
Normal & 250 & 19.6 & $20\,(1.845)\to 31\,(0.745)$ & $27\,(1.035)\to 43\,(0.390)$ & $37\,(0.545)\to 60\,(0.165)$ \\
 & 500 & 31.0 & $24\,(1.450)\to 34\,(0.720)$ & $33\,(0.700)\to 46\,(0.360)$ & $45\,(0.465)\to 64\,(0.240)$ \\
 & 1000 & 49.3 & $30\,(1.300)\to 37\,(0.900)$ & $40\,(0.745)\to 51\,(0.435)$ & $55\,(0.360)\to 71\,(0.220)$ \\
 & 2000 & 79.2 & $37\,(1.275)\to 43\,(0.925)$ & $51\,(0.635)\to 60\,(0.470)$ & $70\,(0.300)\to 83\,(0.245)$ \\
 & 5000 & 149.3 & $49\,(1.100)\to 54\,(0.855)$ & $67\,(0.570)\to 74\,(0.530)$ & $92\,(0.305)\to 102\,(0.255)$ \\
\addlinespace
Mixture normal & 250 & 28.4 & $25\,(1.895)\to 36\,(0.870)$ & $34\,(0.990)\to 50\,(0.460)$ & $47\,(0.520)\to 69\,(0.300)$ \\
 & 500 & 44.9 & $31\,(1.435)\to 40\,(0.945)$ & $42\,(0.855)\to 55\,(0.485)$ & $58\,(0.495)\to 77\,(0.230)$ \\
 & 1000 & 71.4 & $39\,(1.330)\to 47\,(1.005)$ & $53\,(0.805)\to 64\,(0.465)$ & $73\,(0.370)\to 89\,(0.210)$ \\
 & 2000 & 114.5 & $48\,(1.265)\to 54\,(0.995)$ & $65\,(0.675)\to 74\,(0.530)$ & $90\,(0.330)\to 103\,(0.260)$ \\
 & 5000 & 215.4 & $63\,(1.045)\to 68\,(0.945)$ & $86\,(0.585)\to 93\,(0.495)$ & $119\,(0.200)\to 128\,(0.215)$ \\
\addlinespace
\bottomrule
\end{tabular}
\begin{minipage}{0.98\linewidth}
\footnotesize
\emph{Notes.} The effective degrees of freedom are $\nu_{n,p}^{\mathrm{eff}}=\max\{15,\min_{1\le j\le p}\nu_n(x_j)\}$. In these designs, $\nu_{n,p}^{\mathrm{eff}}$ is the same across the three tolerances, and the lower bound of $15$ is never binding. The estimated discretization errors are based on $20{,}000$ Monte Carlo replications.
\end{minipage}
\end{table}

\section{Concluding Remarks}\label{sec:conclusion}

This paper examines the accuracy of uniform inference when supremum over a continuum is replaced by  maximum over a finite evaluation grid.

We first conduct numerical simulations for the kernel density estimator under the same data generating processes as those in \cite{chernozhukov2014anti}.
The simulations, surprisingly, show that increasing the sample size can worsen coverage on a fixed coarse grid, even when the bootstrap approximation error on the grid decreases.

We then derive a upper bound on the coverage probability of a uniform confidence band constructed using a finite-grid bootstrap critical value.
The bound separates the error caused by discretization from the bootstrap approximation error on the grid.
For local empirical processes, we give primitive bound on the coverage probability.
The resulting bound shows that coverage can converge to zero even when the number of evaluation points diverges (as long as $ph_n/\sqrt{\log(h_n^{-1})}\to 0$) and the bootstrap approximation error on the grid vanishes.

We next derive a upper bound on coverage error, defined as the absolute difference between the coverage probability and the nominal level $1-\alpha$.
The bound also separates the error caused by discretization from the bootstrap approximation error on the grid.
Also, for local empirical processes, we derive an explicit coverage error bound with a leading constant and show that asymptotically nominal coverage is attained for grid points satisfying $ph_n/\sqrt{\log(h_n^{-1})}\to\infty$. 
Combined with the coverage failure result, this establishes that the sufficient grid condition is sharp in terms of the growth rate of $p$.

Beyond this asymptotic characterization, we develop quantitative grid selection rules for a user-specified discretization tolerance level. 
Monte Carlo experiments show that the proposed rules select grid sizes very close to those required by the Monte Carlo benchmark.
We further evaluate the finite-sample performance of the proposed rules through simulations of uniform confidence bands based on studentized processes and multiplier-bootstrap critical values.
The results indicate that the proposed rules select nearly optimal grid sizes for larger sample sizes, although they can be somewhat optimistic in small samples. 
Additional numerical investigations suggest that studentization error may be an important source of this small-sample optimism.
Motivated by this finding, we derive a discretization error bound for a Student-$t$ approximating process and adjustment to the Gaussian process based grid selection rule.
Numerical experiments show that this adjustment substantially mitigates the small-sample optimism without introducing excessive conservativeness.

The importance of discretization extends beyond uniform inference for nonparametric estimators.
Related issues may arise in inference involving continuously indexed tuning parameters and in identification problems that rely on information over a continuum, such as the identification of conditional average treatment effects over continuous covariates in causal inference.
Although the statistical consequences of discretization differ across these settings, understanding how a finite representation of a continuous object affects inference or identification is a broadly relevant problem.
The approach developed in this paper may provide transferable ideas and insights into such problems.

\newpage
\appendix
\setcounter{section}{0}
\renewcommand{\thesection}{S.\arabic{section}}

\setcounter{theorem}{0}
\renewcommand{\thetheorem}{S\arabic{theorem}}

\setcounter{lemma}{0}
\renewcommand{\thelemma}{S\arabic{lemma}}

\setcounter{proposition}{0}
\renewcommand{\theproposition}{S\arabic{proposition}}

\setcounter{corollary}{0}
\renewcommand{\thecorollary}{S\arabic{corollary}}

\setcounter{assumption}{0}
\renewcommand{\theassumption}{S\arabic{assumption}}

\setcounter{definition}{0}
\renewcommand{\thedefinition}{S\arabic{definition}}

\setcounter{rmk}{0}
\renewcommand{\thermk}{S\arabic{rmk}}

\setcounter{example}{0}
\renewcommand{\theexample}{S\arabic{example}}

\setcounter{equation}{0}
\renewcommand{\theequation}{S\arabic{equation}}

\setcounter{figure}{0}
\renewcommand{\thefigure}{S\arabic{figure}}

\setcounter{table}{0}
\renewcommand{\thetable}{S\arabic{table}}

\setcounter{page}{1}
\renewcommand{\thepage}{S\arabic{page}}

\vspace{1cm}
\begin{center}
{\LARGE
{\bf Supplementary Materials for ``Accuracy of Uniform Inference on Fine Grid Points''}
}
\end{center}

\section{Proofs of General Results}

\begin{rmk}\label{rmk:small_error_reduction}
As in \citet{CCKK22}, we use the following elementary reduction throughout the proofs.
When establishing a probability bound of the form
\[
    \mathbb P(E_n)\le C\left\{\frac{B_n^2\log^5(2nM_n)}{n}\right\}^{1/4}+R_n,
\]
where $E_n$ is an event, $M_n\ge1$ and $R_n\ge0$, we may assume without loss of generality that $n\ge n_0$ and 
\begin{align*}
    \frac{B_n^2\log^5(2nM_n)}{n}\le c
\end{align*}
for any fixed $n_0\ge2$ and sufficiently small constant $c>0$.
Otherwise, the desired bound is trivial upon increasing $C$.
\end{rmk}

\subsection{Proof of \cref{thm:lower_general}}

Before the proof of \cref{thm:lower_general}, we introduce two technical lemmas.

\begin{lemma}[Theorem 2.1 in \cite{CCKK22}]
\label{lem:lower_studentized_CLT}
Under \cref{as:DGP_lower}(i) and (iii),
for every deterministic finite collection
$\{u_j\}_{j=1}^d\subset\mathcal X$, where $d\ge1$,
\begin{align}
    &\sup_{s\in\mathbb R}
    \left|
        \mathbb P\left(
            \max_{1\le j\le d}| T_n(u_j)|\le s
        \right)
        -
        \mathbb P\left(
            \max_{1\le j\le d}|G_n(u_j)|\le s
        \right)
    \right|
    \le
    C\left\{
        \frac{B_n^2\log^5(2nd)}{n}
    \right\}^{1/4},
    \label{eq:lower_studentized_CLT}
\end{align}
where $C>0$ is universal.
The covariance matrix of
$(G_n(u_1),\ldots,G_n(u_d))$ may be singular.
\end{lemma}

\begin{lemma}[Khatri--\v{S}id\'{a}k inequality {\citep[Corollary~2.4.6]{gine2016mathematical}}]
\label{lem:khatri_sidak}
Let $(Z_1,\ldots,Z_d)$ be a centered Gaussian vector, where $d\ge1$. 
The covariance matrix is allowed to be singular.
Then, for every $t\ge0$, $\mathbb P\left(\max_{1\le j\le d}|Z_j|\le t\right)\ge\prod_{j=1}^d\mathbb P\left(|Z_j|\le t\right)$.
\end{lemma}

\begin{proof}[Proof of \cref{thm:lower_general}]
The proof proceeds in five steps.

\paragraph{Step 1:}

Fix $\alpha\in(0,1)$.
For any integer $k\ge1$ and any deterministic tuple $(z_1,\ldots,z_k)\in\mathcal X^k$, we can see that
\begin{align}
    &\mathbb P\left(\sup_{x\in\mathcal X}|\hat T_n(x)|\le\hat c^\star_{n,1-\alpha}\right)\nonumber\\
    &\le\mathbb P\left(\hat c^\star_{n,1-\alpha}>t_{\alpha,p}\right)
    +\mathbb P\left(\sup_{x\in\mathcal X}|\hat T_n(x)|\le t_{\alpha,p}\right)\nonumber\\
    &\le\mathbb P\left(\hat c^\star_{n,1-\alpha}>t_{\alpha,p}\right)
    +\mathbb P\left(\max_{1\le a\le k}|\hat T_n(z_a)|\le t_{\alpha,p}\right)\nonumber\\
    &\le\mathbb P\left(\hat c^\star_{n,1-\alpha}>t_{\alpha,p}\right)
    +\mathbb P\left(\max_{1\le a\le k}|G_n(z_a)|\le t_{\alpha,p}\right)\nonumber\\
    &\quad+\left|\mathbb P\left(\max_{1\le a\le k}|T_n(z_a)|\le t_{\alpha,p}\right)-\mathbb P\left(\max_{1\le a\le k}|G_n(z_a)|\le t_{\alpha,p}\right)\right|\nonumber\\
    &\quad+\left|\mathbb P\left(\max_{1\le a\le k}|\hat T_n(z_a)|\le t_{\alpha,p}\right)-\mathbb P\left(\max_{1\le a\le k}|T_n(z_a)|\le t_{\alpha,p}\right)\right|.
    \label{eq:lower_coverage_reduction}
\end{align}
We evaluate the four terms on the right-hand side of \cref{eq:lower_coverage_reduction} in the following steps.

\paragraph{Step 2:}
We first bound the first term on the right-hand side of \cref{eq:lower_coverage_reduction}.
By \cref{eq:variance_estimator}, conditional on $X_{1:n}$, the bootstrap vector is centered Gaussian, the Khatri--\v{S}id\'{a}k inequality (\cref{lem:khatri_sidak}) therefore gives
\begin{align*}
    &\mathbb P^\star\left(\max_{1\le j\le p}|\hat T_n^\star(x_j)|\le t_{\alpha,p}\mid X_{1:n}\right)\\
    &\qquad\ge
    \prod_{j=1}^p\mathbb P^\star\left(|\hat T_n^\star(x_j)|\le t_{\alpha,p}\mid X_{1:n}\right)\ge
    \{1-2\bar\Phi(t_{\alpha,p})\}^p
    \ge1-2p\bar\Phi(t_{\alpha,p})
    =1-\alpha,
\end{align*}
where the final equality follows from the fact that  $t_{\alpha,p}$ is the $(1-\alpha/(2p))$-quantile of $N(0,1)$, satisfying $2p\bar\Phi(t_{\alpha,p})=\alpha$.
Thus, the conditional $(1-\alpha)$-quantile of the maximum of Gaussian random vector is bounded above by the Bonferroni critical value $t_{\alpha,p}$.
Hence, $\hat c^\star_{n,1-\alpha}\le t_{\alpha,p}$ almost surely, which implies
\begin{align*}
    \mathbb P(\hat c^\star_{n,1-\alpha}>t_{\alpha,p})
    &=0.
\end{align*}

\paragraph{Step 3:}

We next bound the second term on the right-hand side of \cref{eq:lower_coverage_reduction} from above.
Fix any integer $k\ge1$ and deterministic points $(z_1,\ldots,z_k)\in\mathcal X^k$.
The maximum $\max_{1\le a\le k}|G_n(z_a)|$ is at most $t$ if and only if no index $a$ satisfies $|G_n(z_a)|>t$.
We therefore count the indices satisfying $|G_n(z_a)|>t$ to bound the probability of the maximum being at most $t$ using the mean and variance of the count.
Define the random variable $N$ which counts the indices $a$ for which $|G_n(z_a)|>t$ as
\begin{align*}
    N
    &\coloneqq\sum_{a=1}^k\mathbf 1\{|G_n(z_a)|>t\}.
\end{align*}
Observe that
\begin{align*}
    \mathbb P\left(\max_{1\le a\le k}|G_n(z_a)|\le t\right)
    &=\mathbb P(N=0).
\end{align*}
To bound $\mathbb P(N=0)$ from above, we first bound $\mathbb P(N>0)$ from below.
Since $N=N\mathbf 1\{N>0\}$, Cauchy--Schwarz gives
\begin{align*}
    \{\mathbb E[N]\}^2
    &=\{\mathbb E[N\mathbf 1\{N>0\}]\}^2\le\mathbb E[N^2]\,\mathbb P(N>0).
\end{align*}
Using $\mathbb E[N^2]=\{\mathbb E[N]\}^2+\Var(N)$ and $\mathbb E[N]=k\,2\bar\Phi(t)>0$, we obtain
\begin{align*}
    \mathbb P\left(\max_{1\le a\le k}|G_n(z_a)|\le t\right)
    &=\mathbb P(N=0)=1-\mathbb P(N>0)\le1-\frac{\{\mathbb E[N]\}^2}{\mathbb E[N^2]}=\frac{\Var(N)}{\{\mathbb E[N]\}^2+\Var(N)}.
\end{align*}
Since $G_n(z_a)$ has the standard normal distribution, we can see that
\begin{align*}
    \mathbb E[N]
    & =\mathbb E\left[\sum_{a=1}^k\mathbf 1\{|G_n(z_a)|>t\}\right] \\
    & =\sum_{a=1}^k\mathbb E\left[\mathbf 1\{|G_n(z_a)|>t\}\right]=\sum_{a=1}^k\mathbb P(|G_n(z_a)|>t)=k\,2\bar\Phi(t)>0,
\end{align*}
so we have
\begin{align}
    \mathbb P\left(\max_{1\le a\le k}|G_n(z_a)|\le t\right)
    &=\mathbb P(N=0)\le\frac{\Var(N)}{\{\mathbb E[N]\}^2+\Var(N)}\le\frac{\Var(N)}{\{k\,2\bar\Phi(t)\}^2}.
    \label{eq:lower_count_nonexceedance}
\end{align}

For fixed $k$ and $t$, the denominator $\{k\,2\bar\Phi(t)\}^2$ in \cref{eq:lower_count_nonexceedance} does not depend on the points $z_1,\ldots,z_k$.
To obtain a small upper bound, therefore, we should derive points $z_1,\ldots,z_k$ for which $\Var(N)$ is small.
We first expand $\Var(N)$ to express its dependence on $z_1,\ldots,z_k$:
\begin{align}
    \Var(N)
    &=\sum_{a=1}^k\Var\left(\mathbf 1\{|G_n(z_a)|>t\}\right)+\sum_{\substack{1\le a,b\le k\\a\ne b}}\Cov\left(\mathbf 1\{|G_n(z_a)|>t\},\mathbf 1\{|G_n(z_b)|>t\}\right)\nonumber\\
    &=k\,2\bar\Phi(t)\{1-2\bar\Phi(t)\}+\sum_{\substack{1\le a,b\le k\\a\ne b}}\left\{\mathbb P(|G_n(z_a)|>t,|G_n(z_b)|>t)-\{2\bar\Phi(t)\}^2\right\}.
    \label{eq:lower_count_variance_expansion}
\end{align}
The variance $\Var(N)$ depends on the locations of $z_1,\ldots,z_k$, and different choices of points may give different variances.
For \cref{eq:lower_count_nonexceedance}, it is enough to obtain an upper bound on the variance for one suitable choice of points.
We use the average variance over the locations to obtain this bound.
At least one choice of points must have variance no larger than the average.
We therefore calculate the average of \cref{eq:lower_count_variance_expansion} over all locations of the points.
Specifically, we integrate the variance $\Var(N)$ over $\mathcal X^k$ and divide by $\operatorname{Leb}(\mathcal X^k)$ to obtain the average.

By \cref{eq:lower_count_variance_expansion} and Fubini's theorem,
\begin{align}
     & \frac1{|\mathcal X|^k}\int_{\mathcal X^k}\Var\left(\sum_{a=1}^k\mathbf 1\{|G_n(u_a)|>t\}\right)\,du_1\cdots du_k\nonumber\\
    &=k\,2\bar\Phi(t)\{1-2\bar\Phi(t)\}\nonumber\\
    &\quad+\frac1{|\mathcal X|^k}\sum_{\substack{1\le a,b\le k\\a\ne b}}\int_{\mathcal X^k}\left\{\mathbb P(|G_n(u_a)|>t,|G_n(u_b)|>t)-\{2\bar\Phi(t)\}^2\right\}\,du_1\cdots du_k\nonumber\\
    &=k\,2\bar\Phi(t)\{1-2\bar\Phi(t)\}\nonumber\\
    &\quad+\frac{k(k-1)}{|\mathcal X|^2}\int_{\mathcal X}\int_{\mathcal X}\left\{\mathbb P(|G_n(x)|>t,|G_n(y)|>t)-\{2\bar\Phi(t)\}^2\right\}\,dx\,dy.
    \label{eq:lower_average_count_variance}
\end{align}
For $k\ge2$, each term indexed by $a\ne b$ depends only on $u_a$ and $u_b$.
Integrating over the other $k-2$ coordinates gives the factor $|\mathcal X|^{k-2}$, and there are $k(k-1)$ ordered pairs.
For $k=1$, the sum over $a\ne b$ is empty.
We now express the double integral in \cref{eq:lower_average_count_variance} in terms of $M_n(t)$.
Define
\begin{align*}
    V_n^G(t)
    &\coloneqq\operatorname{Leb}\{x\in\mathcal X:|G_n(x)|>t\}.
\end{align*}
Since each $G_n(x)$ is standard normal, joint measurability and Fubini's theorem give
\begin{align}
    \mathbb E[V_n^G(t)]
    &=\int_{\mathcal X}\mathbb P(|G_n(x)|>t)\,dx
    =|\mathcal X|\,2\bar\Phi(t),
    \label{eq:lower_excursion_mean}
\end{align}
and
\begin{align}
    \Var[V_n^G(t)]
    &=\int_{\mathcal X}\int_{\mathcal X}
    \Cov\left(\mathbf1\{|G_n(x)|>t\},\mathbf1\{|G_n(y)|>t\}\right)\,dx\,dy\nonumber\\
    &=\int_{\mathcal X}\int_{\mathcal X}
    \left\{\mathbb P(|G_n(x)|>t,|G_n(y)|>t)-\{2\bar\Phi(t)\}^2\right\}\,dx\,dy \label{eq:lower_excursion_joint_probability}\\
    &=2\bar\Phi(t)\{1-2\bar\Phi(t)\}
    \int_{\mathcal X}\int_{\mathcal X}
    \Corr\left(\mathbf1\{|G_n(x)|>t\},\mathbf1\{|G_n(y)|>t\}\right)\,dx\,dy\nonumber\\
    &=|\mathcal X|^2\,2\bar\Phi(t)\{1-2\bar\Phi(t)\}M_n(t)^{-1}.
    \label{eq:lower_excursion_variance_integral}
\end{align}
By \cref{eq:lower_average_count_variance,eq:lower_excursion_joint_probability}, for every integer $k\ge1$, there exist deterministic points $z_1,\ldots,z_k\in\mathcal X$ such that
\begin{align}
    \Var(N)
    &\le k\,2\bar\Phi(t)\{1-2\bar\Phi(t)\}+\frac{k(k-1)}{|\mathcal X|^2}\Var[V_n^G(t)].
    \label{eq:lower_selected_count_variance}
\end{align}
Substituting \cref{eq:lower_selected_count_variance} into \cref{eq:lower_count_nonexceedance} and using the definition of $M_n(t)$, we obtain
\begin{align}
    & \mathbb P\left(\max_{1\le a\le k}|G_n(z_a)|\le t\right)\nonumber \\
    &\le\frac{\Var(N)}{\{k\,2\bar\Phi(t)\}^2}\le\frac{1-2\bar\Phi(t)}{k\,2\bar\Phi(t)}+\left(1-\frac1k\right)\frac{\Var[V_n^G(t)]}{|\mathcal X|^2\{2\bar\Phi(t)\}^2}=\frac{1-2\bar\Phi(t)}{2\bar\Phi(t)}\left\{\frac1k+\frac{1-1/k}{M_n(t)}\right\},
    \label{eq:lower_selected_gaussian_bound}
\end{align}
where the final equality follows from \cref{eq:lower_excursion_variance_integral}.
To keep the dimension of the auxiliary finite set polynomial in $n$ and $p$, we choose $k\coloneqq np$ and fix deterministic points satisfying \cref{eq:lower_selected_count_variance} for this value of $k$.
Since $1-1/k\le1$ and $1-2\bar\Phi(t)\le1$, \cref{eq:lower_selected_gaussian_bound} implies
\begin{align}
    \mathbb P\left(\max_{1\le a\le k}|G_n(z_a)|\le t\right)
    &\le\frac{1-2\bar\Phi(t)}{2\bar\Phi(t)}\left\{\frac1{np}+\frac{1-1/(np)}{M_n(t)}\right\}\le\frac1{2np\bar\Phi(t)}+\frac1{2\bar\Phi(t)M_n(t)}.
    \label{eq:lower_witness_gaussian_bound}
\end{align}

\paragraph{Step 4:}

We next evaluate the Gaussian approximation error for the standardized statistic.
For the deterministic points $z_1,\ldots,z_k$ selected in Step 3, \cref{as:DGP_lower} verifies Conditions E and M of \citet{CCKK22}.
Thus, Theorem 2.1 of \citet{CCKK22}, applied to the $2k$-dimensional vector of standardized summands and their negatives, gives
\begin{align}
    &\sup_{s\in\mathbb R}\left|\mathbb P\left(\max_{1\le a\le k}|T_n(z_a)|\le s\right)-\mathbb P\left(\max_{1\le a\le k}|G_n(z_a)|\le s\right)\right|\le C\left\{\frac{B_n^2\log^5(2nk)}{n}\right\}^{1/4}.  \label{eq:hd-CLT}
\end{align}
Since $k=np$, we have $\log(2nk)=\log(2n^2p)\le2\log(2np)$.
Therefore,
\begin{align}
    &\sup_{s\in\mathbb R}\left|\mathbb P\left(\max_{1\le a\le k}|T_n(z_a)|\le s\right)-\mathbb P\left(\max_{1\le a\le k}|G_n(z_a)|\le s\right)\right|\le C\left\{\frac{B_n^2\log^5(2np)}{n}\right\}^{1/4}.
    \label{eq:lower_standardized_CLT}
\end{align}

\paragraph{Step 5:}

We next evaluate the error caused by studentization.
By Lemma 15 in \cite[Chapter~2]{le1986asymptotic}, for every $r>0$, we have
\begin{align}
    &\sup_{t\in\mathbb R}\left|\mathbb P\left(\max_{1\le a\le k}|\hat T_n(z_a)|\le t\right)-\mathbb P\left(\max_{1\le a\le k}|T_n(z_a)|\le t\right)\right|\nonumber\\
    &\quad\le\sup_{t\in\mathbb R}\mathbb P\left(t<\max_{1\le a\le k}|T_n(z_a)|\le t+r\right)
    +\mathbb P\left(\max_{1\le a\le k}|\hat T_n(z_a)-T_n(z_a)|>r\right).
    \label{eq:lower_studentization_decomposition}
\end{align}
We evaluate the two terms on the right-hand side with $r= C\sqrt{B_n^2\log^2(2nk)/n}$, where $C>0$ is a sufficiently large universal constant.

We first evaluate the second term.
By the definitions of $T_n$ and $\hat T_n$, we have
\begin{align*}
    &\max_{1\le a\le k}|\hat T_n(z_a)-T_n(z_a)|\le\max_{1\le a\le k}|T_n(z_a)|\max_{1\le a\le k}\frac{|\hat\sigma_n^2(z_a)-\sigma_n^2(z_a)|}{\hat\sigma_n(z_a)\{\hat\sigma_n(z_a)+\sigma_n(z_a)\}}.
\end{align*}

We apply Lemma 4.2 of \citet{CCKK22} to the standardized summands $\tilde\psi_{i,h_n}(z_a)$.
By Eqs.(35) and (36) in Lemma~4.2 of \citet{CCKK22}, with probability at least $1-1/n-3\sqrt{B_n^2\log^2(2nk)/n}$, we have
\begin{align}
    &\min_{1\le a\le k}\frac{\hat\sigma_n^2(z_a)}{\sigma_n^2(z_a)}\ge\frac12, \quad \max_{1\le a\le k}\left|\frac{\hat\sigma_n^2(z_a)}{\sigma_n^2(z_a)}-1\right|\le C \left\{\frac{B_n^2\log(2nk)}{n}\right\}^{1/2}. \label{eq:lower_relative_variances}
\end{align}
On this event,
\begin{align*}
    &\max_{1\le a\le k}\frac{|\hat\sigma_n^2(z_a)-\sigma_n^2(z_a)|}{\hat\sigma_n(z_a)\{\hat\sigma_n(z_a)+\sigma_n(z_a)\}}\le\sqrt2\max_{1\le a\le k}\left|\frac{\hat\sigma_n^2(z_a)}{\sigma_n^2(z_a)}-1\right|
    \le C\left\{\frac{B_n^2\log(2nk)}{n}\right\}^{1/2}.
\end{align*}
Consequently,
\begin{align*}
    &\mathbb P\left(\max_{1\le a\le k}|\hat T_n(z_a)-T_n(z_a)|>r\right)\\
    &\le\mathbb P\left(\max_{1\le a\le k}|T_n(z_a)|>\sqrt{2\log(2nk)}\right)
    +\frac1n+3\left\{\frac{B_n^2\log^3(2nk)}{n}\right\}^{1/2}.
\end{align*}

By the Gaussian approximation established in Step 4 and the Gaussian union bound,
\begin{align*}
    &\mathbb P\left(\max_{1\le a\le k}|T_n(z_a)|>\sqrt{2\log(2nk)}\right)\\
    &\le\mathbb P\left(\max_{1\le a\le k}|G_n(z_a)|>\sqrt{2\log(2nk)}\right)
    +C\left\{\frac{B_n^2\log^5(2nk)}{n}\right\}^{1/4}\le\frac1n+C\left\{\frac{B_n^2\log^5(2nk)}{n}\right\}^{1/4}.
\end{align*}
Combining these inequalities, we obtain
\begin{align}
    &\mathbb P\left(\max_{1\le a\le k}|\hat T_n(z_a)-T_n(z_a)|>r\right)\nonumber\\
    &\quad\le\frac2n+3\left\{\frac{B_n^2\log^3(2nk)}{n}\right\}^{1/2}
    +C\left\{\frac{B_n^2\log^5(2nk)}{n}\right\}^{1/4}\le C\left\{\frac{B_n^2\log^5(2nk)}{n}\right\}^{1/4}.
    \label{eq:lower_studentization_probability}
\end{align}

We next evaluate the first term on the right-hand side of \cref{eq:lower_studentization_decomposition}.
By Lemma 4.4 of \citet{CCKK22}, applied to the standardized summands and their negatives, we obtain
\begin{align*}
    &\sup_{t\in\mathbb R}\mathbb P\left(t<\max_{1\le a\le k}|T_n(z_a)|\le t+r\right)\\
    &\le C\left\{r\sqrt{\log(2k)}+\left(\frac{B_n^2\log^5(2nk)}{n}\right)^{1/4}\right\}\\
    &\le C\left\{\left(\frac{B_n^2\log^3(2nk)}{n}\right)^{1/2}+\left(\frac{B_n^2\log^5(2nk)}{n}\right)^{1/4}\right\}\le C\left\{\frac{B_n^2\log^5(2nk)}{n}\right\}^{1/4}.
\end{align*}
Combining this inequality with \cref{eq:lower_studentization_decomposition,eq:lower_studentization_probability}, we have
\begin{align}
    &\sup_{t\in\mathbb R}\left|\mathbb P\left(\max_{1\le a\le k}|\hat T_n(z_a)|\le t\right)-\mathbb P\left(\max_{1\le a\le k}|T_n(z_a)|\le t\right)\right|\le C\left\{\frac{B_n^2\log^5(2nk)}{n}\right\}^{1/4}.
    \label{eq:lower_studentization_comparison}
\end{align}
Finally, since $k=np$ and $\log(2nk)=\log(2n^2p)\le2\log(2np)$, we obtain
\begin{align}
    &\sup_{t\in\mathbb R}\left|\mathbb P\left(\max_{1\le a\le k}|\hat T_n(z_a)|\le t\right)-\mathbb P\left(\max_{1\le a\le k}|T_n(z_a)|\le t\right)\right|
    \le C\left\{\frac{B_n^2\log^5(2np)}{n}\right\}^{1/4}.
    \label{eq:lower_studentization_comparison}
\end{align}

\paragraph{Step 6:}

Combining \cref{eq:lower_coverage_reduction,eq:lower_witness_gaussian_bound,eq:lower_standardized_CLT,eq:lower_studentization_comparison} with the result of Step 2, we obtain
\begin{align*}
    \mathbb P\left(\sup_{x\in\mathcal X}|\hat T_n(x)|\le\hat c^\star_{n,1-\alpha}\right)
    &\le\frac1{2np\bar\Phi(t_{\alpha,p})}
    +\frac1{2\bar\Phi(t_{\alpha,p})M_n(t_{\alpha,p})}
    +C\left\{\frac{B_n^2\log^5(2np)}{n}\right\}^{1/4}\\
    &=\frac1{n\alpha}
    +\frac{p}{\alpha M_n(t_{\alpha,p})}
    +C\left\{\frac{B_n^2\log^5(2np)}{n}\right\}^{1/4},
\end{align*}
where the equality follows from $2p\bar\Phi(t_{\alpha,p})=\alpha$.
This completes the proof.
\end{proof}

\subsection{Proof of \cref{thm:upper_general}}

\subsubsection{Auxiliary Results for Proof of \cref{thm:upper_general}}
Before the proof of \cref{thm:upper_general}, we provide two auxiliary results.
The proof of these results are given in  \cref{subsubsec:proof:lem:gaussian_two_point_crossing} and \cref{subsubsec:proof:lem:upper_finite_approximation}, respectively.

\begin{lemma}
\label{lem:gaussian_two_point_crossing}
Let $(U,V)$ be a centered Gaussian vector with unit marginal variances,
and set $\theta=\arccos\{\mathbb E[UV]\}\in[0,\pi]$.  Then, for every
$t\ge0$,
\begin{align}
    \mathbb P(U\le t<V)
    &=\frac{1}{2\pi}\int_0^\theta
        \exp\left(-\frac{t^2}{1+\cos s}\right)ds
     =\frac{e^{-t^2/2}}{2\pi}\{\theta-H_t(\theta)\}.
    \label{eq:upper_bivariate_crossing} 
\end{align}
\end{lemma}

\begin{lemma}
\label{lem:upper_finite_approximation}
Suppose that Assumption~\ref{ass:upper_general} holds and $n\ge2$. There exists a universal constant $C>0$ such that, for every integer $K\ge1$, there exists a deterministic finite set
\[
    \{x_1,\ldots,x_p\}\subseteq\mathcal Z_n=\{z_1<\cdots<z_m\}\subseteq\mathcal X,
    \qquad m\le p+K,
\]
for which, for every $L_n\in(0,1]$,
\begin{align}
    &\mathbb P\left(\sup_{x\in\mathcal X}|\hat T_n(x)|-\max_{z\in\mathcal Z_n}|\hat T_n(z)|>L_n\right) \le C\left\{\frac{B_n^2\log^5(2nm)}{n}\right\}^{1/4}+\frac{Cn(1+\mathcal L_n)^2}{K L_n^2}.
    \label{eq:upper_finite_approximation}
\end{align}
The set $\mathcal Z_n$ can be chosen independently of $L_n$.
\end{lemma}

\subsubsection{Proof of \cref{thm:upper_general}}

\begin{proof}[Proof of Theorem~\ref{thm:upper_general}]
Our goal is to bound the coverage error
\[
    \left|\mathbb P\left(\sup_{x\in\mathcal X}|\hat T_n(x)|\le\hat c_{n,1-\alpha}^{\star}\right)-(1-\alpha)\right|.
\]
Since the evaluation grid is contained in $\mathcal X$, we have
\begin{align*}
    &\left\{\max_{1\le j\le p}|\hat T_n(x_j)|\le\hat c_{n,1-\alpha}^{\star}\right\}\\
    &=\left\{\sup_{x\in\mathcal X}|\hat T_n(x)|\le\hat c_{n,1-\alpha}^{\star}\right\}
    \,\dot\cup\,
    \left\{\max_{1\le j\le p}|\hat T_n(x_j)|\le\hat c_{n,1-\alpha}^{\star}<\sup_{x\in\mathcal X}|\hat T_n(x)|\right\},
\end{align*}
where $\dot\cup$ denotes a disjoint union. Taking probabilities and subtracting $1-\alpha$ gives
\begin{align}
    &\mathbb P\left(\sup_{x\in\mathcal X}|\hat T_n(x)|\le\hat c_{n,1-\alpha}^{\star}\right)-(1-\alpha)\nonumber\\
    &=\left\{\mathbb P\left(\max_{1\le j\le p}|\hat T_n(x_j)|\le\hat c_{n,1-\alpha}^{\star}\right)-(1-\alpha)\right\}
    -\mathbb P\left(\max_{1\le j\le p}|\hat T_n(x_j)|\le\hat c_{n,1-\alpha}^{\star}<\sup_{x\in\mathcal X}|\hat T_n(x)|\right).
    \label{eq:upper_coverage_identity}
\end{align}
To evaluate the second term on the right-hand side of \eqref{eq:upper_coverage_identity},  fix an integer $K\ge1$ and $L_n\in(0,1]$, whose values will be specified at the end of the proof.
By \cref{lem:upper_finite_approximation}, there exists a deterministic finite set
\[
    \{x_1,\ldots,x_p\}\subseteq\mathcal Z_n=\{z_1<\cdots<z_m\}\subseteq\mathcal X,
    \qquad m\le p+K,
\]
such that
\begin{align}
    &\mathbb P\left(\sup_{x\in\mathcal X}|\hat T_n(x)|-\max_{1\le a\le m}|\hat T_n(z_a)|>L_n\right)\le C\left\{\frac{B_n^2\log^5(2nm)}{n}\right\}^{1/4}+\frac{Cn(1+\mathcal L_n)^2}{KL_n^2}=:\varepsilon_n.
    \label{eq:upper_auxiliary_approximation}
\end{align}
If the supremum over $\mathcal{X}$ exceeds the maximum over $\mathcal{Z}_n$ by at most $L_n$, then
\[
    \sup_{x\in\mathcal X}|\hat T_n(x)|
    \le\max_{1\le a\le m}|\hat T_n(z_a)|+L_n.
\]
Consequently,
\begin{align}
    &\left\{\max_{1\le j\le p}|\hat T_n(x_j)|\le\hat c_{n,1-\alpha}^{\star}<\sup_{x\in\mathcal X}|\hat T_n(x)|\right\}\nonumber\\
    &\quad\subseteq
    \left\{\sup_{x\in\mathcal X}|\hat T_n(x)|-\max_{1\le a\le m}|\hat T_n(z_a)|>L_n\right\}
    \cup
    \left\{\max_{1\le j\le p}|\hat T_n(x_j)|\le\hat c_{n,1-\alpha}^{\star}<\max_{1\le a\le m}|\hat T_n(z_a)|+L_n\right\}.
    \label{eq:upper_auxiliary_event_inclusion}
\end{align}
Therefore, we have
\begin{align}
    &\left|\mathbb P\left(\sup_{x\in\mathcal X}|\hat T_n(x)|\le\hat c_{n,1-\alpha}^{\star}\right)-(1-\alpha)\right|\nonumber\\
    &\le\left|\mathbb P\left(\max_{1\le j\le p}|\hat T_n(x_j)|\le\hat c_{n,1-\alpha}^{\star}\right)-(1-\alpha)\right|
    +\mathbb P\left(\max_{1\le j\le p}|\hat T_n(x_j)|\le\hat c_{n,1-\alpha}^{\star}<\sup_{x\in\mathcal X}|\hat T_n(x)|\right)\nonumber\\
    &\le\left|\mathbb P\left(\max_{1\le j\le p}|\hat T_n(x_j)|\le\hat c_{n,1-\alpha}^{\star}\right)-(1-\alpha)\right|
    \nonumber\\
    &\qquad+\varepsilon_n+\mathbb P\left(\max_{1\le j\le p}|\hat T_n(x_j)|\le\hat c_{n,1-\alpha}^{\star}<\max_{1\le a\le m}|\hat T_n(z_a)|+L_n\right),
    \label{eq:upper_error_decomposition}
\end{align}
where the first inequality follows from the triangle inequality, and the second follows from the preceding event inclusion and \eqref{eq:upper_auxiliary_approximation}.
To bound the first term and the last probability in \eqref{eq:upper_error_decomposition}, we first derive the important bounds with a deterministic threshold $t$ instead of the random threshold $\hat c_{n,1-\alpha}^{\star}$ (Steps~1 and 2). 
We then establish deterministic upper and lower bounds for $\hat c_{n,1-\alpha}^{\star}$ (Step~3) and combine these results to complete the proof (Step~4).


\paragraph{Step 1:}
We first establish preliminary Gaussian approximation bounds.

\subparagraph{Step 1-1:}
We first consider Gaussian approximation of the studentized statistics over $\mathcal Z_n$ and $\{x_1,\ldots,x_p\}$. 
The Gaussian approximation and studentization arguments in Steps~4 and 5 of the proof of \cref{thm:lower_general} apply to any deterministic finite collection of points satisfying \cref{as:DGP_lower}(i) and (iii). Applying the same arguments under \cref{ass:upper_general}(iii), first to $\mathcal Z_n$ and then to $\{x_1,\ldots,x_p\}$, gives
\begin{align}
    &\sup_{t\in\mathbb R}\left|\mathbb P\left(\max_{1\le a\le m}|\hat T_n(z_a)|\le t\right)-\mathbb P\left(\max_{1\le a\le m}|G_n(z_a)|\le t\right)\right|
    \le C\left\{\frac{B_n^2\log^5(2nm)}{n}\right\}^{1/4},
    \label{eq:upper_finite_studentized}\\
    &\sup_{t\in\mathbb R}\left|\mathbb P\left(\max_{1\le j\le p}|\hat T_n(x_j)|\le t\right)-\mathbb P\left(\max_{1\le j\le p}|G_n(x_j)|\le t\right)\right|
    \le C\left\{\frac{B_n^2\log^5(2np)}{n}\right\}^{1/4}.
    \label{eq:upper_finite_studentized_grid}
\end{align}

\subparagraph{Step 1-2:}
We next consider  Gaussian approximation of the conditional bootstrap distribution over $\{x_1,\ldots,x_p\}$. 
By \eqref{eq:lower_relative_variances}, Since the bootstrap vector is conditionally centered Gaussian, Proposition~2.1 of \citet{CCKK22} gives
\begin{align}
    &\sup_{t\in\mathbb R}\left|\mathbb P^\star\left(\max_{j\in[p]}|\hat T_n^\star(x_j)|\le t\mid X_{1:n}\right)-\mathbb P\left(\max_{j\in[p]}|G_n(x_j)|\le t\right)\right|\nonumber\\
    &\quad\le C\log(2p)\left(\max_{j,k\in[p]}\left|\operatorname{Cov}^\star\left(\hat T_n^\star(x_j),\hat T_n^\star(x_k)\mid X_{1:n}\right)-\Cov(G_n(x_j)G_n(x_k))\right|\right)^{1/2}. \nonumber
\end{align}
Since $\operatorname{Cov}(T_n(x_j),T_n(x_k))=\operatorname{Cov}(G_n(x_j),G_n(x_k))$ and $|\operatorname{Cov}(T_n(x_j),T_n(x_k))|\le1$,
\begin{align}
    &\left|\operatorname{Cov}^{\star}\left(\hat T_n^\star(x_j),\hat T_n^\star(x_k)\mid X_{1:n}\right)-\operatorname{Cov}(G_n(x_j),G_n(x_k))\right|\nonumber\\
    &\le\frac{\sigma_n(x_j)\sigma_n(x_k)}{\hat\sigma_n(x_j)\hat\sigma_n(x_k)}
    \left|\frac{\hat\sigma_n(x_j)\hat\sigma_n(x_k)}{\sigma_n(x_j)\sigma_n(x_k)}
    \operatorname{Cov}^{\star}\left(\hat T_n^\star(x_j),\hat T_n^\star(x_k)\mid X_{1:n}\right)
    -\operatorname{Cov}(T_n(x_j),T_n(x_k))\right|\nonumber\\
    &\quad+\left|\frac{\sigma_n(x_j)\sigma_n(x_k)}{\hat\sigma_n(x_j)\hat\sigma_n(x_k)}-1\right|. \nonumber
\end{align}
By Eq.~(36) in Lemma~4.2 of \citet{CCKK22}, with probability at least $1-1/n-3\{B_n^2\log^3(2np)/n\}^{1/2}$, we have
\begin{align}
    &\max_{j,k\in[p]}\left|\frac{\hat\sigma_n(x_j)\hat\sigma_n(x_k)}{\sigma_n(x_j)\sigma_n(x_k)}
    \operatorname{Cov}^{\star}\left(\hat T_n^\star(x_j),\hat T_n^\star(x_k)\mid X_{1:n}\right)
    -\operatorname{Cov}(T_n(x_j),T_n(x_k))\right|
    \le C\left\{\frac{B_n^2\log(2np)}{n}\right\}^{1/2},\nonumber\\
    &\max_{j,k\in[p]}\left|\frac{\sigma_n(x_j)\sigma_n(x_k)}{\hat\sigma_n(x_j)\hat\sigma_n(x_k)}-1\right|
    =\max_{j\in[p]}\left|\frac{\sigma_n^2(x_j)}{\hat\sigma_n^2(x_j)}-1\right|
    \le2\max_{j\in[p]}\left|\frac{\hat\sigma_n^2(x_j)}{\sigma_n^2(x_j)}-1\right|
    \le C\left\{\frac{B_n^2\log(2np)}{n}\right\}^{1/2}. \nonumber
\end{align}
Combining these bounds gives
\begin{align}
    &\sup_{t\in\mathbb R}\left|\mathbb P^\star\left(\max_{j\in[p]}|\hat T_n^\star(x_j)|\le t\mid X_{1:n}\right)-\mathbb P\left(\max_{j\in[p]}|G_n(x_j)|\le t\right)\right|
    \le C\left\{\frac{B_n^2\log^5(2np)}{n}\right\}^{1/4},
    \label{eq:upper_finite_bootstrap}
\end{align}
with probability at least $1-C\{B_n^2\log^5(2np)/n\}^{1/4}$.

\paragraph{Step 2:}
To bound the second term in \eqref{eq:upper_error_decomposition}, we first establish the  bound on the following quantity for every $t\ge0$;
\begin{align*}
    \mathbb P\left(\max_{1\le j\le p}|G_n(x_j)|\le t<\max_{1\le a\le m}|G_n(z_a)|\right).
\end{align*}

\subparagraph{Step 2-1:}
Suppose that
\[
    \max_{1\le j\le p}|G_n(x_j)|\le t<\max_{1\le a\le m}|G_n(z_a)|.
\]
Then $|G_n(x_j)|\le t$ for every $j\in[p]$, while $|G_n(z_a)|>t$ for some $a\in[m]$. Since $\mathcal Z_n\subseteq[x_1,x_p]=\bigcup_{j=1}^{p-1}[x_j,x_{j+1}]$, there exists $j\in[p-1]$ such that $z_a\in[x_j,x_{j+1}]$. Hence,
\[
    |G_n(x_j)|\le t,\qquad |G_n(x_{j+1})|\le t<|G_n(z_a)|\le\max_{z\in\mathcal Z_n\cap[x_j,x_{j+1}]}|G_n(z)|.
\]
Therefore,
\begin{align*}
    &\left\{\max_{1\le j\le p}|G_n(x_j)|\le t<\max_{1\le a\le m}|G_n(z_a)|\right\}\\
    &\quad\subseteq\bigcup_{j=1}^{p-1}\left\{|G_n(x_j)|\le t,\ |G_n(x_{j+1})|\le t<\max_{z\in\mathcal Z_n\cap[x_j,x_{j+1}]}|G_n(z)|\right\}.
\end{align*}
Since the centered Gaussian vector is symmetric in distribution, the union bound applied to positive and negative thresholds gives
\begin{align}
    &\mathbb P\left(\max_{1\le j\le p}|G_n(x_j)|\le t<\max_{1\le a\le m}|G_n(z_a)|\right)\nonumber\\
    &\quad\le 2\sum_{j=1}^{p-1}\mathbb P\left(G_n(x_j)\le t,\ G_n(x_{j+1})\le t<\max_{z\in\mathcal Z_n\cap[x_j,x_{j+1}]}G_n(z)\right).
    \label{eq:upper_cell_reduction}
\end{align}

\subparagraph{Step 2-2:}
Fix $j\in[p-1]$. For every realization of the Gaussian vector, we have
\begin{align}
    &\mathbf1\left\{G_n(x_j)\le t,\ G_n(x_{j+1})\le t<\max_{z\in\mathcal Z_n\cap[x_j,x_{j+1}]}G_n(z)\right\}\nonumber\\
    &\quad\le\sum_{\substack{1\le k<m\\x_j\le z_k<z_{k+1}\le x_{j+1}}}\mathbf1\{G_n(z_k)\le t<G_n(z_{k+1})\}
    -\mathbf1\{G_n(x_j)\le t<G_n(x_{j+1})\}.
    \label{eq:upper_crossing_count}
\end{align}
Indeed, the first and last indicators cannot both equal one, since they require $G_n(x_{j+1})\le t$ and $G_n(x_{j+1})>t$, respectively. If either equals one, then
\[
    G_n(x_j)\le t<G_n(z)
    \quad\text{for some }z\in\mathcal Z_n\cap[x_j,x_{j+1}].
\]
Consequently, there exists $k$ such that
\[
    x_j\le z_k<z_{k+1}\le x_{j+1},
    \qquad G_n(z_k)\le t<G_n(z_{k+1}).
\]
Thus, the sum over $k$ is at least the sum of the first and last indicators, proving \eqref{eq:upper_crossing_count}.

\subparagraph{Step 2-3:}
Taking expectations in \cref{eq:upper_crossing_count}, we obtain
\begin{align}
    &\mathbb P\left(G_n(x_j)\le t,\ G_n(x_{j+1})\le t<\max_{z\in\mathcal Z_n\cap[x_j,x_{j+1}]}G_n(z)\right)\nonumber\\
    &\quad\le\sum_{\substack{1\le k<m\\x_j\le z_k<z_{k+1}\le x_{j+1}}}\mathbb P\left(G_n(z_k)\le t<G_n(z_{k+1})\right)-\mathbb P\left(G_n(x_j)\le t<G_n(x_{j+1})\right)\nonumber\\
    &\quad=\frac{e^{-t^2/2}}{2\pi}\Bigg\{\sum_{\substack{1\le k<m\\x_j\le z_k<z_{k+1}\le x_{j+1}}}\left[d_n(z_k,z_{k+1})-H_t(d_n(z_k,z_{k+1}))\right]-d_n(x_j,x_{j+1})+H_t(d_n(x_j,x_{j+1}))\Bigg\}\nonumber\\
   &\quad\le\frac{e^{-t^2/2}}{2\pi}\Bigg\{\mathcal L_{n,j}-d_n(x_j,x_{j+1})+H_t(d_n(x_j,x_{j+1}))-\sum_{\substack{1\le k<m\\x_j\le z_k<z_{k+1}\le x_{j+1}}}H_t(d_n(z_k,z_{k+1}))\Bigg\}\nonumber\\
    &\quad\le\frac{e^{-t^2/2}}{2\pi}\left\{\mathcal L_{n,j}-d_n(x_j,x_{j+1})+H_t(d_n(x_j,x_{j+1}))\right\}.
    \label{eq:upper_cell_crossing}
\end{align}
The first inequality follows from \cref{eq:upper_crossing_count},
the equality follows from \cref{lem:gaussian_two_point_crossing},
the second inequality follows from the definition of $\mathcal L_{n,j}$, 
and the last inequality follows from $H_t\ge0$.

\subparagraph{Step 2-4:}
Combining \eqref{eq:upper_cell_reduction} and \eqref{eq:upper_cell_crossing}, we obtain
\begin{align}
    &\mathbb P\left(\max_{1\le j\le p}|G_n(x_j)|\le t<\max_{1\le a\le m}|G_n(z_a)|\right)\nonumber\\
    &\quad\le\frac{e^{-t^2/2}}{\pi}\sum_{j=1}^{p-1}\left\{\mathcal L_{n,j}-d_n(x_j,x_{j+1})+H_t(d_n(x_j,x_{j+1}))\right\}.
    \label{eq:upper_finite_crossing}
\end{align}

\paragraph{Step 3:}
Since Steps 1 and 2 provide bounds at deterministic thresholds $t$, we next establish deterministic bounds for the random bootstrap critical value $\hat c_{n,1-\alpha}^{\star}$.
Specifically, we show that
\begin{align}
    \mathbb P\left(q_n(u_-)\le\hat c_{n,1-\alpha}^{\star}\le q_n(u_+)\right)
    \ge1-C\left\{\frac{B_n^2\log^5(2np)}{n}\right\}^{1/4},
    \label{eq:upper_critical_brackets_probability}
\end{align}
where $u_-$ and $u_+$ are defined below.

\subparagraph{Step 3-1:} 
We first establish some preliminary properties of the Gaussian quantiles and define some quantities used in Step 3-2.
For every $s\in\mathbb R$,
\[
    \mathbb P\left(\max_{1\le j\le p}|G_n(x_j)|=s\right)
    \le\sum_{j=1}^p\mathbb P(|G_n(x_j)|=s)=0.
\]
Thus the Gaussian maximum has a continuous distribution even for a singular covariance matrix. 
In addition to $q_n(0)=0$, put
$q_n(1)=+\infty$ for this proof. 
Continuity gives
\begin{align}
    \mathbb P\left(\max_{1\le j\le p}|G_n(x_j)|\le q_n(u)\right)
    &=u,\qquad u\in[0,1].
    \label{eq:upper_quantile_identity}
\end{align}
Define the two deterministic probability levels
\begin{align*}
    u_-&:=\left[1-\alpha
         -C\left\{\frac{B_n^2\log^5(2np)}{n}\right\}^{1/4}\right]_+,
    \quad
    u_+:=\left[1-\alpha
         +C\left\{\frac{B_n^2\log^5(2np)}{n}\right\}^{1/4}\right]\wedge1.
\end{align*}
By their definitions,
\begin{align}
    0\le(1-\alpha)-u_-&\le
       C\left\{\frac{B_n^2\log^5(2np)}{n}\right\}^{1/4},
    &0\le u_+-(1-\alpha)&\le
       C\left\{\frac{B_n^2\log^5(2np)}{n}\right\}^{1/4}.
    \label{eq:upper_quantile_levels}
\end{align}

\subparagraph{Step 3-2:}
We now establish \eqref{eq:upper_critical_brackets_probability}. On the event where \eqref{eq:upper_finite_bootstrap} holds, suppose first that $u_->0$. For every $s<q_n(u_-)$, we have
\begin{align*}
    &\mathbb P^\star\left(\max_{1\le j\le p}|\hat T_n^\star(x_j)|\le s\mid X_{1:n}\right)\\
    &\le\mathbb P\left(\max_{1\le j\le p}|G_n(x_j)|\le s\right)+C\left\{\frac{B_n^2\log^5(2np)}{n}\right\}^{1/4}<u_-+C\left\{\frac{B_n^2\log^5(2np)}{n}\right\}^{1/4}=1-\alpha,
\end{align*}
where the strict inequality follows from the definition of $q_n(u_-)$. Hence $\hat c_{n,1-\alpha}^\star\ge q_n(u_-)$.
Next, suppose that $u_+<1$. By \eqref{eq:upper_quantile_identity},
\begin{align*}
    &\mathbb P^\star\left(\max_{1\le j\le p}|\hat T_n^\star(x_j)|\le q_n(u_+)\mid X_{1:n}\right)\\
    &\ge\mathbb P\left(\max_{1\le j\le p}|G_n(x_j)|\le q_n(u_+)\right)-C\left\{\frac{B_n^2\log^5(2np)}{n}\right\}^{1/4}=u_+-C\left\{\frac{B_n^2\log^5(2np)}{n}\right\}^{1/4}=1-\alpha.
\end{align*}
Hence $\hat c_{n,1-\alpha}^\star\le q_n(u_+)$. When $u_-=0$ or $u_+=1$, the corresponding bounds follow from $q_n(0)=0$, $q_n(1)=+\infty$, and the nonnegativity of the bootstrap maximum. Therefore,
\begin{align}
    q_n(u_-)\le\hat c_{n,1-\alpha}^{\star}\le q_n(u_+).
    \label{eq:upper_critical_brackets}
\end{align}
By \eqref{eq:upper_finite_bootstrap}, these inequalities hold simultaneously with probability at least
\[
    1-C\left\{\frac{B_n^2\log^5(2np)}{n}\right\}^{1/4},
\]
which proves \eqref{eq:upper_critical_brackets_probability}.

\paragraph{Step 4:}
We now combine Steps~1--3 with \eqref{eq:upper_coverage_identity} to establish the desired coverage error bound.

\subparagraph{Step 4-1:}
We first bound the first term in \cref{eq:upper_error_decomposition} in absolute value. 
By \cref{eq:upper_critical_brackets_probability} in Step~3,
\[
    \mathbb P\left(q_n(u_-)\le\hat c_{n,1-\alpha}^{\star}\le q_n(u_+)\right)
    \ge1-C\left\{\frac{B_n^2\log^5(2np)}{n}\right\}^{1/4}.
\]
Since
\begin{align*}
    &\left\{\max_{1\le j\le p}|\hat T_n(x_j)|\le q_n(u_-)\right\}
    \cap\{\hat c_{n,1-\alpha}^{\star}\ge q_n(u_-)\}
    \subseteq\left\{\max_{1\le j\le p}|\hat T_n(x_j)|\le\hat c_{n,1-\alpha}^{\star}\right\},\\
    &\left\{\max_{1\le j\le p}|\hat T_n(x_j)|\le\hat c_{n,1-\alpha}^{\star}\right\}
    \cap\{\hat c_{n,1-\alpha}^{\star}\le q_n(u_+)\}
    \subseteq\left\{\max_{1\le j\le p}|\hat T_n(x_j)|\le q_n(u_+)\right\},
\end{align*}
taking probabilities gives
\begin{align}
    &\mathbb P\left(\max_{1\le j\le p}|\hat T_n(x_j)|\le q_n(u_-)\right)
    -\mathbb P\left(\hat c_{n,1-\alpha}^{\star}<q_n(u_-)\right)\nonumber\\
    &\quad\le\mathbb P\left(\max_{1\le j\le p}|\hat T_n(x_j)|\le\hat c_{n,1-\alpha}^{\star}\right)\nonumber\\
    &\quad\le\mathbb P\left(\max_{1\le j\le p}|\hat T_n(x_j)|\le q_n(u_+)\right)
    +\mathbb P\left(\hat c_{n,1-\alpha}^{\star}>q_n(u_+)\right).
    \label{eq:upper_grid_probability_inclusions}
\end{align}
Moreover, \cref{eq:upper_critical_brackets_probability} implies
\begin{align*}
    \mathbb P\left(\hat c_{n,1-\alpha}^{\star}<q_n(u_-)\right)
    &\le\mathbb P\left(\{q_n(u_-)\le\hat c_{n,1-\alpha}^{\star}\le q_n(u_+)\}^c\right)
    \le C\left\{\frac{B_n^2\log^5(2np)}{n}\right\}^{1/4},\\
    \mathbb P\left(\hat c_{n,1-\alpha}^{\star}>q_n(u_+)\right)
    &\le\mathbb P\left(\{q_n(u_-)\le\hat c_{n,1-\alpha}^{\star}\le q_n(u_+)\}^c\right)
    \le C\left\{\frac{B_n^2\log^5(2np)}{n}\right\}^{1/4}.
\end{align*}
Substituting these bounds into \cref{eq:upper_grid_probability_inclusions} yields
\begin{align}
    &\mathbb P\left(\max_{1\le j\le p}|\hat T_n(x_j)|\le q_n(u_-)\right)
    -C\left\{\frac{B_n^2\log^5(2np)}{n}\right\}^{1/4}\nonumber\\
    &\quad\le\mathbb P\left(\max_{1\le j\le p}|\hat T_n(x_j)|\le\hat c_{n,1-\alpha}^{\star}\right)\nonumber\\
    &\quad\le\mathbb P\left(\max_{1\le j\le p}|\hat T_n(x_j)|\le q_n(u_+)\right)
    +C\left\{\frac{B_n^2\log^5(2np)}{n}\right\}^{1/4}.
    \label{eq:upper_grid_probability_sandwich}
\end{align}
Applying \cref{eq:upper_finite_studentized_grid} from Step~1 and \cref{eq:upper_quantile_identity} from Step~3, we obtain
\begin{align}
    &u_--2C\left\{\frac{B_n^2\log^5(2np)}{n}\right\}^{1/4}\le\mathbb P\left(\max_{1\le j\le p}|\hat T_n(x_j)|\le\hat c_{n,1-\alpha}^{\star}\right)\le u_++2C\left\{\frac{B_n^2\log^5(2np)}{n}\right\}^{1/4}.
    \label{eq:upper_grid_gaussian_sandwich}
\end{align}
Finally, \cref{eq:upper_quantile_levels} gives
\begin{align}
    &\left|\mathbb P\left(\max_{1\le j\le p}|\hat T_n(x_j)|\le\hat c_{n,1-\alpha}^{\star}\right)-(1-\alpha)\right|
    \le3C\left\{\frac{B_n^2\log^5(2np)}{n}\right\}^{1/4}.
    \label{eq:upper_grid_coverage_bound}
\end{align}

\subparagraph{Step 4-2:}
We next bound the final term in \eqref{eq:upper_error_decomposition};
\[
    \mathbb P\left(\max_{1\le j\le p}|\hat T_n(x_j)|\le\hat c_{n,1-\alpha}^{\star}<\max_{1\le a\le m}|\hat T_n(z_a)|+L_n\right).
\]

Since $\{x_1,\ldots,x_p\}\subset\mathcal Z_n$ and $L_n\ge0$, we have
\begin{align*}
    &\left\{\max_{1\le a\le m}|\hat T_n(z_a)|\le\hat c_{n,1-\alpha}^{\star}-L_n\right\}\subseteq\left\{\max_{1\le j\le p}|\hat T_n(x_j)|\le\hat c_{n,1-\alpha}^{\star}\right\},
\end{align*}
because, whenever the first event occurs,
\begin{align*}
    \max_{1\le j\le p}|\hat T_n(x_j)|
    &\le\max_{1\le a\le m}|\hat T_n(z_a)|\le\hat c_{n,1-\alpha}^{\star}-L_n\le\hat c_{n,1-\alpha}^{\star}.
\end{align*}
Consequently,
\begin{align}
    &\mathbb P\left(\max_{1\le j\le p}|\hat T_n(x_j)|\le\hat c_{n,1-\alpha}^{\star}<\max_{1\le a\le m}|\hat T_n(z_a)|+L_n\right)\nonumber\\
    &=\mathbb P\left(\max_{1\le j\le p}|\hat T_n(x_j)|\le\hat c_{n,1-\alpha}^{\star},\ \max_{1\le a\le m}|\hat T_n(z_a)|>\hat c_{n,1-\alpha}^{\star}-L_n\right)\nonumber\\
    &=\mathbb P\left(\max_{1\le j\le p}|\hat T_n(x_j)|\le\hat c_{n,1-\alpha}^{\star}\right)
    -\mathbb P\left(\max_{1\le j\le p}|\hat T_n(x_j)|\le\hat c_{n,1-\alpha}^{\star},\ \max_{1\le a\le m}|\hat T_n(z_a)|\le\hat c_{n,1-\alpha}^{\star}-L_n\right)\nonumber\\
    &=\mathbb P\left(\max_{1\le j\le p}|\hat T_n(x_j)|\le\hat c_{n,1-\alpha}^{\star}\right)
    -\mathbb P\left(\max_{1\le a\le m}|\hat T_n(z_a)|\le\hat c_{n,1-\alpha}^{\star}-L_n\right),
    \label{eq:upper_miss_reduction}
\end{align}
where the first equality rearranges the threshold inequality, the second takes the complement within the first event, and the last uses the preceding event inclusion.

We next bound the second probability on the right-hand side of \eqref{eq:upper_miss_reduction} from below.
By the lower bound for $\hat c_{n,1-\alpha}^{\star}$ from Step~3,
\begin{align*}
    &\left\{\max_{1\le a\le m}|\hat T_n(z_a)|\le q_n(u_-)-L_n\right\}
    \cap\{\hat c_{n,1-\alpha}^{\star}\ge q_n(u_-)\}\subseteq\left\{\max_{1\le a\le m}|\hat T_n(z_a)|\le\hat c_{n,1-\alpha}^{\star}-L_n\right\}.
\end{align*}
Hence, taking probabilities and using \eqref{eq:upper_critical_brackets_probability},
\begin{align}
    &\mathbb P\left(\max_{1\le a\le m}|\hat T_n(z_a)|\le\hat c_{n,1-\alpha}^{\star}-L_n\right)\nonumber\\
    &\quad\ge\mathbb P\left(\max_{1\le a\le m}|\hat T_n(z_a)|\le q_n(u_-)-L_n\right)
    -C\left\{\frac{B_n^2\log^5(2np)}{n}\right\}^{1/4}.
    \label{eq:upper_auxiliary_quantile_bound}
\end{align}
Recall that $t_n=q_n(u_-)$. Since each $G_n(z_a)$ has unit variance, Gaussian anti-concentration (Example 9 in \citealp{cattaneo2026sharp}) gives
\begin{align*}
    &\mathbb P\left(t_n-L_n<\max_{1\le a\le m}|G_n(z_a)|\le t_n\right)
    \le L_n\{\sqrt{2\log(2m)}+1\}.
\end{align*}
Consequently,
\begin{align*}
    &\mathbb P\left(\max_{1\le a\le m}|G_n(z_a)|\le t_n-L_n\right)\ge\mathbb P\left(\max_{1\le a\le m}|G_n(z_a)|\le t_n\right)-L_n\{\sqrt{2\log(2m)}+1\}.
\end{align*}
Applying the Gaussian approximation from Step~1 at the deterministic threshold $t_n-L_n$ to \eqref{eq:upper_auxiliary_quantile_bound}, we obtain
\begin{align}
    &\mathbb P\left(\max_{1\le a\le m}|\hat T_n(z_a)|\le\hat c_{n,1-\alpha}^{\star}-L_n\right)\nonumber\\
    &\quad\ge\mathbb P\left(\max_{1\le a\le m}|G_n(z_a)|\le t_n\right)
    -L_n\{\sqrt{2\log(2m)}+1\}-C\left\{\frac{B_n^2\log^5\{2n(p+m)\}}{n}\right\}^{1/4}.
    \label{eq:upper_auxiliary_gaussian_lower}
\end{align}
It remains to bound the Gaussian probability on the right-hand side of \eqref{eq:upper_auxiliary_gaussian_lower} from below. 
Since $\{x_1,\ldots,x_p\}\subseteq\mathcal Z_n$, we have
\begin{align*}
    &\left\{\max_{1\le j\le p}|G_n(x_j)|\le t_n\right\}=\left\{\max_{1\le a\le m}|G_n(z_a)|\le t_n\right\}
    \,\dot\cup\,
    \left\{\max_{1\le j\le p}|G_n(x_j)|\le t_n<\max_{1\le a\le m}|G_n(z_a)|\right\}.
\end{align*}
Taking probabilities and using $t_n=q_n(u_-)$ gives
\begin{align}
    &\mathbb P\left(\max_{1\le a\le m}|G_n(z_a)|\le t_n\right)\nonumber\\
    &=\mathbb P\left(\max_{1\le j\le p}|G_n(x_j)|\le t_n\right)
    -\mathbb P\left(\max_{1\le j\le p}|G_n(x_j)|\le t_n<\max_{1\le a\le m}|G_n(z_a)|\right)\nonumber\\
    &=\mathbb P\left(\max_{1\le j\le p}|G_n(x_j)|\le q_n(u_-)\right)
    -\mathbb P\left(\max_{1\le j\le p}|G_n(x_j)|\le t_n<\max_{1\le a\le m}|G_n(z_a)|\right)\nonumber\\
    &=u_--\mathbb P\left(\max_{1\le j\le p}|G_n(x_j)|\le t_n<\max_{1\le a\le m}|G_n(z_a)|\right),
    \label{eq:upper_auxiliary_gaussian_decomposition}
\end{align}
where the last equality follows from \eqref{eq:upper_quantile_identity}.
By \eqref{eq:upper_finite_crossing} from Step~2,
\begin{align*}
    &\mathbb P\left(\max_{1\le j\le p}|G_n(x_j)|\le t_n<\max_{1\le a\le m}|G_n(z_a)|\right)\\
    &\quad\le\frac{e^{-t_n^2/2}}{\pi}\sum_{j=1}^{p-1}\left\{\mathcal L_{n,j}-d_n(x_j,x_{j+1})+H_{t_n}(d_n(x_j,x_{j+1}))\right\}.
\end{align*}
Combining this inequality with \eqref{eq:upper_auxiliary_gaussian_decomposition} and \eqref{eq:upper_auxiliary_gaussian_lower}, and using $u_-\ge1-\alpha-C\{B_n^2\log^5(2np)/n\}^{1/4}$ and $p\le m$, we obtain
\begin{align}
    &\mathbb P\left(\max_{1\le a\le m}|\hat T_n(z_a)|\le\hat c_{n,1-\alpha}^{\star}-L_n\right)\nonumber\\
    &\quad\ge1-\alpha
    -L_n\{\sqrt{2\log(2m)}+1\}
    -3C\left\{\frac{B_n^2\log^5(2nm)}{n}\right\}^{1/4}\nonumber\\
    &\qquad-\frac{e^{-t_n^2/2}}{\pi}\sum_{j=1}^{p-1}\left\{\mathcal L_{n,j}-d_n(x_j,x_{j+1})+H_{t_n}(d_n(x_j,x_{j+1}))\right\}.
    \label{eq:upper_auxiliary_probability_lower}
\end{align}
Finally, substituting \eqref{eq:upper_auxiliary_probability_lower} into \eqref{eq:upper_miss_reduction}, we have
\begin{align}
    &\mathbb P\left(\max_{1\le j\le p}|\hat T_n(x_j)|\le\hat c_{n,1-\alpha}^{\star}<\max_{1\le a\le m}|\hat T_n(z_a)|+L_n\right)\nonumber\\
    &\quad\le\left\{\mathbb P\left(\max_{1\le j\le p}|\hat T_n(x_j)|\le\hat c_{n,1-\alpha}^{\star}\right)-(1-\alpha)\right\}\nonumber\\
    &\qquad+L_n\{\sqrt{2\log(2m)}+1\}
    +3C\left\{\frac{B_n^2\log^5(2nm)}{n}\right\}^{1/4}\nonumber\\
    &\qquad+\frac{e^{-t_n^2/2}}{\pi}
    \sum_{j=1}^{p-1}\left\{\mathcal L_{n,j}-d_n(x_j,x_{j+1})+H_{t_n}(d_n(x_j,x_{j+1}))\right\}.
    \label{eq:upper_auxiliary_miss_bound}
\end{align}

\subparagraph{Step 4-3:}
We finally combine the bounds from Steps 4-1 and 4-2 to establish lower and upper bounds for the coverage probability.

First, by \eqref{eq:upper_coverage_identity}, we have
\begin{align}
    &(1-\alpha)-\mathbb P\left(\sup_{x\in\mathcal X}|\hat T_n(x)|\le\hat c_{n,1-\alpha}^{\star}\right)\nonumber\\
    &=\mathbb P\left(\max_{1\le j\le p}|\hat T_n(x_j)|\le\hat c_{n,1-\alpha}^{\star}<\sup_{x\in\mathcal X}|\hat T_n(x)|\right)
    -\left\{\mathbb P\left(\max_{1\le j\le p}|\hat T_n(x_j)|\le\hat c_{n,1-\alpha}^{\star}\right)-(1-\alpha)\right\}\nonumber\\
    &\le\varepsilon_n+\mathbb P\left(\max_{1\le j\le p}|\hat T_n(x_j)|\le\hat c_{n,1-\alpha}^{\star}<\max_{1\le a\le m}|\hat T_n(z_a)|+L_n\right)\nonumber\\
    &\qquad-\left\{\mathbb P\left(\max_{1\le j\le p}|\hat T_n(x_j)|\le\hat c_{n,1-\alpha}^{\star}\right)-(1-\alpha)\right\}\nonumber\\
    &\le\varepsilon_n
    +\left\{\mathbb P\left(\max_{1\le j\le p}|\hat T_n(x_j)|\le\hat c_{n,1-\alpha}^{\star}\right)-(1-\alpha)\right\}\nonumber\\
    &\qquad+L_n\{\sqrt{2\log(2m)}+1\}
    +3C\left\{\frac{B_n^2\log^5(2nm)}{n}\right\}^{1/4} \nonumber\\
    & \qquad
    +\frac{e^{-t_n^2/2}}{\pi}\sum_{j=1}^{p-1}\left\{\mathcal L_{n,j}-d_n(x_j,x_{j+1})+H_{t_n}(d_n(x_j,x_{j+1}))\right\}\nonumber\\
    &\qquad-\left\{\mathbb P\left(\max_{1\le j\le p}|\hat T_n(x_j)|\le\hat c_{n,1-\alpha}^{\star}\right)-(1-\alpha)\right\}\nonumber\\
    &=\varepsilon_n+L_n\{\sqrt{2\log(2m)}+1\}
    +3C\left\{\frac{B_n^2\log^5(2nm)}{n}\right\}^{1/4}\nonumber\\
    &\qquad+\frac{e^{-t_n^2/2}}{\pi}\sum_{j=1}^{p-1}\left\{\mathcal L_{n,j}-d_n(x_j,x_{j+1})+H_{t_n}(d_n(x_j,x_{j+1}))\right\},
    \label{eq:upper_coverage_lower}
\end{align}
where the first equality follows from \eqref{eq:upper_coverage_identity},
the first inequality follows from \eqref{eq:upper_auxiliary_approximation} and \eqref{eq:upper_auxiliary_event_inclusion}, 
and the second inequality follows from \eqref{eq:upper_auxiliary_miss_bound}. 
The final equality follows by canceling the two identical terms.

For the upper bound, \eqref{eq:upper_coverage_identity} and \eqref{eq:upper_grid_coverage_bound} yield
\begin{align}
    &\mathbb P\left(\sup_{x\in\mathcal X}|\hat T_n(x)|\le\hat c_{n,1-\alpha}^{\star}\right)-(1-\alpha)\nonumber\\
    &\le\mathbb P\left(\max_{1\le j\le p}|\hat T_n(x_j)|\le\hat c_{n,1-\alpha}^{\star}\right)-(1-\alpha)\le3C\left\{\frac{B_n^2\log^5(2np)}{n}\right\}^{1/4}.
    \label{eq:upper_coverage_upper}
\end{align}

From \eqref{eq:upper_coverage_lower} and \eqref{eq:upper_coverage_upper}, since $\mathcal L_{n,j}\ge d_n(x_j,x_{j+1})$ and $H_{t_n}\ge0$,
\begin{align*}
    &\left|\mathbb P\left(\sup_{x\in\mathcal X}|\hat T_n(x)|\le\hat c_{n,1-\alpha}^{\star}\right)-(1-\alpha)\right|\nonumber\\
    & \le \varepsilon_n+L_n\{\sqrt{2\log(2m)}+1\}
    +3C\left\{\frac{B_n^2\log^5(2nm)}{n}\right\}^{1/4}\nonumber\\
    &\qquad+\frac{e^{-t_n^2/2}}{\pi}\sum_{j=1}^{p-1}\left\{\mathcal L_{n,j}-d_n(x_j,x_{j+1})+H_{t_n}(d_n(x_j,x_{j+1}))\right\}.
\end{align*}

\subparagraph{Step 4-4:}
It remains to absorb the auxiliary approximation error $ \varepsilon_n+L_n\{\sqrt{2\log(2m)}+1\}$ into the high-dimensional bootstrap approximation error and to bound the number of auxiliary grid points $m$.
Choose
\[
    L_n=n^{-1/4},
    \qquad K=\left\lceil n^2(1+\mathcal L_n)^2\right\rceil.
\]
By the definition of $\varepsilon_n$ in \eqref{eq:upper_auxiliary_approximation}, we have
\begin{align*}
    \varepsilon_n
    &\le C\left\{\frac{B_n^2\log^5(2nm)}{n}\right\}^{1/4}
    +\frac{Cn(1+\mathcal L_n)^2}{KL_n^2}\\
    &\le C\left\{\frac{B_n^2\log^5(2nm)}{n}\right\}^{1/4}
    +Cn^{-1/2}\le C\left\{\frac{B_n^2\log^5(2nm)}{n}\right\}^{1/4},
\end{align*}
where the last inequality uses $B_n\ge1$ and $n\ge2$.
Also,
\begin{align*}
    L_n\{\sqrt{2\log(2m)}+1\}
    &\le Cn^{-1/4}\log^{1/2}(2nm)\le C\left\{\frac{B_n^2\log^5(2nm)}{n}\right\}^{1/4}.
\end{align*}
Since $m\le p+K$ by \cref{lem:upper_finite_approximation} and $K=\lceil n^2(1+\mathcal L_n)^2\rceil$, we have
\begin{align*}
    m&\le p+\left\lceil n^2(1+\mathcal L_n)^2\right\rceil\le p+2n^2(1+\mathcal L_n)^2\le2pn^2(1+\mathcal L_n)^2,
\end{align*}
where the second inequality uses $n^2(1+\mathcal L_n)^2\ge1$, and the last inequality uses $p\ge2$.
Consequently,
\begin{align*}
    \log(2nm)
    &\le\log\{4n^3p(1+\mathcal L_n)^2\}\\
    & =3\log\{2np(1+\mathcal L_n)\}-\log\{2p^2(1+\mathcal L_n)\}\le3\log\{2np(1+\mathcal L_n)\}.
\end{align*}
Therefore,
\[
    \left\{\frac{B_n^2\log^5(2nm)}{n}\right\}^{1/4}
    \le C\left\{\frac{B_n^2\log^5\{2np(1+\mathcal L_n)\}}{n}\right\}^{1/4}.
\]
This completes the proof.
\end{proof}

\subsubsection{Proof of \cref{lem:gaussian_two_point_crossing}} \label{subsubsec:proof:lem:gaussian_two_point_crossing}

Before the proof of \cref{lem:gaussian_two_point_crossing},  for the reader's convenience, we recall the following identity from Eq.(6) in \cite{Plackett54}.
\begin{lemma}[Plackett's identity]
\label{lem:plackett}
For each $r\in(-1,1)$, let $(U_r,V_r)$ be a centered Gaussian vector with unit marginal variances and correlation $r$. Then, for every $a,b\in\mathbb R$,
\begin{align}
    \frac{\partial}{\partial r}\mathbb P(U_r\le a,V_r\le b)
    &=\frac{1}{2\pi\sqrt{1-r^2}}\exp\left\{-\frac{a^2-2rab+b^2}{2(1-r^2)}\right\}. \qedhere
\end{align}
\end{lemma}

\begin{proof}[Proof of \cref{lem:gaussian_two_point_crossing}]
Suppose first that $0<\theta<\pi$. Let $(U_r,V_r)$ be as in Lemma~\ref{lem:plackett}, and let $Z_1,Z_2$ be independent standard normal variables. 

\paragraph{Step 1 :}
We first establish the following preliminary result for later use.
Note that
\begin{align*}
    (U_r,V_r)&\sim N\left(0,\begin{pmatrix}1&r\\r&1\end{pmatrix}\right),\\
    (Z_1,rZ_1+\sqrt{1-r^2}Z_2)&\sim N\left(0,\begin{pmatrix}1&r\\r&1\end{pmatrix}\right).
\end{align*}
Since a Gaussian distribution is uniquely determined by its mean and covariance matrix, we have
\[
    (U_r,V_r)\stackrel{d}{=}(Z_1,rZ_1+\sqrt{1-r^2}Z_2).
\]
Therefore, since the indicators are bounded by $1$, dominated convergence yields
\begin{align*}
    & \lim_{r\uparrow1}\mathbb P(U_r\le t,V_r\le t) \\
    &=\lim_{r\uparrow1}\mathbb E\left[\mathbf1\{Z_1\le t,\ rZ_1+\sqrt{1-r^2}Z_2\le t\}\right]=\mathbb E[\mathbf1\{Z_1\le t\}]
    =\mathbb P(Z_1\le t)=\Phi(t),
\end{align*}
where $\Phi$ denotes the standard normal distribution function.

\paragraph{Step2 :}
In terms of $\mathbb P(U\le t<V)$, it holds that
\begin{align}
    \mathbb P(U\le t<V)
    &=\Phi(t)-\mathbb P(U\le t,V\le t)\nonumber\\
    &=\lim_{r\uparrow1}\mathbb P(U_r\le t,V_r\le t)-\mathbb P(U_{\cos\theta}\le t,V_{\cos\theta}\le t)\nonumber\\
    &=\lim_{r\uparrow1}\int_{\cos\theta}^r\frac{\partial}{\partial\rho}\mathbb P(U_\rho\le t,V_\rho\le t)\,d\rho\nonumber\\
    &=\frac{1}{2\pi}\int_{\cos\theta}^1\frac{\exp\{-t^2/(1+\rho)\}}{\sqrt{1-\rho^2}}\,d\rho\nonumber\\
    &=\frac{1}{2\pi}\int_0^\theta\exp\left(-\frac{t^2}{1+\cos s}\right)ds\nonumber\\
    &=\frac{e^{-t^2/2}}{2\pi}\int_0^\theta\exp\left\{-\frac{t^2}{2}\tan^2(s/2)\right\}ds\nonumber\\
    &=\frac{e^{-t^2/2}}{2\pi}\{\theta-H_t(\theta)\},
    \label{eq:upper_crossing_derivation}
\end{align}
where the first equality follows from $\{U\le t<V\}=\{U\le t\}\setminus\{U\le t,V\le t\}$. The second uses the preceding limit and $(U,V)\stackrel{d}{=}(U_{\cos\theta},V_{\cos\theta})$. The third follows from the fundamental theorem of calculus, and the fourth from Lemma~\ref{lem:plackett} and monotone convergence. The fifth uses the substitution $\rho=\cos s$, the sixth uses $(1+\cos s)^{-1}=\{1+\tan^2(s/2)\}/2$, and the last follows from the definition of $H_t$.

\paragraph{Step 3:}
It remains to consider $\theta=0$ and $\theta=\pi$. When $\theta=0$, we have $U=V$ almost surely, and both sides of \eqref{eq:upper_bivariate_crossing} are zero. When $\theta=\pi$, we have $(U,V)\stackrel{d}{=}(Z_1,-Z_1)$. Moreover, as $\theta\uparrow\pi$,
\[
    \mathbf1\{Z_1\le t<\cos\theta\,Z_1+\sin\theta\,Z_2\}
    \longrightarrow\mathbf1\{Z_1\le t<-Z_1\}
    \quad\text{almost surely}.
\]
Therefore, from dominated convergence, we have
\begin{align*}
    &\lim_{\theta\uparrow\pi}\mathbb P(Z_1\le t<\cos\theta\,Z_1+\sin\theta\,Z_2)\\
    &=\lim_{\theta\uparrow\pi}\mathbb E\left[\mathbf1\{Z_1\le t<\cos\theta\,Z_1+\sin\theta\,Z_2\}\right]=\mathbb E\left[\mathbf1\{Z_1\le t<-Z_1\}\right]=\mathbb P(Z_1\le t<-Z_1).
\end{align*}
On the other hand, for $0\le s<\pi$,
\[
    0\le\exp\left(-\frac{t^2}{1+\cos s}\right)\le1,
    \qquad
    0\le1-\exp\left\{-\frac{t^2}{2}\tan^2(s/2)\right\}\le1.
\]
Consequently,
\begin{align*}
    0&\le\int_0^\pi\exp\left(-\frac{t^2}{1+\cos s}\right)ds-\int_0^\theta\exp\left(-\frac{t^2}{1+\cos s}\right)ds\le\pi-\theta,\\
    0&\le H_t(\pi)-H_t(\theta) = \int_\theta^\pi
        \left[
            1-\exp\left\{
                -\frac{t^2}{2}\tan^2\left(\frac{s}{2}\right)
            \right\}
        \right]ds\le\pi-\theta.
\end{align*}
Thus, both expressions on the right-hand side of \eqref{eq:upper_bivariate_crossing} are continuous at $\theta=\pi$. 
Taking $\theta\uparrow\pi$ in \eqref{eq:upper_crossing_derivation} gives
\begin{align*}
    \mathbb P(U\le t<V)
    &=\mathbb P(Z_1\le t<-Z_1)\\
    &=\lim_{\theta\uparrow\pi}\mathbb P(Z_1\le t<\cos\theta\,Z_1+\sin\theta\,Z_2)\\
    &=\frac{1}{2\pi}\int_0^\pi\exp\left(-\frac{t^2}{1+\cos s}\right)ds=\frac{e^{-t^2/2}}{2\pi}\{\pi-H_t(\pi)\}.
\end{align*}
This proves \eqref{eq:upper_bivariate_crossing} for $\theta=\pi$. The argument also covers $t=0$, for which $H_0(\theta)=0$.
\end{proof}

\subsubsection{Proof of \cref{lem:upper_finite_approximation}} \label{subsubsec:proof:lem:upper_finite_approximation}

\begin{proof}[Proof of \cref{lem:upper_finite_approximation}]
Our goal is to bound the probability that the supremum over the continuum $\mathcal X$ exceeds the maximum over a deterministic finite set $\mathcal Z_n$ by more than $L_n$.
The proof proceeds in four steps.
In Step 1, we bound the left-hand side of \eqref{eq:upper_finite_approximation} by two tractable terms.
In Steps 2 and 3, we bound these two terms, respectively.
In Step 4, we combine these results to establish the desired probability bound.

\paragraph{Step 1: }
Let $\mathcal Z_n=\{z_1<\cdots<z_m\}$ be any deterministic finite subset of $\mathcal X$ containing $\{x_1,\ldots,x_p\}$.
As noted above, the goal of this step is to reduce the probability of a discretization error exceeding $L_n$ to more tractable probabilities (as obtained in \eqref{eq:finite_approximation_probability_reduction}).
In Step 1-1, we derive the deterministic inequalities needed for this reduction.
In Step 1-2, we use these inequalities to obtain the desired probability bound.

\subparagraph{Step 1-1:}
First observe that
\begin{align}
    &\sup_{x\in\mathcal X}|\hat T_n(x)|-\max_{z\in\mathcal Z_n}|\hat T_n(z)|=\sup_{x\in\mathcal X}\min_{z\in\mathcal Z_n}\{|\hat T_n(x)|-|\hat T_n(z)|\}\le\sup_{x\in\mathcal X}\min_{z\in\mathcal Z_n}|\hat T_n(x)-\hat T_n(z)|,
    \label{eq:finite_approximation_initial_reduction}
\end{align}
where the inequality follows from the reverse triangle inequality.
Thus, to control the discretization error, it suffices to bound the difference between $\hat T_n(x)$ and $\hat T_n(z)$ for $x\in\mathcal X$ and $z\in\mathcal Z_n$.

The definitions of the standardized summands $\{\tilde\psi_{i,h_n}(x)\}_{i=1}^n$ and the variance estimator give
\begin{align}
    T_n(x)&=\frac1{\sqrt n}\sum_{i=1}^n\tilde\psi_{i,h_n}(x), \quad 
    \frac{\hat\sigma_n(x)}{\sigma_n(x)}=\left\{\frac1n\sum_{i=1}^n\left[\tilde\psi_{i,h_n}(x)-\frac1n\sum_{b=1}^n\tilde\psi_{b,h_n}(x)\right]^2\right\}^{1/2},
    \label{eq:finite_approximation_studentization}
\end{align}
and $\hat T_n(x)=T_n(x)/\{\hat\sigma_n(x)/\sigma_n(x)\}$.
For any $x,z\in\mathcal X$, the the triangle inequality gives
\begin{align}
    |\hat T_n(x)-\hat T_n(z)|
    &=\left|\frac{T_n(x)-T_n(z)}{\hat\sigma_n(x)/\sigma_n(x)}+T_n(z)\left\{\frac1{\hat\sigma_n(x)/\sigma_n(x)}-\frac1{\hat\sigma_n(z)/\sigma_n(z)}\right\}\right|\nonumber\\
    &\le\frac{|T_n(x)-T_n(z)|}{\hat\sigma_n(x)/\sigma_n(x)}+\frac{|T_n(z)|\left|\hat\sigma_n(x)/\sigma_n(x)-\hat\sigma_n(z)/\sigma_n(z)\right|}{\{\hat\sigma_n(x)/\sigma_n(x)\}\{\hat\sigma_n(z)/\sigma_n(z)\}}.
    \label{eq:finite_approximation_studentized_decomposition}
\end{align}
To bound the right-hand side of \eqref{eq:finite_approximation_studentized_decomposition}, for fixed $x,z\in\mathcal X$, define
\[
    \begin{aligned}
        \mathrm{I}&:=|T_n(x)-T_n(z)|,~~
        \mathrm{II}:=|T_n(z)|,~~
        \mathrm{III}:=\left|\frac{\hat\sigma_n(x)}{\sigma_n(x)}-\frac{\hat\sigma_n(z)}{\sigma_n(z)}\right|,~~
        \mathrm{IV}:=\min\left\{\frac{\hat\sigma_n(x)}{\sigma_n(x)},\frac{\hat\sigma_n(z)}{\sigma_n(z)}\right\}.
    \end{aligned}
\]
Then, whenever $\mathrm{IV}>0$, \eqref{eq:finite_approximation_studentized_decomposition} gives
\[
    |\hat T_n(x)-\hat T_n(z)|
    \le\frac{\mathrm{I}}{\mathrm{IV}}
    +\frac{\mathrm{II} \times \mathrm{III}}{\mathrm{IV}^2}.
\]
We bound $\mathrm{I}$, $\mathrm{II}$, and $\mathrm{III}$ from above and $\mathrm{IV}$ from below.

First, Cauchy--Schwarz inequality gives
\begin{align}
    \mathrm{I}
    &=\left|\frac1{\sqrt n}\sum_{i=1}^n\{\tilde\psi_{i,h_n}(x)-\tilde\psi_{i,h_n}(z)\}\right|\le\sqrt n\left\{\frac1n\sum_{i=1}^n[\tilde\psi_{i,h_n}(x)-\tilde\psi_{i,h_n}(z)]^2\right\}^{1/2}.
    \label{eq:finite_approximation_numerator_difference}
\end{align}

Similarly, Cauchy--Schwarz inequality gives
\begin{equation}
    \mathrm{II}\le\sqrt n\left\{\frac1n\sum_{i=1}^n\tilde\psi_{i,h_n}(z)^2\right\}^{1/2}.
    \label{eq:finite_approximation_numerator_bound}
\end{equation}

Next, we bound $\mathrm{III}$.
By \eqref{eq:finite_approximation_studentization} and the reverse triangle inequality,
\begin{align}
    \mathrm{III}
    &=\left|\left\{\frac1n\sum_{i=1}^n\left[\tilde\psi_{i,h_n}(x)-\frac1n\sum_{b=1}^n\tilde\psi_{b,h_n}(x)\right]^2\right\}^{1/2}-\left\{\frac1n\sum_{i=1}^n\left[\tilde\psi_{i,h_n}(z)-\frac1n\sum_{b=1}^n\tilde\psi_{b,h_n}(z)\right]^2\right\}^{1/2}\right|\nonumber\\
    &\le\left\{\frac1n\sum_{i=1}^n\left[\tilde\psi_{i,h_n}(x)-\tilde\psi_{i,h_n}(z)-\frac1n\sum_{b=1}^n\{\tilde\psi_{b,h_n}(x)-\tilde\psi_{b,h_n}(z)\}\right]^2\right\}^{1/2}\nonumber\\
    &=\left\{\frac1n\sum_{i=1}^n[\tilde\psi_{i,h_n}(x)-\tilde\psi_{i,h_n}(z)]^2-\left[\frac1n\sum_{i=1}^n\{\tilde\psi_{i,h_n}(x)-\tilde\psi_{i,h_n}(z)\}\right]^2\right\}^{1/2}\nonumber\\
    &\le\left\{\frac1n\sum_{i=1}^n[\tilde\psi_{i,h_n}(x)-\tilde\psi_{i,h_n}(z)]^2\right\}^{1/2}.
    \label{eq:finite_approximation_scale_difference}
\end{align}

It remains to bound $\mathrm{IV}$ from below.
By the definitions of $\mathrm{III}$ and $\mathrm{IV}$,
\begin{align}
    \mathrm{IV}
    &=\min\left\{\frac{\hat\sigma_n(x)}{\sigma_n(x)},\frac{\hat\sigma_n(z)}{\sigma_n(z)}\right\}\nonumber\\
    &\ge\frac{\hat\sigma_n(z)}{\sigma_n(z)}-\left|\frac{\hat\sigma_n(x)}{\sigma_n(x)}-\frac{\hat\sigma_n(z)}{\sigma_n(z)}\right|\nonumber\\
    &=\frac{\hat\sigma_n(z)}{\sigma_n(z)}-\mathrm{III}\ge\frac{\hat\sigma_n(z)}{\sigma_n(z)}-\left\{\frac1n\sum_{i=1}^n[\tilde\psi_{i,h_n}(x)-\tilde\psi_{i,h_n}(z)]^2\right\}^{1/2},
    \label{eq:finite_approximation_IV_bound}
\end{align}
where the first inequality follows from $\min\{a,b\}\ge b-|a-b|$, and the last follows from \eqref{eq:finite_approximation_scale_difference}.

\subparagraph{Step 1-2:}
We now use the bounds for $\mathrm{I}$--$\mathrm{IV}$ obtained in Step 1-1 to identify sufficient conditions under which the discretization error does not exceed $L_n$.

Fix a universal constant $C_0\ge2$ and let $\eta\in(0,(2C_0)^{-1}]$.
Suppose that
\begin{align}
    &\max_{z\in\mathcal Z_n}\left\{\frac1n\sum_{i=1}^n\tilde\psi_{i,h_n}(z)^2\right\}^{1/2}\le C_0,\qquad
    \min_{z\in\mathcal Z_n}\frac{\hat\sigma_n(z)}{\sigma_n(z)}\ge C_0^{-1},\nonumber\\
    &\sup_{x\in\mathcal X}\min_{z\in\mathcal Z_n}\left\{\frac1n\sum_{i=1}^n[\tilde\psi_{i,h_n}(x)-\tilde\psi_{i,h_n}(z)]^2\right\}^{1/2}\le\eta.
    \label{eq:finite_approximation_good_bounds}
\end{align}
For every $x\in\mathcal X$, choose $z\in\mathcal Z_n$ attaining the minimum in the second line of \eqref{eq:finite_approximation_good_bounds}.
Then, by \eqref{eq:finite_approximation_numerator_difference}, \eqref{eq:finite_approximation_numerator_bound}, and \eqref{eq:finite_approximation_scale_difference}, we have
\begin{align*}
    \mathrm{I}&\le\sqrt n\,\eta, \quad
    \mathrm{II}\le C_0\sqrt n, \quad
    \mathrm{III}\le\eta.
\end{align*}
Also, \eqref{eq:finite_approximation_IV_bound} and \eqref{eq:finite_approximation_good_bounds} imply
\begin{align}
    \mathrm{IV}
    &\ge\frac{\hat\sigma_n(z)}{\sigma_n(z)}
    -\left\{\frac1n\sum_{i=1}^n[\tilde\psi_{i,h_n}(x)-\tilde\psi_{i,h_n}(z)]^2\right\}^{1/2}\ge C_0^{-1}-\eta
    \ge(2C_0)^{-1}.
    \label{eq:finite_approximation_denominator}
\end{align}
The last inequality uses $\eta\le(2C_0)^{-1}$.
Combining these bounds with \eqref{eq:finite_approximation_studentized_decomposition}, we obtain
\begin{align}
    |\hat T_n(x)-\hat T_n(z)|
    &\le\frac{\mathrm{I}}{\mathrm{IV}}+\frac{\mathrm{II}\,\mathrm{III}}{\mathrm{IV}^2}\le 2C_0\sqrt n\,\eta+4C_0^3\sqrt n\,\eta=2(C_0+2C_0^3)\sqrt n\,\eta.
    \label{eq:finite_approximation_studentized_difference}
\end{align}
We therefore choose
\begin{equation}
    \eta:=\frac{L_n}{2(C_0+2C_0^3)\sqrt n},
    \label{eq:finite_approximation_eta}
\end{equation}
then, \eqref{eq:finite_approximation_studentized_difference} yields
\[
    |\hat T_n(x)-\hat T_n(z)|\le L_n,
\]
and \eqref{eq:finite_approximation_initial_reduction} implies
\begin{equation}
    \sup_{x\in\mathcal X}|\hat T_n(x)|-\max_{z\in\mathcal Z_n}|\hat T_n(z)|\le L_n.
    \label{eq:finite_approximation_gap_implication}
\end{equation}

Consequently, if the discretization error exceeds $L_n$, at least one of the conditions in \eqref{eq:finite_approximation_good_bounds} must fail.
Taking probabilities and applying the union bound, we obtain
\begin{align}
    &\mathbb P\left(\sup_{x\in\mathcal X}|\hat T_n(x)|-\max_{z\in\mathcal Z_n}|\hat T_n(z)|>L_n\right)\nonumber\\
    &\le\mathbb P\left(
    \left\{\max_{z\in\mathcal Z_n}\left\{\frac1n\sum_{i=1}^n\tilde\psi_{i,h_n}(z)^2\right\}^{1/2}>C_0\right\}
    \cup
    \left\{\min_{z\in\mathcal Z_n}\frac{\hat\sigma_n(z)}{\sigma_n(z)}<C_0^{-1}\right\}
    \right)\nonumber\\
    &\quad+\mathbb P\left(
    \sup_{x\in\mathcal X}\min_{z\in\mathcal Z_n}
    \left\{\frac1n\sum_{i=1}^n[\tilde\psi_{i,h_n}(x)-\tilde\psi_{i,h_n}(z)]^2\right\}^{1/2}>\eta
    \right).
    \label{eq:finite_approximation_probability_reduction}
\end{align}
We bound the first probability in Step 2 and construct a deterministic finite set $\mathcal Z_n$ to bound the second probability in Step 3.

\paragraph{Step 2:}
Our goal is to bound the first probability on the right-hand side of \eqref{eq:finite_approximation_probability_reduction}.
We first reduce this probability to two probabilities involving the estimated variances and the sample means on $\mathcal Z_n$.

For every $z\in\mathcal Z_n$, expanding the centered square gives
\begin{align*}
    & \frac1n\sum_{i=1}^n\tilde\psi_{i,h_n}(z)^2 \\
    &=\frac1n\sum_{i=1}^n\left[\tilde\psi_{i,h_n}(z)-\frac1n\sum_{b=1}^n\tilde\psi_{b,h_n}(z)\right]^2+\left\{\frac1n\sum_{i=1}^n\tilde\psi_{i,h_n}(z)\right\}^2
    =\frac{\hat\sigma_n^2(z)}{\sigma_n^2(z)}+\left\{\frac1n\sum_{i=1}^n\tilde\psi_{i,h_n}(z)\right\}^2,
\end{align*}
where the second equality follows from \eqref{eq:finite_approximation_studentization}.
Taking the maximum over $z\in\mathcal Z_n$ and using the monotonicity of the square root, we obtain
\begin{align}
    &\max_{z\in\mathcal Z_n}\left\{\frac1n\sum_{i=1}^n\tilde\psi_{i,h_n}(z)^2\right\}^{1/2}\le\left[\max_{z\in\mathcal Z_n}\frac{\hat\sigma_n^2(z)}{\sigma_n^2(z)}+\left\{\max_{z\in\mathcal Z_n}\left|\frac1n\sum_{i=1}^n\tilde\psi_{i,h_n}(z)\right|\right\}^2\right]^{1/2},
    \label{eq:finite_approximation_second_moment_bound}
\end{align}
where the inequality bounds the maximum of a sum by the sum of the maxima.
Moreover, since the relative estimated standard deviations are nonnegative,
\begin{equation}
    \min_{z\in\mathcal Z_n}\frac{\hat\sigma_n(z)}{\sigma_n(z)}=\left\{\min_{z\in\mathcal Z_n}\frac{\hat\sigma_n^2(z)}{\sigma_n^2(z)}\right\}^{1/2}.
    \label{eq:finite_approximation_minimum_variance}
\end{equation}

We use these identities to identify sufficient conditions under which neither event in the first probability of \eqref{eq:finite_approximation_probability_reduction} occurs.
Suppose that
\begin{equation}
    \max_{z\in\mathcal Z_n}\left|\frac{\hat\sigma_n^2(z)}{\sigma_n^2(z)}-1\right|\le\frac12,
    \qquad
    \max_{z\in\mathcal Z_n}\left|\frac1n\sum_{i=1}^n\tilde\psi_{i,h_n}(z)\right|\le1.
    \label{eq:finite_approximation_finite_targets}
\end{equation}
Then
\begin{align*}
    \max_{z\in\mathcal Z_n}\left\{\frac1n\sum_{i=1}^n\tilde\psi_{i,h_n}(z)^2\right\}^{1/2}
    &\le\left(\frac32+1\right)^{1/2}\le C_0, \quad 
    \min_{z\in\mathcal Z_n}\frac{\hat\sigma_n(z)}{\sigma_n(z)}
    \ge\left(\frac12\right)^{1/2}\ge C_0^{-1},
\end{align*}
where the first inequalities follow from \eqref{eq:finite_approximation_second_moment_bound}, \eqref{eq:finite_approximation_minimum_variance}, and \eqref{eq:finite_approximation_finite_targets}, and the last inequalities use $C_0\ge2$.
Taking the contrapositive of this implication and applying the union bound yields
\begin{align}
    &\mathbb P\left(\left\{\max_{z\in\mathcal Z_n}\left\{\frac1n\sum_{i=1}^n\tilde\psi_{i,h_n}(z)^2\right\}^{1/2}>C_0\right\}\cup\left\{\min_{z\in\mathcal Z_n}\frac{\hat\sigma_n(z)}{\sigma_n(z)}<C_0^{-1}\right\}\right)\nonumber\\
    &\quad\le\mathbb P\left(\max_{z\in\mathcal Z_n}\left|\frac{\hat\sigma_n^2(z)}{\sigma_n^2(z)}-1\right|>\frac12\right)
    +\mathbb P\left(\max_{z\in\mathcal Z_n}\left|\frac1n\sum_{i=1}^n\tilde\psi_{i,h_n}(z)\right|>1\right).
    \label{eq:finite_approximation_finite_reduction}
\end{align}
We bound the first probability in Step 2-1 and the second probability in Step 2-2, and combine the bounds in Step 2-3.

\subparagraph{Step 2-1:}
We bound the first probability on the right-hand side of \eqref{eq:finite_approximation_finite_reduction}.

By Remark~\ref{rmk:small_error_reduction}, it suffices to prove \eqref{eq:finite_approximation_finite_event} when $n\ge n_0$ and $B_n^2\log^5(2nm)/n\le c$, where $n_0$ is the threshold in Lemma~4.2 of \citet{CCKK22} and $c>0$ is a sufficiently small universal constant.
By Assumption~\ref{ass:upper_general}(iii), the standardized summands satisfy Conditions E and M of \citet{CCKK22}.
Applying Eq.~(36) in Lemma~4.2 of \citet{CCKK22} and using \eqref{eq:finite_approximation_studentization}, we obtain, with probability at least $1-n^{-1}-3\{B_n^2\log^3(2nm)/n\}^{1/2}$,
\begin{equation}
    \max_{z\in\mathcal Z_n}\left|\frac{\hat\sigma_n^2(z)}{\sigma_n^2(z)}-1\right|
    \le C\left\{\frac{B_n^2\log(2nm)}n\right\}^{1/2}\le C\sqrt c,
    \label{eq:finite_approximation_variance_deviation}
\end{equation}
where the last inequality uses $B_n^2\log^5(2nm)/n\le c$ and $\log(2nm)\ge1$.
Choosing $c$ sufficiently small that the constant $C$ in \eqref{eq:finite_approximation_variance_deviation} satisfies $C\sqrt c\le1/2$, we obtain
\begin{equation}
    \mathbb P\left(\max_{z\in\mathcal Z_n}\left|\frac{\hat\sigma_n^2(z)}{\sigma_n^2(z)}-1\right|>\frac12\right)
    \le\frac1n+3\left\{\frac{B_n^2\log^3(2nm)}n\right\}^{1/2},
    \label{eq:finite_approximation_finite_variances}
\end{equation}
where the inequality follows from \eqref{eq:finite_approximation_variance_deviation}.

\subparagraph{Step 2-2:}
We next bound the second probability on the right-hand side of \eqref{eq:finite_approximation_finite_reduction}.

To apply Bernstein's inequality, we first verify its moment condition.
Since $u^r\le r!e^u$ for every $u\ge0$ and integer $r\ge2$, the definition of the $\psi_1$-norm and Assumption~\ref{ass:upper_general}(iii) give, for every $z\in\mathcal Z_n$,
\begin{align*}
    \sum_{i=1}^n\mathbb E|\tilde\psi_{i,h_n}(z)|^r
    &\le r!B_n^r\sum_{i=1}^n\mathbb E\exp\left\{\frac{|\tilde\psi_{i,h_n}(z)|}{B_n}\right\}\le2nr!B_n^r=\frac{r!}{2}(4nB_n^2)B_n^{r-2}.
\end{align*}
This verifies the moment condition for Bernstein of the independent centered summands $\{\tilde\psi_{i,h_n}(z)\}_{i=1}^n$.
Applying Bernstein's inequality at threshold $n$ and taking the union bound over $\mathcal Z_n$, we obtain
\begin{align*}
    &\mathbb P\left(\max_{z\in\mathcal Z_n}\left|\frac1n\sum_{i=1}^n\tilde\psi_{i,h_n}(z)\right|>1\right)\le2m\exp\left\{-\frac{n}{2(4B_n^2+B_n)}\right\}.
\end{align*}
To bound the right-hand side, observe that
\[
    \frac{n}{2(4B_n^2+B_n)}
    \ge\frac{n}{10B_n^2}
    \ge\frac{\log^5(2nm)}{10c}
    \ge\log(2nm),
\]
where the first inequality uses $B_n\ge1$, the second follows from $B_n^2\log^5(2nm)/n\le c$, and the last holds for sufficiently small $c>0$ since $\log(2nm)\ge1$.
Therefore,
\begin{align}
    &\mathbb P\left(\max_{z\in\mathcal Z_n}\left|\frac1n\sum_{i=1}^n\tilde\psi_{i,h_n}(z)\right|>1\right)\le2m\exp\{-\log(2nm)\}=\frac1n.
    \label{eq:finite_approximation_finite_means}
\end{align}

\subparagraph{Step 2-3: }
Substituting \eqref{eq:finite_approximation_finite_variances} and \eqref{eq:finite_approximation_finite_means} into \eqref{eq:finite_approximation_finite_reduction}, we obtain
\begin{align}
    &\mathbb P\left(\left\{\max_{z\in\mathcal Z_n}\left\{\frac1n\sum_{i=1}^n\tilde\psi_{i,h_n}(z)^2\right\}^{1/2}>C_0\right\}\cup\left\{\min_{z\in\mathcal Z_n}\frac{\hat\sigma_n(z)}{\sigma_n(z)}<C_0^{-1}\right\}\right)\nonumber\\
    &\le\frac2n+3\left\{\frac{B_n^2\log^3(2nm)}n\right\}^{1/2}\nonumber\\
    &\le C\left\{\frac{B_n^2\log^5(2nm)}n\right\}^{1/4},
    \label{eq:finite_approximation_finite_event}
\end{align}
where the last inequality follows from
\begin{align*}
    \frac1n\le\left\{\frac{B_n^2\log^5(2nm)}n\right\}^{1/4},
\end{align*}
and
\begin{align*}
    \left\{\frac{B_n^2\log^3(2nm)}n\right\}^{1/2}
    =\left\{\frac{B_n^2\log^5(2nm)}n\right\}^{1/4}
    \left\{\frac{B_n^2\log(2nm)}n\right\}^{1/4}
    \le\left\{\frac{B_n^2\log^5(2nm)}n\right\}^{1/4},
\end{align*}
which hold because $B_n\ge1$, $\log(2nm)\ge1$, and $B_n^2\log^5(2nm)/n\le c\le1$.

\paragraph{Step 3:}
Fix an integer $K\ge1$.
To bound the second probability in \eqref{eq:finite_approximation_probability_reduction} by Markov's inequality, we construct a deterministic finite set $\mathcal Z_n$ containing $\{x_1,\ldots,x_p\}$, with $m\le p+K$, such that
\[
    \mathbb E\left[\sup_{x\in\mathcal X}\min_{z\in\mathcal Z_n}\frac1n\sum_{i=1}^n[\tilde\psi_{i,h_n}(x)-\tilde\psi_{i,h_n}(z)]^2\right]\le\frac{(1+\mathcal L_n)^2}{K}.
\]
In Step 3-1, we derive a sufficient condition for this bound.
In Step 3-2, we construct $\mathcal Z_n$ to satisfy that condition.
In Step 3-3, we combine the results to obtain the desired bound.

\subparagraph{Step 3-1:}
By Assumption~\ref{ass:upper_general}(iv), choose a deterministic countable set $\widetilde{\mathcal X}_n\subseteq\mathcal X$ containing $\{x_1,\ldots,x_p\}$ such that every $\tilde\psi_{h_n}(\cdot,x)$ is the pointwise limit of a sequence of functions indexed by $\widetilde{\mathcal X}_n$.
See also Assumption~(A1) in \cite{chernozhukov2014gaussian} and Example~2.3.4 in \cite{van1996weak}.
To select at most $K$ additional points, consider any deterministic pairwise disjoint subsets $\mathcal X_{n,1},\ldots,\mathcal X_{n,K}$ whose union is $\widetilde{\mathcal X}_n$; some subsets may be empty.
For each $k\in\{1,\ldots,K\}$ such that $\mathcal X_{n,k}\ne\varnothing$, choose a fixed representative $z_k^0\in\mathcal X_{n,k}$ and define
\begin{equation}
    \mathcal Z_n:=\{x_1,\ldots,x_p\}\cup\{z_k^0:1\le k\le K,\ \mathcal X_{n,k}\ne\varnothing\},\qquad m:=|\mathcal Z_n|\le p+K.
    \label{eq:finite_approximation_grid_construction}
\end{equation}
The bound $m\le p+K$ holds because we add at most one representative from each of the $K$ subsets.
We first bound the approximation error for arbitrary subsets $\mathcal X_{n,k}$ and representatives $z_k^0$, and specify the subsets $\mathcal X_{n,k}$ in Step 3-2.

For any such partition and representative points, we have
\begin{align}
    &\mathbb E\left[\sup_{x\in\mathcal X}\min_{z\in\mathcal Z_n}\frac1n\sum_{i=1}^n[\tilde\psi_{i,h_n}(x)-\tilde\psi_{i,h_n}(z)]^2\right]\nonumber\\
    &=\mathbb E\left[\sup_{x\in\widetilde{\mathcal X}_n}\min_{z\in\mathcal Z_n}\frac1n\sum_{i=1}^n[\tilde\psi_{i,h_n}(x)-\tilde\psi_{i,h_n}(z)]^2\right]\nonumber\\
    &=\mathbb E\left[\max_{\substack{1\le k\le K\\\mathcal X_{n,k}\ne\varnothing}}\sup_{x\in\mathcal X_{n,k}}\min_{z\in\mathcal Z_n}\frac1n\sum_{i=1}^n[\tilde\psi_{i,h_n}(x)-\tilde\psi_{i,h_n}(z)]^2\right]\nonumber\\
    &\le\mathbb E\left[\max_{\substack{1\le k\le K\\\mathcal X_{n,k}\ne\varnothing}}\sup_{x\in\mathcal X_{n,k}}\frac1n\sum_{i=1}^n[\tilde\psi_{i,h_n}(x)-\tilde\psi_{i,h_n}(z_k^0)]^2\right]\nonumber\\
    &\le\sum_{\substack{1\le k\le K\\\mathcal X_{n,k}\ne\varnothing}}\mathbb E\left[\sup_{x\in\mathcal X_{n,k}}\frac1n\sum_{i=1}^n[\tilde\psi_{i,h_n}(x)-\tilde\psi_{i,h_n}(z_k^0)]^2\right],
    \label{eq:finite_approximation_partition_reduction}
\end{align}
where the first equality follows from Assumption~\ref{ass:upper_general}(iv) and continuity of the minimum over the finite set $\mathcal Z_n$.
The second equality uses $\widetilde{\mathcal X}_n=\bigcup_{k=1}^K\mathcal X_{n,k}$.
The first inequality follows from $z_k^0\in\mathcal Z_n$ for every nonempty $\mathcal X_{n,k}$.
The last inequality bounds the maximum of nonnegative quantities by their sum.

Since at most $K$ subsets are nonempty, \eqref{eq:finite_approximation_partition_reduction} shows that it suffices to establish
\[
    \mathbb E\left[\sup_{x\in\mathcal X_{n,k}}\frac1n\sum_{i=1}^n[\tilde\psi_{i,h_n}(x)-\tilde\psi_{i,h_n}(z_k^0)]^2\right]\le\frac{(1+\mathcal L_n)^2}{K^2}
\]
for every nonempty $\mathcal X_{n,k}$.
To bound the expected supremum over $\mathcal X_{n,k}$, we first consider its finite subsets.
Fix $k\in\{1,\ldots,K\}$ such that $\mathcal X_{n,k}\ne\varnothing$, and list the elements of any finite subset containing $z_k^0$ as $u_0<\cdots<u_r$, with $u_b=z_k^0$ for some $b\in\{0,\ldots,r\}$.
For every $a\in\{0,\ldots,r\}$ and $i\in[n]$, adding and subtracting the values at the intermediate points gives
\[
    [\tilde\psi_{i,h_n}(u_a)-\tilde\psi_{i,h_n}(u_b)]^2
    =\left[\sum_{\ell=\min\{a,b\}+1}^{\max\{a,b\}}\{\tilde\psi_{i,h_n}(u_\ell)-\tilde\psi_{i,h_n}(u_{\ell-1})\}\right]^2.
\]
Applying the triangle inequality and taking the maximum over $a\in\{0,\ldots,r\}$, we obtain
\begin{align}
    &\max_{0\le a\le r}\left\{\frac1n\sum_{i=1}^n[\tilde\psi_{i,h_n}(u_a)-\tilde\psi_{i,h_n}(u_b)]^2\right\}^{1/2}\nonumber\\
    &\le\max_{0\le a\le r}\sum_{\ell=\min\{a,b\}+1}^{\max\{a,b\}}\left\{\frac1n\sum_{i=1}^n[\tilde\psi_{i,h_n}(u_\ell)-\tilde\psi_{i,h_n}(u_{\ell-1})]^2\right\}^{1/2}\nonumber\\
    &\le\max_{0\le a\le r}\sum_{\ell=1}^r\left\{\frac1n\sum_{i=1}^n[\tilde\psi_{i,h_n}(u_\ell)-\tilde\psi_{i,h_n}(u_{\ell-1})]^2\right\}^{1/2}\nonumber\\
    &=\sum_{\ell=1}^r\left\{\frac1n\sum_{i=1}^n[\tilde\psi_{i,h_n}(u_\ell)-\tilde\psi_{i,h_n}(u_{\ell-1})]^2\right\}^{1/2}.
    \label{eq:finite_approximation_pathwise}
\end{align}
The second inequality follows because the summation range is contained in $\{1,\ldots,r\}$ and all summands are nonnegative.
The last equality holds because the resulting sum does not depend on $a$.
Then, taking the $L^2$ norm of both sides of \eqref{eq:finite_approximation_pathwise} gives
\begin{align}
    &\left\{\mathbb E\left[\max_{0\le a\le r}\frac1n\sum_{i=1}^n[\tilde\psi_{i,h_n}(u_a)-\tilde\psi_{i,h_n}(u_b)]^2\right]\right\}^{1/2}\nonumber\\
    &\le\left\{\mathbb E\left[\left(\sum_{\ell=1}^r\left\{\frac1n\sum_{i=1}^n[\tilde\psi_{i,h_n}(u_\ell)-\tilde\psi_{i,h_n}(u_{\ell-1})]^2\right\}^{1/2}\right)^2\right]\right\}^{1/2}\nonumber\\
    &\le\sum_{\ell=1}^r\left\{\mathbb E\left[\frac1n\sum_{i=1}^n[\tilde\psi_{i,h_n}(u_\ell)-\tilde\psi_{i,h_n}(u_{\ell-1})]^2\right]\right\}^{1/2},
    \label{eq:finite_approximation_minkowski}
\end{align}
where the second follows from Minkowski's inequality.

To bound the expectations on the right-hand side of \eqref{eq:finite_approximation_minkowski}, for every $x,y\in\mathcal X$, we have
\begin{align}
    &\mathbb E\left[\frac1n\sum_{i=1}^n\{\tilde\psi_{i,h_n}(x)-\tilde\psi_{i,h_n}(y)\}^2\right]\nonumber\\
    &=\mathbb E\left[\tilde\psi_{1,h_n}(x)^2-2\tilde\psi_{1,h_n}(x)\tilde\psi_{1,h_n}(y)+\tilde\psi_{1,h_n}(y)^2\right]\nonumber\\
    &=2-2\mathbb E[\tilde\psi_{1,h_n}(x)\tilde\psi_{1,h_n}(y)]
    =2-2\cos d_n(x,y)\le d_n(x,y)^2,
    \label{eq:finite_approximation_increment}
\end{align}
where the first equality uses identical distributions and expands the square, the second uses unit second moments, and the third follows from
\[
    \mathbb E[\tilde\psi_{1,h_n}(x)\tilde\psi_{1,h_n}(y)]=\mathbb E[T_n(x)T_n(y)]=\mathbb E[G_n(x)G_n(y)]=\cos d_n(x,y).
\]
Here the first equality uses \eqref{eq:finite_approximation_studentization}, independence, centering, and identical distributions, the second uses covariance matching, and the last uses the definition of $d_n$ and the fact that $\arccos$ is the inverse of $\cos$.
The inequality in \eqref{eq:finite_approximation_increment} uses $2(1-\cos t)\le t^2$.
Substituting \eqref{eq:finite_approximation_increment} into \eqref{eq:finite_approximation_minkowski} gives
\begin{equation}
    \left\{\mathbb E\left[\max_{0\le a\le r}\frac1n\sum_{i=1}^n[\tilde\psi_{i,h_n}(u_a)-\tilde\psi_{i,h_n}(u_b)]^2\right]\right\}^{1/2}\le\sum_{\ell=1}^r d_n(u_{\ell-1},u_\ell).
    \label{eq:finite_approximation_finite_subset}
\end{equation}
Thus, it suffices to choose the subsets $\mathcal X_{n,k}$ so that, for every $k\in\{1,\ldots,K\}$ and every finite increasing sequence $u_0<\cdots<u_r$ in $\mathcal X_{n,k}$,
\[
    \sum_{\ell=1}^r d_n(u_{\ell-1},u_\ell)\le\frac{1+\mathcal L_n}{K}.
\]

\subparagraph{Step 3-2:}
We now construct a partition of $\widetilde{\mathcal X}_n$ satisfying the condition derived in Step 3-1.
The idea is to divide the cumulative spherical path length into $K$ parts.
Define
\[
    s_n(x):=\sup_{r\ge0}\sup_{\ul x=u_0<\cdots<u_r=x}\sum_{a=1}^r d_n(u_{a-1},u_a),\qquad x\in\mathcal X,
\]
where $r=0$ corresponds to $x=\ul x$, with the empty sum equal to zero.

By the definition of $\mathcal L_n$, we have $0\le s_n(x)\le\mathcal L_n$ for every $x\in\mathcal X$.
Moreover, for any $x,y\in\mathcal X$ with $x<y$, extending a partition of $[\ul x,x]$ to $y$ increases its spherical path length by $d_n(x,y)$.
Hence,
\begin{align*}
    s_n(y)
    &\ge\sup_{r\ge0}\sup_{\ul x=u_0<\cdots<u_r=x}\left\{\sum_{a=1}^r d_n(u_{a-1},u_a)+d_n(x,y)\right\}=s_n(x)+d_n(x,y),
\end{align*}
where the inequality follows from the definition of $s_n(y)$, and the equality follows from the definition of $s_n(x)$.
We therefore have
\begin{equation}
    0\le s_n(x)\le\mathcal L_n\quad(x\in\mathcal X),\qquad
    d_n(x,y)\le s_n(y)-s_n(x)\quad(x,y\in\mathcal X,\ x<y).
    \label{eq:finite_approximation_cumulative}
\end{equation}

By \eqref{eq:finite_approximation_cumulative}, the spherical distance between two points is bounded by the difference in their cumulative spherical path lengths.
We therefore divide $[0,1+\mathcal L_n)$ into $K$ equal intervals and define
\begin{equation}
    \mathcal X_{n,k}:=\left\{x\in\widetilde{\mathcal X}_n:\frac{(k-1)(1+\mathcal L_n)}K\le s_n(x)<\frac{k(1+\mathcal L_n)}K\right\},\qquad k=1,\ldots,K.
    \label{eq:finite_approximation_partition}
\end{equation}
Since $0\le s_n(x)\le\mathcal L_n<1+\mathcal L_n$, these subsets form a partition of $\widetilde{\mathcal X}_n$.

To verify that this partition satisfies the condition derived in Step 3-1, consider any finite increasing sequence $u_0<\cdots<u_r$ in a nonempty $\mathcal X_{n,k}$.
By \eqref{eq:finite_approximation_cumulative},
\begin{align}
    \sum_{a=1}^r d_n(u_{a-1},u_a)
    &\le\sum_{a=1}^r\{s_n(u_a)-s_n(u_{a-1})\}
    =s_n(u_r)-s_n(u_0)
    \le\frac{1+\mathcal L_n}{K},
    \label{eq:finite_approximation_group_length}
\end{align}
where the equality follows by cancellation of the intermediate terms, and the last inequality follows because $s_n(u_0)$ and $s_n(u_r)$ belong to the same interval in \eqref{eq:finite_approximation_partition}.
Thus, the constructed partition satisfies the required condition.

\subparagraph{Step 3-3:}
We now verify the bound for each $\mathcal X_{n,k}$ identified in Step 3-1, using the partition constructed in Step 3-2.

Fix a nonempty $\mathcal X_{n,k}$.
Since $\mathcal X_{n,k}$ is countable, choose increasing finite subsets $\mathcal X_{n,k}^{(N)}$ such that
\[
    z_k^0\in\mathcal X_{n,k}^{(N)}\subseteq\mathcal X_{n,k}^{(N+1)},
    \qquad
    \bigcup_{N\ge1}\mathcal X_{n,k}^{(N)}=\mathcal X_{n,k}.
\]
For each $N$, list the elements of $\mathcal X_{n,k}^{(N)}$ as $u_0<\cdots<u_r$, with $u_b=z_k^0$.
Applying \eqref{eq:finite_approximation_finite_subset} and \eqref{eq:finite_approximation_group_length}, we obtain
\begin{align*}
    &\left\{\mathbb E\left[\max_{x\in\mathcal X_{n,k}^{(N)}}\frac1n\sum_{i=1}^n[\tilde\psi_{i,h_n}(x)-\tilde\psi_{i,h_n}(z_k^0)]^2\right]\right\}^{1/2}\le\sum_{a=1}^r d_n(u_{a-1},u_a)
    \le\frac{1+\mathcal L_n}{K}.
\end{align*}
Since $\mathcal X_{n,k}^{(N)}$ increases to $\mathcal X_{n,k}$, we have
\[
    \max_{x\in\mathcal X_{n,k}^{(N)}}\frac1n\sum_{i=1}^n[\tilde\psi_{i,h_n}(x)-\tilde\psi_{i,h_n}(z_k^0)]^2
    \uparrow
    \sup_{x\in\mathcal X_{n,k}}\frac1n\sum_{i=1}^n[\tilde\psi_{i,h_n}(x)-\tilde\psi_{i,h_n}(z_k^0)]^2.
\]
Thus, the monotone convergence theorem gives
\begin{align}
    &\mathbb E\left[\sup_{x\in\mathcal X_{n,k}}\frac1n\sum_{i=1}^n[\tilde\psi_{i,h_n}(x)-\tilde\psi_{i,h_n}(z_k^0)]^2\right]\nonumber\\
    &=\lim_{N\to\infty}\mathbb E\left[\max_{x\in\mathcal X_{n,k}^{(N)}}\frac1n\sum_{i=1}^n[\tilde\psi_{i,h_n}(x)-\tilde\psi_{i,h_n}(z_k^0)]^2\right]
    \le\frac{(1+\mathcal L_n)^2}{K^2},
    \label{eq:finite_approximation_group_vectors}
\end{align}
where the inequality follows from the finite-subset bound established above.
Finally, substituting \eqref{eq:finite_approximation_group_vectors} into \eqref{eq:finite_approximation_partition_reduction}, we obtain
\begin{align}
    &\mathbb E\left[\sup_{x\in\mathcal X}\min_{z\in\mathcal Z_n}\frac1n\sum_{i=1}^n[\tilde\psi_{i,h_n}(x)-\tilde\psi_{i,h_n}(z)]^2\right]\nonumber\\
    &\le\sum_{\substack{1\le k\le K\\\mathcal X_{n,k}\ne\varnothing}}\mathbb E\left[\sup_{x\in\mathcal X_{n,k}}\frac1n\sum_{i=1}^n[\tilde\psi_{i,h_n}(x)-\tilde\psi_{i,h_n}(z_k^0)]^2\right]
    \le\sum_{\substack{1\le k\le K\\\mathcal X_{n,k}\ne\varnothing}}\frac{(1+\mathcal L_n)^2}{K^2}
    \le\frac{(1+\mathcal L_n)^2}{K},
    \label{eq:finite_approximation_vectors}
\end{align}
where the first inequality follows from \eqref{eq:finite_approximation_partition_reduction}, the second follows from \eqref{eq:finite_approximation_group_vectors}, and the last uses the fact that at most $K$ subsets are nonempty.
This establishes the desired expectation bound.

\paragraph{Step 4:}
We now combine the bounds obtained in Steps 2 and 3 to establish \eqref{eq:upper_finite_approximation}.
Use the deterministic set $\mathcal Z_n$ constructed in Step 3-2 and, for any $L_n\in(0,1]$, choose $\eta$ as in \eqref{eq:finite_approximation_eta}.
In terms of the second probability on the right-hand side of \eqref{eq:finite_approximation_probability_reduction}, we have
\begin{align}
    &\mathbb P\left(\sup_{x\in\mathcal X}\min_{z\in\mathcal Z_n}\left\{\frac1n\sum_{i=1}^n[\tilde\psi_{i,h_n}(x)-\tilde\psi_{i,h_n}(z)]^2\right\}^{1/2}>\eta\right)\nonumber\\
    &=\mathbb P\left(\sup_{x\in\mathcal X}\min_{z\in\mathcal Z_n}\frac1n\sum_{i=1}^n[\tilde\psi_{i,h_n}(x)-\tilde\psi_{i,h_n}(z)]^2>\eta^2\right)\nonumber\\
    &\le\frac1{\eta^2}\mathbb E\left[\sup_{x\in\mathcal X}\min_{z\in\mathcal Z_n}\frac1n\sum_{i=1}^n[\tilde\psi_{i,h_n}(x)-\tilde\psi_{i,h_n}(z)]^2\right] \nonumber\\
    & \le\frac{(1+\mathcal L_n)^2}{K\eta^2}
    =\frac{4(C_0+2C_0^3)^2n(1+\mathcal L_n)^2}{KL_n^2}
    \le\frac{Cn(1+\mathcal L_n)^2}{KL_n^2},
    \label{eq:finite_approximation_vector_probability}
\end{align}
where the first equality follows by squaring both sides of the inequality inside the probability, the first inequality follows from Markov's inequality, the second inequality follows from \eqref{eq:finite_approximation_vectors}, the second equality uses \eqref{eq:finite_approximation_eta}, and the last inequality absorbs the fixed universal constant into $C$.

Finally, substituting \eqref{eq:finite_approximation_finite_event} and \eqref{eq:finite_approximation_vector_probability} into \eqref{eq:finite_approximation_probability_reduction}, we obtain
\begin{align}
    &\mathbb P\left(\sup_{x\in\mathcal X}|\hat T_n(x)|-\max_{z\in\mathcal Z_n}|\hat T_n(z)|>L_n\right)\le C\left\{\frac{B_n^2\log^5(2nm)}n\right\}^{1/4}+\frac{Cn(1+\mathcal L_n)^2}{KL_n^2}.
    \label{eq:finite_approximation_final}
\end{align}
This proves \eqref{eq:upper_finite_approximation}.
\end{proof}

\subsection{Proof of \cref{prop:student_t_discretization_bound}}
\label{app:proof_student_t_discretization_bound}

\begin{proof}[Proof of \cref{prop:student_t_discretization_bound}]
Our goal is to bound
\begin{align*}
    \mathbb P\left(\max_{1\le j\le p}|S_{n,\nu,\mu}(x_j)|\le t
    <\sup_{x\in\mathcal X}|S_{n,\nu,\mu}(x)|\right).
\end{align*}
Fix $n$, $t\ge0$, $\mu>0$, an integer $\nu\ge3$, and a deterministic evaluation grid
$\ul x=x_1<\cdots<x_p=\ol x$, and write $S(x)\coloneqq S_{n,\nu,\mu}(x)$ for $x\in\mathcal X$.
We use the versions of the Gaussian components provided by \cref{lem:student_t_separability} below.
In Step 1, we apply the argument in the proof of \cref{thm:upper_general} to reduce the discretization error over a finite set to probabilities involving two points.
In Step 2, we bound these probabilities.
In Step 3, we combine the bounds and pass to the supremum over $\mathcal X$.

\paragraph{Step 1:}
Let $\mathcal Z\coloneqq\{z_1,\ldots,z_m\}$ be any deterministic finite subset of $\mathcal X$ containing the evaluation grid, with $z_1<\cdots<z_m$.
Since $S$ and $-S$ have the same distribution and \eqref{eq:upper_crossing_count} holds for every finite sequence, the argument in Step 2 of the proof of \cref{thm:upper_general} gives
\begin{align}
    &\mathbb P\left(\max_{1\le j\le p}|S(x_j)|\le t
      <\max_{z\in\mathcal Z}|S(z)|\right)\nonumber\\
    &\quad\le2\sum_{j=1}^{p-1}\Bigg\{
       \sum_{\substack{1\le k<m\\x_j\le z_k<z_{k+1}\le x_{j+1}}}
       \mathbb P\{S(z_k)\le t<S(z_{k+1})\}
       -\mathbb P\{S(x_j)\le t<S(x_{j+1})\}\Bigg\}.
    \label{eq:student_t_finite_crossing_reduction}
\end{align}
To bound the right-hand side of \eqref{eq:student_t_finite_crossing_reduction}, we need an upper bound for $\mathbb P\{S(z_k)\le t<S(z_{k+1})\}$ and a lower bound for $\mathbb P\{S(x_j)\le t<S(x_{j+1})\}$.
We obtain both bounds in Step 2 by evaluating
\begin{align*}
    \mathbb P\{S(x)\le t<S(y)\},\qquad x,y\in\mathcal X.
\end{align*}

\paragraph{Step 2:}
Fix $x,y\in\mathcal X$ and set $\theta\coloneqq d_n(x,y)$.
We show that
\begin{align}
    \frac{(1+\mu t^2/\nu)^{-(\nu-1)/2}}{2\pi}
       \{\theta-H_{t,\nu,\mu}(\theta)\}
    &\le\mathbb P\{S(x)\le t<S(y)\}
    \le\frac{(1+\mu t^2/\nu)^{-(\nu-1)/2}}{2\pi}\theta.
    \label{eq:student_t_pair_bounds}
\end{align}
We first focus on the case $t>0$.
In Step 2-1, we express the crossing probability
$\mathbb P\{S(x)\le t<S(y)\}$ as an integral.
Using this representation, we derive the upper bound in
\eqref{eq:student_t_pair_bounds} in Step 2-2 and the lower bound in Step 2-3.
Finally, we verify \eqref{eq:student_t_pair_bounds} for $t=0$.

\subparagraph{Step 2-1:}

Directly evaluating the crossing probability $\mathbb P\{S(x)\le t<S(y)\}$ is difficult because $S(x)$ and $S(y)$ are ratios of correlated Gaussian random variables.
However, the numerator and denominator of $S$ are constructed from independent copies of $G_n$.
Thus, for each $k=0,\ldots,\nu$, the pair $(G_{n,k}(x),G_{n,k}(y))$ is a centered bivariate Gaussian vector with unit variances and correlation $\cos\theta$.
We use this property to represent $(S(x),S(y))$ as two values of a random function evaluated at angles separated by $\theta$.

First, let $R$ and $\varphi$ be independent random variables, with $R\ge0$, $R^2\sim\chi_2^2$, and $\varphi\sim\operatorname{Uniform}[0,2\pi)$.
The numerator satisfies
\begin{align*}
    (G_{n,0}(x),G_{n,0}(y))
    &\stackrel{d}{=} (Z_1,Z_1\cos\theta+Z_2\sin\theta)\\
    &\stackrel{d}{=} (R\cos\varphi, R\cos\varphi\cos\theta+R\sin\varphi\sin\theta)\\
    &= (R\cos(-\varphi),R\cos(\theta-\varphi)),
\end{align*}
where $Z_1,Z_2$ are independent $N(0,1)$ random variables.
The first distributional equality follows because both sides
are centered bivariate Gaussian vectors with unit variances
and correlation $\cos\theta$.
The second follows from the polar representation
$(Z_1,Z_2)\stackrel{d}{=}
(R\cos\varphi,R\sin\varphi)$,
and the final equality follows from the cosine addition formula.

Next, let $\mathbf V_1,\mathbf V_2$ be independent
$N(0,I_\nu)$ random vectors, independent of $(R,\varphi)$.
The denominator vectors satisfy
\begin{align*}
    \bigl((G_{n,k}(x))_{k=1}^{\nu},(G_{n,k}(y))_{k=1}^{\nu}\bigr)
    &\stackrel{d}{=} \bigl(\mathbf V_1,\mathbf V_1\cos\theta+\mathbf V_2\sin\theta\bigr)\\
    &\stackrel{d}{=}
    \bigl(\mathbf V_1\cos(-\varphi)+\mathbf V_2\sin(-\varphi),\mathbf V_1\cos(\theta-\varphi)+\mathbf V_2\sin(\theta-\varphi)\bigr),
\end{align*}
The first equality follows because the coordinate pairs are independent across $k$, and each pair is centered bivariate Gaussian with unit variances and correlation $\cos\theta$.
The second equality follows from the rotational invariance of the standard Gaussian distribution, conditional on $\varphi$.

Consequently, by the independence of the numerator and denominator Gaussian processes,
\begin{align}
    &(S(x),S(y))\nonumber\\
    &\stackrel{d}{=}
    \Bigg(\frac{R\cos(-\varphi)}{(\mu/\nu)^{1/2}\|\mathbf V_1\cos(-\varphi)+\mathbf V_2\sin(-\varphi)\|},~\frac{R\cos(\theta-\varphi)}{(\mu/\nu)^{1/2}\|\mathbf V_1\cos(\theta-\varphi)+\mathbf V_2\sin(\theta-\varphi)\|}\Bigg).
    \label{eq:student_t_angular_representation}
\end{align}
Thus, $(S(x),S(y))$ can be represented as two values of
\begin{align*}
    \frac{R\cos u}
    {(\mu/\nu)^{1/2}
    \|\mathbf V_1\cos u+\mathbf V_2\sin u\|}
\end{align*}
evaluated at $u=-\varphi$ and $u=\theta-\varphi$.

We use this representation to evaluate $\mathbb P\{S(x)\le t<S(y)\}$.
Conditional on $(R,\mathbf V_1,\mathbf V_2)$, the ratio is a deterministic function of $u$.
We first determine the set of angles $u$ for which the ratio exceeds $t$ and compute its length.
We then use the fact that $\varphi$ is uniformly distributed to express the crossing probability in terms of this length.

Since $R\ge0$ and $t>0$, the ratio can exceed $t$ only if $\cos u>0$.
On $\{\cos u>0\}$, we have
\begin{align}
    \frac{R\cos u}{(\mu/\nu)^{1/2}\|\mathbf V_1\cos u+\mathbf V_2\sin u\|}>t
    &\Longleftrightarrow
    \frac{R}{(\mu/\nu)^{1/2}\|\mathbf V_1+\mathbf V_2\tan u\|}>t \nonumber\\
    &\Longleftrightarrow
    \|\mathbf V_1+\mathbf V_2\tan u\|^2<\frac{\nu R^2}{\mu t^2}.
    \label{eq:student_t_threshold_exceedance}
\end{align}
where the first equivalence follows by dividing the numerator and denominator by $\cos u>0$.
To solve the last inequality for $u$, we expand
the squared norm and complete the square.
Define $\beta\coloneqq \mathbf V_1^\top\mathbf V_2\|\mathbf V_2\|^{-1}$.
Then
\begin{align}
    \|\mathbf V_1+\mathbf V_2\tan u\|^2
    &=
    \|\mathbf V_1\|^2
    +2\mathbf V_1^\top\mathbf V_2\tan u
    +\|\mathbf V_2\|^2\tan^2u \nonumber\\
    &=
    \|\mathbf V_1\|^2
    +2\beta\|\mathbf V_2\|\tan u
    +\|\mathbf V_2\|^2\tan^2u \nonumber\\
    &=
    \|\mathbf V_1\|^2-\beta^2
    +(\beta+\|\mathbf V_2\|\tan u)^2.
    \label{eq:student_t_completing_square}
\end{align}
Let $D\coloneqq \nu R^2/(\mu t^2)-(\|\mathbf V_1\|^2-\beta^2)$.
Substituting \eqref{eq:student_t_completing_square} into the
\eqref{eq:student_t_threshold_exceedance} gives
\begin{align*}
    (\beta+\|\mathbf V_2\|\tan u)^2<D.
\end{align*}
Since the left-hand side is nonnegative, there is no solution if $D\le0$.
If $D>0$, the inequality is equivalent to
\begin{align*}
    \frac{-\beta-\sqrt D}{\|\mathbf V_2\|} &<\tan u <\frac{-\beta+\sqrt D}{\|\mathbf V_2\|}.
\end{align*}
We now solve the preceding inequality for $u$ to obtain the corresponding condition on the angles.
Recall that $\cos u>0$.
Since the ratio is $2\pi$-periodic in $u$, we may restrict $u$ to $(-\pi/2,\pi/2)$.
On this interval, $\tan u$ is strictly increasing with inverse function $\arctan$.
Therefore, applying $\arctan$ to the inequality gives
\begin{align*}
    \arctan\left(\frac{-\beta-\sqrt D}{\|\mathbf V_2\|}\right)
    <u<
    \arctan\left(\frac{-\beta+\sqrt D}{\|\mathbf V_2\|}\right).
\end{align*}
Let $L$ denote the length of the set of angles $u$ for which the ratio exceeds $t$.
As shown above, this set is empty if $D\le0$ and is a single interval if $D>0$.
Therefore,
\begin{align}
    L&\coloneqq
    \begin{cases}
       \displaystyle
       \arctan\!\left(
       \frac{-\beta+\sqrt D}{\|\mathbf V_2\|}
       \right)
       -\arctan\!\left(
       \frac{-\beta-\sqrt D}{\|\mathbf V_2\|}
       \right),
       &D>0,\\[6pt]
       0,&D\le0.
    \end{cases}
    \label{eq:student_t_excursion_length}
\end{align}
Also, since both arctangent values lie in $(-\pi/2,\pi/2)$, it holds that $0\le L<\pi$.

We now use the length $L$ to evaluate the crossing probability.
Conditional on $(R,\mathbf V_1,\mathbf V_2)$, the interval on which the ratio exceeds $t$ is fixed.
By \eqref{eq:student_t_angular_representation}, the event $S(x)\le t<S(y)$ corresponds to the event that $-\varphi$ lies outside this interval while $\theta-\varphi$ lies inside it.
Since $\varphi$ is uniformly distributed on $[0,2\pi)$ independently of $(R,\mathbf V_1,\mathbf V_2)$, we can rotate the angles so that the interval on which the ratio exceeds $t$ becomes $(0,L)$. Let $a$ denote the angle obtained by applying this rotation to $-\varphi$. 
Then, by \eqref{eq:student_t_angular_representation}, the crossing event is equivalent to
\begin{align*}
    a\notin(0,L),\qquad a+\theta\in(0,L),
\end{align*}
where angles are understood modulo $2\pi$.
Conditional on $(R,\mathbf V_1,\mathbf V_2)$, $a$ remains uniformly distributed on $[0,2\pi)$.
Since $0\le\theta\le\pi$ and $0\le L<\pi$, these conditions are equivalent, up to endpoints, to
\begin{align*}
    -\theta<a<\min\{0,L-\theta\}.
\end{align*}
This interval has length $\theta+\min\{0,L-\theta\}=\min\{\theta,L\}$.
Since $a$ is uniformly distributed on $[0,2\pi)$ conditional on $(R,\mathbf V_1,\mathbf V_2)$, the conditional crossing probability equals $\min\{\theta,L\}/(2\pi)$. Taking the expectation over $(R,\mathbf V_1,\mathbf V_2)$ and applying Tonelli's theorem to $\min\{\theta,L\}=\int_0^\theta\mathbf 1\{L>s\}\,ds$, we obtain
\begin{align}
    \mathbb P\{S(x)\le t<S(y)\}
    &=\frac{1}{2\pi}\mathbb E[\min\{\theta,L\}]=\frac{1}{2\pi}\int_0^\theta\mathbb P(L>s)\,ds.
    \label{eq:student_t_phase_average}
\end{align}

\subparagraph{Step 2-2:}
In this step, we bound $\mathbb P(L>s)$ from above.
By \eqref{eq:student_t_excursion_length}, $\{L>0\}=\{D>0\}$, so $\mathbb P(L>s)\le\mathbb P(D>0)$ for every $s\ge0$. We therefore evaluate $\mathbb P(D>v)$ for $v\ge0$ to evaluate \eqref{eq:student_t_phase_average}.
By the definition of $D$, for every $v\ge0$, we have
\begin{align*}
    \mathbb P(D>v)
    &=\mathbb P\left(
    R^2>\frac{\mu t^2}{\nu}
    (\|\mathbf V_1\|^2-\beta^2+v)
    \right).
\end{align*}
Since $R^2\sim\chi_2^2$ is independent of $(\mathbf V_1,\mathbf V_2)$, the survival function of the $\chi_2^2$ distribution gives
\begin{align*}
    &\mathbb P(D>v\mid\mathbf V_1,\mathbf V_2)=\exp\left\{-\frac{\mu t^2}{2\nu}
    (\|\mathbf V_1\|^2-\beta^2+v)\right\}.
\end{align*}
Taking the expectation over $(\mathbf V_1,\mathbf V_2)$, we obtain
\begin{align*}
    \mathbb P(D>v)
    &=\mathbb E\left[
    \exp\left\{-\frac{\mu t^2}{2\nu}
    (\|\mathbf V_1\|^2-\beta^2+v)\right\}
    \right].
\end{align*}
Finally, using the identity $\mathbb E[e^{-\lambda W}]=(1+2\lambda)^{-k/2}$ for $W\sim\chi_k^2$ and $\lambda\ge0$, together with $\|\mathbf V_1\|^2-\beta^2\sim\chi_{\nu-1}^2$, we obtain
\begin{align}
    \mathbb P(D>v)
    &=(1+\mu t^2/\nu)^{-(\nu-1)/2}
    \exp\left(-\frac{\mu t^2}{2\nu}v\right).
    \label{eq:student_t_excursion_tail}
\end{align}
Taking $v=0$ and substituting $\mathbb P(L>s)\le\mathbb P(D>0)$ into \eqref{eq:student_t_phase_average} proves the upper bound in \eqref{eq:student_t_pair_bounds}.

\subparagraph{Step 2-3:}
In this step, we derive a lower bound for $\mathbb P(L>s)$. Fix $s\in(0,\pi)$. To use the tail probability of $D$ obtained in \eqref{eq:student_t_excursion_tail}, we first establish a sufficient condition for $L>s$ in terms of $D$. Since $L,s\in(0,\pi)$ on $\{D>0\}$ and the cotangent function is strictly decreasing on $(0,\pi)$, we have
\begin{align*}
    L>s\quad\Longleftrightarrow\quad\cot L<\cot s.
\end{align*}
By \eqref{eq:student_t_excursion_length} and the identity $\cot(A-B)=(1+\tan A\tan B)/(\tan A-\tan B)$, with $A=\arctan\{(-\beta+\sqrt D)/\|\mathbf V_2\|\}$ and $B=\arctan\{(-\beta-\sqrt D)/\|\mathbf V_2\|\}$, we obtain, on $\{D>0\}$,
\begin{align*}
    \cot L
    &=\frac{\|\mathbf V_2\|^2+\beta^2-D}
    {2\|\mathbf V_2\|\sqrt D}=\frac{\|\mathbf V_2\|^2+\beta^2}
    {2\|\mathbf V_2\|\sqrt D}
    -\frac{\sqrt D}{2\|\mathbf V_2\|}.
\end{align*}
To establish a sufficient condition for $L>s$, we use the identity
\begin{align*}
    \cot s=\frac12\{\cot(s/2)-\tan(s/2)\}.
\end{align*}
Comparing this identity with the expression for $\cot L$, it suffices to show that
\begin{align*}
    \frac{\|\mathbf V_2\|^2+\beta^2}
       {2\|\mathbf V_2\|\sqrt D}
    &<\frac12\cot(s/2), \quad 
    \frac{\sqrt D}{2\|\mathbf V_2\|}
    >\frac12\tan(s/2).
\end{align*}
The first inequality is equivalent to
\begin{align*}
    \sqrt D>
    \frac{\|\mathbf V_2\|^2+\beta^2}{\|\mathbf V_2\|}
    \tan(s/2).
\end{align*}
Since $\|\mathbf V_2\|^2+\beta^2\ge\|\mathbf V_2\|^2$, this condition also implies the second inequality. Thus, whenever this condition holds, it holds that $\cot L<\cot s$.
Since the cotangent is strictly decreasing on $(0,\pi)$, this implies $L>s$. Therefore,
\begin{align*}
    \left\{
    D>
    \frac{(\|\mathbf V_2\|^2+\beta^2)^2}
         {\|\mathbf V_2\|^2}\tan^2(s/2)
    \right\}
    \subseteq\{L>s\}.
\end{align*}

We next evaluate the probability of the event on the left-hand side of the  inclusion. Using the independence of $D$ and $(\beta,\mathbf V_2)$ and \eqref{eq:student_t_excursion_tail}, we have
\begin{align*}
    &\mathbb P\left(
    \left.
    D>\frac{(\|\mathbf V_2\|^2+\beta^2)^2}
    {\|\mathbf V_2\|^2}\tan^2(s/2)
    \,\right|\,\beta,\mathbf V_2
    \right)\\
    &\quad=\mathbb P(D>0)\exp\left\{
    -\frac{\mu t^2}{2\nu}
    \frac{(\|\mathbf V_2\|^2+\beta^2)^2}
    {\|\mathbf V_2\|^2}\tan^2(s/2)
    \right\}.
\end{align*}
Taking the expectation of this conditional probability over $(\beta,\mathbf V_2)$ and using the  inclusion, we obtain
\begin{align}
    \mathbb P(L>s)
    &\ge\mathbb P(D>0)\mathbb E\left[
    \exp\left\{
    -\frac{\mu t^2}{2\nu}
    \frac{(\|\mathbf V_2\|^2+\beta^2)^2}
    {\|\mathbf V_2\|^2}\tan^2(s/2)
    \right\}
    \right].
    \label{eq:student_t_lower_probability}
\end{align}
To bound the expectation in \eqref{eq:student_t_lower_probability}, observe that $\|\mathbf V_2\|^2$ and $\beta^2$ are independent with distributions $\chi_\nu^2$ and $\chi_1^2$, respectively. Hence, $\|\mathbf V_2\|^2+\beta^2\sim\chi_{\nu+1}^2$, and
\begin{align*}
    \frac{\|\mathbf V_2\|^2}{\|\mathbf V_2\|^2+\beta^2}
    \sim\operatorname{Beta}(\nu/2,1/2),
\end{align*}
independently of $\|\mathbf V_2\|^2+\beta^2$.
Using the identity $\mathbb E[W^{-1}]=(a+b-1)/(a-1)$ for $W\sim\operatorname{Beta}(a,b)$ with $a>1$, we obtain
\begin{align*}
    \mathbb E\left[
    \left.
    \frac{\|\mathbf V_2\|^2+\beta^2}{\|\mathbf V_2\|^2}
    \,\right|\,
    \|\mathbf V_2\|^2+\beta^2
    \right]
    &=\frac{\nu-1}{\nu-2}.
\end{align*}
Since the exponential function is convex, Jensen's inequality conditional on $\|\mathbf V_2\|^2+\beta^2$ gives
\begin{align*}
    &\mathbb E\left[\left.\exp\left\{-\frac{\mu t^2}{2\nu}\frac{(\|\mathbf V_2\|^2+\beta^2)^2}{\|\mathbf V_2\|^2}\tan^2(s/2)\right\}\,\right|\,\|\mathbf V_2\|^2+\beta^2\right]\\
    &\ge\exp\left\{-\frac{\mu t^2}{2\nu}\mathbb E\left[\left.\frac{(\|\mathbf V_2\|^2+\beta^2)^2}{\|\mathbf V_2\|^2}\,\right|\,\|\mathbf V_2\|^2+\beta^2\right]\tan^2(s/2)\right\}\\
    &=\exp\left\{-\frac{\mu t^2(\nu-1)}{2\nu(\nu-2)}(\|\mathbf V_2\|^2+\beta^2)\tan^2(s/2)\right\}.
\end{align*}
Taking the expectation of both sides and using the identity $\mathbb E[e^{-\lambda W}]=(1+2\lambda)^{-k/2}$ for $W\sim\chi_k^2$ and $\lambda\ge0$, with $k=\nu+1$ and $\lambda=\mu t^2(\nu-1)\tan^2(s/2)/\{2\nu(\nu-2)\}$, we obtain
\begin{align*}
    &\mathbb E\left[\exp\left\{-\frac{\mu t^2}{2\nu}\frac{(\|\mathbf V_2\|^2+\beta^2)^2}{\|\mathbf V_2\|^2}\tan^2(s/2)\right\}\right]\\
    &\ge\mathbb E\left[\exp\left\{-\frac{\mu t^2(\nu-1)}{2\nu(\nu-2)}(\|\mathbf V_2\|^2+\beta^2)\tan^2(s/2)\right\}\right]=\left\{1+\frac{\mu t^2(\nu-1)}{\nu(\nu-2)}\tan^2(s/2)\right\}^{-(\nu+1)/2},
\end{align*}
where the last equality uses $\|\mathbf V_2\|^2+\beta^2\sim\chi_{\nu+1}^2$.
Substituting this inequality into \eqref{eq:student_t_lower_probability} and using \eqref{eq:student_t_excursion_tail} with $v=0$, we obtain
\begin{align}
    \mathbb P(L>s)
    &\ge
    (1+\mu t^2/\nu)^{-(\nu-1)/2}
    \left\{
    1+\frac{\mu t^2(\nu-1)}{\nu(\nu-2)}
    \tan^2(s/2)
    \right\}^{-(\nu+1)/2}.
    \label{eq:student_t_excursion_lower_tail}
\end{align}
Finally, substituting \eqref{eq:student_t_excursion_lower_tail} into \eqref{eq:student_t_phase_average}, we obtain
\begin{align*}
    \mathbb P\{S(x)\le t<S(y)\}
    &\ge
    \frac{(1+\mu t^2/\nu)^{-(\nu-1)/2}}{2\pi}
    \int_0^\theta
    \left\{
    1+\frac{\mu t^2(\nu-1)}{\nu(\nu-2)}
    \tan^2(s/2)
    \right\}^{-(\nu+1)/2}ds\\
    &=
    \frac{(1+\mu t^2/\nu)^{-(\nu-1)/2}}{2\pi}
    \{\theta-H_{t,\nu,\mu}(\theta)\}.
\end{align*}
This proves the lower bound in \eqref{eq:student_t_pair_bounds} for $t>0$.

Finally, when $t=0$, the crossing event depends only on the signs of the Gaussian numerator values and has probability $\theta/(2\pi)$ by \cref{lem:gaussian_two_point_crossing}. Both bounds in \eqref{eq:student_t_pair_bounds} then hold with equality.

\paragraph{Step 3:}
In Step 3, we combine the results of Steps 1 and 2 to derive a bound on the discretization error over an auxiliary finite grid $\mathcal Z$ containing the evaluation grid. We then extend this bound to the supremum over $\mathcal X$. In Step 3-1, we derive the bound over $\mathcal Z$. In Step 3-2, we extend this bound to the supremum over $\mathcal X$. Before proceeding to these steps, we state an auxiliary result used in Step 3-2.

\begin{lemma}[Finite-set approximation for the Student-$t$ process]
\label{lem:student_t_separability}
Suppose that $\mathcal L_n<\infty$, $\mu>0$, and $\nu\ge3$ is an integer. Then $S_{n,\nu,\mu}$ admits a real-valued separable modification $\widetilde S_{n,\nu,\mu}$. For every deterministic evaluation grid $\ul x=x_1<\cdots<x_p=\ol x$, there exist increasing deterministic finite sets $\mathcal Z_m\subseteq\mathcal X$ containing the evaluation grid such that
\begin{align*}
    &\max_{z\in\mathcal Z_m}|\widetilde S_{n,\nu,\mu}(z)|\uparrow\sup_{x\in\mathcal X}|\widetilde S_{n,\nu,\mu}(x)|\qquad\text{almost surely}. \qedhere
\end{align*}
\end{lemma}

\subparagraph{Step 3-1:}
Recall the auxiliary finite set $\mathcal Z$ defined in Step 1. In \eqref{eq:student_t_finite_crossing_reduction}, apply the upper bound in \eqref{eq:student_t_pair_bounds} to each probability in the inner sum and the lower bound to the subtracted probability. For each $j$, the sum of $d_n$ over consecutive points of $\mathcal Z\cap[x_j,x_{j+1}]$ is at most $\mathcal L_{n,j}$ by the definition of total variation. Therefore,
\begin{align}
    &\mathbb P\left(\max_{1\le j\le p}|S(x_j)|\le t<\max_{z\in\mathcal Z}|S(z)|\right)\nonumber\\
    &\quad\le\frac{(1+\mu t^2/\nu)^{-(\nu-1)/2}}{\pi}\sum_{j=1}^{p-1}\left\{\mathcal L_{n,j}-d_n(x_j,x_{j+1})+H_{t,\nu,\mu}\bigl(d_n(x_j,x_{j+1})\bigr)\right\}.
    \label{eq:student_t_finite_set_bound}
\end{align}
The right-hand side does not depend on $\mathcal Z$. 

\subparagraph{Step 3-2:}
We now use the separable modification $\widetilde S_{n,\nu,\mu}$ provided by \cref{lem:student_t_separability}, and write $S(x)=\widetilde S_{n,\nu,\mu}(x)$. Let $\mathcal Z_m$ be the finite sets provided by the lemma. This modification has the same finite-dimensional distributions as the original Student-$t$ process, so \eqref{eq:student_t_finite_set_bound} applies to each $\mathcal Z_m$.

The almost sure convergence in the lemma implies that the events
\begin{align*}
    \left\{\max_{1\le j\le p}|S(x_j)|\le t<\max_{z\in\mathcal Z_m}|S(z)|\right\}
\end{align*}
increase to $\{\max_{1\le j\le p}|S(x_j)|\le t<\sup_{x\in\mathcal X}|S(x)|\}$ almost surely. By continuity of probability from below and \eqref{eq:student_t_finite_set_bound},
\begin{align*}
    &\mathbb P\left(\max_{1\le j\le p}|S(x_j)|\le t<\sup_{x\in\mathcal X}|S(x)|\right)\\
    &\quad=\lim_{m\to\infty}\mathbb P\left(\max_{1\le j\le p}|S(x_j)|\le t<\max_{z\in\mathcal Z_m}|S(z)|\right)\\
    &\quad\le\frac{(1+\mu t^2/\nu)^{-(\nu-1)/2}}{\pi}\sum_{j=1}^{p-1}\left\{\mathcal L_{n,j}-d_n(x_j,x_{j+1})+H_{t,\nu,\mu}\bigl(d_n(x_j,x_{j+1})\bigr)\right\}.
\end{align*}
This proves the proposition.
\end{proof}

\begin{proof}[Proof of \cref{lem:student_t_separability}]
In Step 1, we choose a separable modification and construct the required finite sets. The general modification theorem allows infinite values, so in Step 2 we use the finite-set bound established above to show that the modification has bounded sample paths almost surely.

\subparagraph{Step 1: }
Since $\mathcal X\subseteq\mathbb R$, Theorem~2.4 in \cite{doob1953stochastic} gives a separable modification of $S_{n,\nu,\mu}$, denoted by $\widetilde S(x)=\widetilde S_{n,\nu,\mu}(x)$, initially taking values in $[-\infty,\infty]$. In particular, there is a deterministic countable set $\{z_m^0:m\ge1\}\subseteq\mathcal X$ such that
\begin{align*}
    \sup_{x\in\mathcal X}|\widetilde S(x)|&=\sup_{m\ge1}|\widetilde S(z_m^0)|\qquad\text{almost surely}.
\end{align*}
For a deterministic evaluation grid, set
\begin{align*}
    \mathcal Z_m&\coloneqq\{x_1,\ldots,x_p,z_1^0,\ldots,z_m^0\}.
\end{align*}
These sets are finite, contain the evaluation grid, and satisfy $\mathcal Z_m\subseteq\mathcal Z_{m+1}$. The displayed equality of suprema therefore gives
\begin{align*}
    \max_{z\in\mathcal Z_m}|\widetilde S(z)|&\uparrow\sup_{x\in\mathcal X}|\widetilde S(x)|\qquad\text{almost surely}.
\end{align*}

\subparagraph{Step 2:}
We show that the modification $\widetilde S$ constructed in Step 1 has bounded sample paths almost surely.
Since $\widetilde S$ has the same finite-dimensional distributions as $S_{n,\nu,\mu}$, the finite-set bound in \eqref{eq:student_t_finite_set_bound} applies to $\widetilde S$ on each $\mathcal Z_m$.
Using $0\le H_{u,\nu,\mu}(\theta)\le\theta$ and $\sum_{j=1}^{p-1}\mathcal L_{n,j}\le\mathcal L_n$, we obtain, for every $u>0$,
\begin{align*}
    \mathbb P\left(\max_{z\in\mathcal Z_m}|\widetilde S(z)|>u\right)&=\mathbb P\left(\max_{1\le j\le p}|\widetilde S(x_j)|>u\right)\\
    &\quad+\mathbb P\left(\max_{1\le j\le p}|\widetilde S(x_j)|\le u<\max_{z\in\mathcal Z_m}|\widetilde S(z)|\right)\\
    &\le\mathbb P\left(\max_{1\le j\le p}|\widetilde S(x_j)|>u\right)+\frac{\mathcal L_n}{\pi}(1+\mu u^2/\nu)^{-(\nu-1)/2}.
\end{align*}
By the almost sure convergence established in Step 1 and continuity of probability from below,
\begin{align*}
    \mathbb P\left(\sup_{x\in\mathcal X}|\widetilde S(x)|>u\right)&=\lim_{m\to\infty}\mathbb P\left(\max_{z\in\mathcal Z_m}|\widetilde S(z)|>u\right)\\
    &\le\mathbb P\left(\max_{1\le j\le p}|\widetilde S(x_j)|>u\right)+\frac{\mathcal L_n}{\pi}(1+\mu u^2/\nu)^{-(\nu-1)/2}.
\end{align*}
We now let $u\to\infty$ with $n$, $\nu$, and $\mu$ fixed.
Since each $\widetilde S(x_j)$ is finite almost surely,
\begin{align*}
    \lim_{u\to\infty}\mathbb P\left(\max_{1\le j\le p}|\widetilde S(x_j)|>u\right)&=0.
\end{align*}
Moreover, $\mathcal L_n<\infty$, $\mu>0$, and $\nu\ge3$ imply
\begin{align*}
    \lim_{u\to\infty}\frac{\mathcal L_n}{\pi}(1+\mu u^2/\nu)^{-(\nu-1)/2}&=0.
\end{align*}
Consequently,
\begin{align*}
    \mathbb P\left(\sup_{x\in\mathcal X}|\widetilde S(x)|=\infty\right)&=\lim_{u\to\infty}\mathbb P\left(\sup_{x\in\mathcal X}|\widetilde S(x)|>u\right)=0.
\end{align*}
Thus, $\widetilde S$ has bounded sample paths almost surely.
This proves the lemma.
\end{proof}

\section{Proofs for the results on local empirical process}

\subsection{Proof of \cref{lem:local_1}}
\begin{proof}[Proof of \cref{lem:local_1}(i)]
Since
\begin{align*}
    \mathbb E[K_{i,h_n}(x)^2]
    &=
    \frac1{h_n}\int_{\mathbb R}
    K(u)^2f(x+uh_n)
    \mathbb E[g(Y_i)^2\mid X_i=x+uh_n]du,\\
    \mathbb E[K_{i,h_n}(x)]
    &=
    \int_{\mathbb R}
    K(u)f(x+uh_n)
    \mathbb E[g(Y_i)\mid X_i=x+uh_n]du,
\end{align*}
under \cref{as:DGP_local}, it holds that
\begin{align*}
    &\sup_{x\in\mathcal X}
    \left|
        \Var[K_{i,h_n}(x)]
        -
        \frac{
            f(x)\mathbb E[g(Y_i)^2\mid X_i=x]
        }{h_n}
        \int_{\mathbb R}K(u)^2du
    \right|\\
    &=
    \sup_{x\in\mathcal X}
    \left|
        \mathbb E[K_{i,h_n}(x)^2]
        -
        \mathbb E[K_{i,h_n}(x)]^2
        -
        \frac1{h_n}\int_{\mathbb R}
        K(u)^2f(x)\mathbb E[g(Y_i)^2\mid X_i=x]du
    \right|\\
    &\le
    \sup_{x\in\mathcal X}
    \left|
        \frac1{h_n}\int_{\mathbb R}K(u)^2
        \left\{
            f(x+uh_n)
            \mathbb E[g(Y_i)^2\mid X_i=x+uh_n]
            -
            f(x)\mathbb E[g(Y_i)^2\mid X_i=x]
        \right\}du
    \right|\\
    &\quad+
    \sup_{x\in\mathcal X}
    \left|
        \int_{\mathbb R}
        K(u)f(x+uh_n)
        \mathbb E[g(Y_i)\mid X_i=x+uh_n]du
    \right|^2\\
    &\le
    \frac1{h_n}\sup_{x\in\mathcal X}
    \int_{\mathbb R}K(u)^2
    \left|
        f(x+uh_n)
        \mathbb E[g(Y_i)^2\mid X_i=x+uh_n]
        -
        f(x)\mathbb E[g(Y_i)^2\mid X_i=x]
    \right|du\\
    &\quad+
    \|g\|_\infty^2
    \left(\sup_{x\in\mathbb R}f(x)\right)^2
    \left(\int_{\mathbb R}|K(u)|du\right)^2\\
    &\le
    \sup_{x\in\mathbb R}
    \left|
        \frac{d}{dx}
        \{f(x)\mathbb E[g(Y_i)^2\mid X_i=x]\}
    \right|
    \int_{\mathbb R}|u|K(u)^2du +
    \|g\|_\infty^2
    \left(\sup_{x\in\mathbb R}f(x)\right)^2
    \left(\int_{\mathbb R}|K(u)|du\right)^2\\
    &=O(1)=o(h_n^{-1}).
\end{align*}
The final inequality follows from the mean-value theorem
applied to $x\mapsto f(x)\mathbb E[g(Y_i)^2\mid X_i=x]$.
The last line follows from \cref{as:DGP_local}(i), (iv),
and $h_n\to0$.
\end{proof}

\begin{proof}[Proof of \cref{lem:local_1}(ii)]
By $c_r$-inequality and Jensen's inequality,
\begin{align*}
    & \mathbb E\left[
        |K_{i,h_n}(x)-\mathbb E[K_{i,h_n}(x)]|^4
    \right] \\
    &\le
    8\mathbb E[K_{i,h_n}(x)^4]
    +
    8|\mathbb E[K_{i,h_n}(x)]|^4\\
    &\le
    16\mathbb E[K_{i,h_n}(x)^4]\\
    &=
    \frac{16}{h_n^3}\int_{\mathbb R}
    K(u)^4f(x+uh_n)
    \mathbb E[g(Y_i)^4\mid X_i=x+uh_n]du\le
    \frac{16\|g\|_\infty^4}{h_n^3}
    \sup_{x\in\mathbb R}f(x)
    \int_{\mathbb R}K(u)^4du.
\end{align*}
Hence,
\begin{align*}
    & \max_{1\le j\le p}\mathbb E[\tilde \psi_{i,h_n}(x_j)^4]\\
    & 
    =
    \max_{1\le j\le p}
    \frac{
        \mathbb E\left[
            |K_{i,h_n}(x_j)
             -\mathbb E[K_{i,h_n}(x_j)]|^4
        \right]
    }{
        \Var[K_{i,h_n}(x_j)]^2
    }\le
    \max_{1\le j\le p}
    \frac{
        16\|g\|_\infty^4\sup_{x\in\mathbb R}f(x)
    }{
        h_n^3\Var[K_{i,h_n}(x_j)]^2
    }
    \int_{\mathbb R}K(u)^4du.
\end{align*}

Also, since
$|K_{i,h_n}(x)|
 \le h_n^{-1}\|g\|_\infty\sup_{u\in\mathbb R}|K(u)|$,
we have
\begin{align*}
    |\tilde\psi_{i,h_n}(x)|
    &\le
    \frac{
        |K_{i,h_n}(x)|
        +
        |\mathbb E[K_{i,h_n}(x)]|
    }{
        \sqrt{\Var[K_{i,h_n}(x)]}
    } \le
    \frac{
        2\|g\|_\infty\sup_{u\in\mathbb R}|K(u)|
    }{
        h_n\sqrt{\Var[K_{i,h_n}(x)]}
    }.
\end{align*}
Consequently,
\begin{align*}
    \mathbb E\left[
        \exp\left\{
            \frac{
                h_n\sqrt{\Var[K_{i,h_n}(x)]}\log2
            }{
                2\|g\|_\infty\sup_{u\in\mathbb R}|K(u)|
            }
            |Y_{i,h_n}(x)|
        \right\}
    \right]
    \le2.
\end{align*}
By the definition of the $\psi_1$-norm,
\begin{align*}
    &\max_{1\le j\le p}\|\tilde \psi_{i,h_n}(x_j)\|_{\psi_1}
    \le
    \max_{1\le j\le p}
    \frac1{\log2}
    \frac{
        2\|g\|_\infty\sup_{u\in\mathbb R}|K(u)|
    }{
        h_n\sqrt{\Var[K_{i,h_n}(x_j)]}
    }. \qedhere
\end{align*}
\end{proof}

\subsection{Proof of \cref{lem:evaluation_M}}

Before proof of  \cref{lem:evaluation_M}, we privide an auxiliary result.
The proof of \cref{lem:local_gaussian_indicator} is provided later.
\begin{lemma}
\label{lem:local_gaussian_indicator}
Let $(Z_1,Z_2)$ be a centered Gaussian vector with unit marginal variances and correlation $r\in[-1,1]$.
There exist universal constants $C,c_1>0$ such that the following statements hold for every $t>0$.
\begin{itemize}
    \item[(i)] $0\le\Corr\left(\mathbf1\{|Z_1|>t\},\mathbf1\{|Z_2|>t\}\right)\le r^2$.

    \item[(ii)] $\Corr\left(\mathbf1\{|Z_1|>t\},\mathbf1\{|Z_2|>t\}\right)\le C\exp\{-c_1(1+t)^2(1-|r|)\}$. \qedhere
\end{itemize}
\end{lemma}

\begin{proof}[Proof of \cref{lem:local_lower_M}]
Let $ r_h(x,y)=\Corr(K_{i,h}(x),K_{i,h}(y))$.
We first establish the following two bounds:
\begin{align}
    1-|r_{h_n}(x,y)|
    &\ge c\frac{|x-y|^2}{h_n^2}
    \qquad\text{if } |x-y|\le2h_n,
    \label{eq:local_M_near_correlation}\\
    |r_{h_n}(x,y)|
    &\le Ch_n
    \qquad\qquad~\text{if } |x-y|>2h_n.
    \label{eq:local_M_far_correlation}
\end{align}

\paragraph{Step 1:}
First, we prove \cref{eq:local_M_near_correlation}.
Consider $0<|x-y|\le2h_n$ and set $u=(y-x)/h_n$.
By the definition of $r_{h_n}(x,y)$,
\begin{align*}
    1-|r_{h_n}(x,y)|
    &\ge\frac{1-r_{h_n}(x,y)^2}{2}=\frac{\inf_{a\in\mathbb R}\Var[K_{i,h_n}(x)-aK_{i,h_n}(y)]}{2\Var[K_{i,h_n}(x)]},
\end{align*}
where the equality uses the fact that minimizing the variance of the linear combination over $a\in\mathbb R$ yields
\begin{align*}
    &\inf_{a\in\mathbb R}\Var[K_{i,h_n}(x)-aK_{i,h_n}(y)]\\
    &= \inf_{a\in\mathbb R} \left( \Var[K_{i,h_n}(x)] - 2a\Cov(K_{i,h_n}(x), K_{i,h_n}(y)) + a^2\Var[K_{i,h_n}(y)]\right)\\
    &=\Var[K_{i,h_n}(x)]
    -\frac{\Cov(K_{i,h_n}(x),K_{i,h_n}(y))^2}{\Var[K_{i,h_n}(y)]}=\Var[K_{i,h_n}(x)]\{1-r_{h_n}(x,y)^2\} \\
    \implies & 1-r_{h_n}(x,y)^2 = \frac{\inf_{a\in\mathbb R}\Var[K_{i,h_n}(x)-aK_{i,h_n}(y)]}{\Var[K_{i,h_n}(x)]}.
\end{align*}
Since $\Var[K_{i,h_n}(x)]\le Ch_n^{-1}$ uniformly over $x\in\mathcal X$ by \cref{lem:local_1}(i), it suffices to show that
\begin{align}
    \inf_{a\in\mathbb R}\Var[K_{i,h_n}(x)-aK_{i,h_n}(y)]
    &\ge\frac{c}{h_n}u^2.
    \label{eq:local_M_variance_lower}
\end{align}

Since $K$ is supported on $[-1,1]$ by \cref{as:DGP_local_lower}, the linear combination $K_{i,h_n}(x)-aK_{i,h_n}(y)$ can be nonzero only if $|X_i-x|\le h_n$ or $|X_i-y|\le h_n$.
Also $ P(|X_i-x|\le h_n) = \int_{x-h_n}^{x+h_n} f(z)dz \le 2h_n \sup_{x\in\mathbb R}f(x)$.
Therefore, for every $a\in\mathbb R$,
\begin{align*}
    &\mathbb P(K_{i,h_n}(x)-aK_{i,h_n}(y)\ne0)\le\mathbb P(|X_i-x|\le h_n)+\mathbb P(|X_i-y|\le h_n)
    \le4h_n\sup_{z\in\mathbb R}f(z).
\end{align*}
Applying the Cauchy--Schwarz inequality to $K_{i,h_n}(x)-aK_{i,h_n}(y)$ and its nonzero-event indicator, we obtain
\begin{align*}
    &|\mathbb E[K_{i,h_n}(x)-aK_{i,h_n}(y)]|^2\\
    &=|\mathbb E[(K_{i,h_n}(x)-aK_{i,h_n}(y))\mathbf1\{K_{i,h_n}(x)-aK_{i,h_n}(y)\ne0\}]|^2\\
    &\le\mathbb E[\{K_{i,h_n}(x)-aK_{i,h_n}(y)\}^2]\mathbb P(K_{i,h_n}(x)-aK_{i,h_n}(y)\ne0)\\
    &\le4h_n\sup_{z\in\mathbb R}f(z)\mathbb E[\{K_{i,h_n}(x)-aK_{i,h_n}(y)\}^2].
\end{align*}
Therefore, for all sufficiently large $n$,
\begin{align*}
    \Var[K_{i,h_n}(x)-aK_{i,h_n}(y)]
    &=\mathbb E[\{K_{i,h_n}(x)-aK_{i,h_n}(y)\}^2]-|\mathbb E[K_{i,h_n}(x)-aK_{i,h_n}(y)]|^2\\
    &\ge\left\{1-4h_n\sup_{z\in\mathbb R}f(z)\right\}\mathbb E[\{K_{i,h_n}(x)-aK_{i,h_n}(y)\}^2]\\
    &\ge\frac12\mathbb E[\{K_{i,h_n}(x)-aK_{i,h_n}(y)\}^2].
\end{align*}
Also, by the definition of $K_{i,h_n}$ and the change of variables $v=(X_i-x)/h_n$,
\begin{align*}
    &\mathbb E[\{K_{i,h_n}(x)-aK_{i,h_n}(y)\}^2]\\
    &\qquad=\frac1{h_n}\int_{\mathbb R}\{K(v)-aK(v-u)\}^2 f(x+h_nv)\mathbb E[g(Y_i)^2\mid X_i=x+h_nv]\,dv\\
    &\qquad\ge\frac{c}{h_n}\int_{\mathbb R}\{K(v)-aK(v-u)\}^2\,dv,
\end{align*}
where the last inequality follows from the uniform positive lower bound on $f(z)\mathbb E[g(Y_i)^2\mid X_i=z]$ in a neighborhood of $\mathcal X$ (\cref{as:DGP_local_lower}(ii)).
Combining these inequalities yields
\begin{align*}
    \Var[K_{i,h_n}(x)-aK_{i,h_n}(y)]
    &\ge\frac{c}{h_n}\int_{\mathbb R}\{K(v)-aK(v-u)\}^2\,dv.
\end{align*}
Thus, it remains to show that
\begin{align*}
    \inf_{a\in\mathbb R}\int_{\mathbb R}\{K(v)-aK(v-u)\}^2\,dv
    &\ge cu^2,\qquad 0<|u|\le2.
\end{align*}
Since $\int_{\mathbb R}K(v)\,dv=1$, let $\mu_K=\int_{\mathbb R}vK(v)\,dv$.
For every $a\in\mathbb R$, the definition of $\mu_K$, the normalization $\int_{\mathbb R}K(v)\,dv=1$, and the change of variables $w=v-u$ yield
\begin{align*}
    &\int_{\mathbb R}(v-u-\mu_K)\{K(v)-aK(v-u)\}\,dv\\
    &\qquad=\int_{\mathbb R}(v-u-\mu_K)K(v)\,dv
    -a\int_{\mathbb R}(w-\mu_K)K(w)\,dw=-u.
\end{align*}
Since $K$ is supported on $[-1,1]$ and $|u|\le2$, the function $K(v)-aK(v-u)$ vanishes outside $[-3,3]$.
Thus, the Cauchy--Schwarz inequality yields
\begin{align*}
    u^2
    &=\left(\int_{-3}^3(v-u-\mu_K)\{K(v)-aK(v-u)\}\,dv\right)^2\\
    &\le\left(\int_{-3}^3(v-u-\mu_K)^2\,dv\right)
    \int_{\mathbb R}\{K(v)-aK(v-u)\}^2\,dv\\
    &\le C\int_{\mathbb R}\{K(v)-aK(v-u)\}^2\,dv,
\end{align*}
where $C$ does not depend on $a$ or $u$.
Taking the infimum over $a\in\mathbb R$ establishes the desired inequality.

\paragraph{Step 2:}
Next, if $|x-y|>2h$, the compact support of $K$ gives
\begin{align*}
    \mathbb E[K_{i,h}(x)K_{i,h}(y)]=0, \quad 
    \Cov(K_{i,h}(x),K_{i,h}(y))
    &=-\mathbb E[K_{i,h}(x)]\mathbb E[K_{i,h}(y)].
\end{align*}
Therefore, since $\Var[K_{i,h}(x)]\asymp h^{-1}$ uniformly over $x\in\mathcal X$, it follows that
\begin{align*}
    |r_h(x,y)|
    &=\frac{|\mathbb E[K_{i,h}(x)]\mathbb E[K_{i,h}(y)]|}{\sqrt{\Var[K_{i,h}(x)]\Var[K_{i,h}(y)]}}
    \le Ch,
\end{align*}
which proves \cref{eq:local_M_far_correlation}.

\paragraph{Step 3:}
By the definition of $M_n(t)$,
\begin{align*}
    M_n(t)^{-1}
    &=\frac1{|\mathcal X|^2}\int_{\mathcal X}\int_{\mathcal X}
    \Corr\left(\mathbf1\{|G_n(x)|>t\},\mathbf1\{|G_n(y)|>t\}\right)\,dy\,dx.
\end{align*}
We bound the integrand separately for $|x-y|\le2h_n$ and $|x-y|>2h_n$.

For $|x-y|\le2h_n$, \cref{lem:local_gaussian_indicator}(ii) and \cref{eq:local_M_near_correlation} imply
\begin{align*}
    &\Corr\left(\mathbf1\{|G_n(x)|>t\},\mathbf1\{|G_n(y)|>t\}\right)\\
    &\qquad\le C\exp\{-c(1+t)^2(1-|r_{h_n}(x,y)|)\}\le C\exp\left\{-c(1+t)^2\frac{|x-y|^2}{h_n^2}\right\}.
\end{align*}

For $|x-y|>2h_n$, \cref{eq:local_M_far_correlation} gives $|r_{h_n}(x,y)|\le Ch_n\le1/2$ for all sufficiently large $n$.
Hence, \cref{lem:local_gaussian_indicator}(i), (ii) yield
\begin{align*}
    &\Corr\left(\mathbf1\{|G_n(x)|>t\},\mathbf1\{|G_n(y)|>t\}\right)\\
    &\qquad\le C\min\left\{h_n^2,\exp\{-c(1+t)^2\}\right\}\le Ch_n^{3/2}\exp\{-c(1+t)^2/4\}\le\frac{Ch_n^{3/2}}{1+t},
\end{align*}
where the second inequality follows from $\min\{a,b\}\le a^{3/4}b^{1/4}$ for $a,b\ge0$, and the last inequality uses $\exp\{-c(1+t)^2/4\}\le C/(1+t)$.

Combining these bounds, we obtain
\begin{align*}
    M_n(t)^{-1}
    &\le\frac C{|\mathcal X|^2}\int_{\mathcal X}\int_{\mathcal X}
    \mathbf1\{|x-y|\le2h_n\}
    \exp\left\{-c(1+t)^2\frac{|x-y|^2}{h_n^2}\right\}\,dy\,dx\\
    &\quad+\frac{Ch_n^{3/2}}{|\mathcal X|^2(1+t)}
    \int_{\mathcal X}\int_{\mathcal X}
    \mathbf1\{|x-y|>2h_n\}\,dy\,dx\\
    &\le\frac{2Ch_n}{|\mathcal X|}\int_0^2\exp\{-c(1+t)^2u^2\}\,du
    +\frac{Ch_n^{3/2}}{1+t} \\
    & \le C\left\{\frac{h_n}{|\mathcal X|(1+t)}+\frac{h_n^{3/2}}{1+t}\right\}\le\frac{Ch_n}{|\mathcal X|(1+t)},
\end{align*}
where the third inequality follows from
\begin{align*}
    \int_0^2\exp\{-c(1+t)^2u^2\}\,du
    &\le\int_0^\infty\exp\{-c(1+t)^2u^2\}\,du=\frac{\sqrt{\pi}}{2\sqrt c(1+t)},
\end{align*}
and the last inequality uses $h_n\le1$ and the fact that $|\mathcal X|$ is fixed.

Finally, since $t_{\alpha,p}=\Phi^{-1}(1-\alpha/(2p))$ satisfies $1+t_{\alpha,p}
    \ge c\sqrt{\log(2p)}$
uniformly over $p\ge2$ and $\alpha\in(0,1)$, taking $t=t_{\alpha,p}$ gives
\begin{align*}
    M_n(t_{\alpha,p})^{-1}
    &\le\frac{Ch_n}{|\mathcal X|\sqrt{\log(2p)}}. \qedhere
\end{align*}
\end{proof}

\subsubsection{Proof of \cref{lem:local_gaussian_indicator}}
We first record the following special case of Gebelein's inequality (Theorem 2.3 in \cite{nourdin2019gebel}), which will be used in the proof of \cref{lem:local_gaussian_indicator}.

\begin{lemma}[Gebelein's inequality]
\label{lem:gebelein}
Let $(Z_1,Z_2)$ be a centered Gaussian vector with unit marginal variances and correlation $r\in[-1,1]$.
Let $f:\mathbb R\to\mathbb R$ be a measurable function satisfying $\mathbb E[f(Z_1)]=\mathbb E[Z_1f(Z_1)]=0$ and $\mathbb E[f(Z_1)^2]<\infty$.
Then $|\mathbb E[f(Z_1)f(Z_2)]|\le r^2\mathbb E[f(Z_1)^2]$.
\end{lemma}

\begin{proof}[Proof of \cref{lem:local_gaussian_indicator}(i)]
Since $\mathbb P(|Z_1|>t)=\mathbb P(|Z_2|>t)=2\bar\Phi(t)$, we have
\begin{align*}
    &\Corr\left(\mathbf1\{|Z_1|>t\},\mathbf1\{|Z_2|>t\}\right)=\frac{\Cov(\mathbf1\{|Z_1|>t\},\mathbf1\{|Z_2|>t\})}{2\bar\Phi(t)\{1-2\bar\Phi(t)\}}.
\end{align*}
Thus, it suffices to show that
\begin{align*}
    0
    &\le\Cov(\mathbf1\{|Z_1|>t\},\mathbf1\{|Z_2|>t\})
    \le r^2 2\bar\Phi(t)\{1-2\bar\Phi(t)\}.
\end{align*}

For the lower bound, the Khatri--\v{S}id\'ak inequality (\cref{lem:khatri_sidak}) gives
\begin{align*}
    \mathbb P(|Z_1|\le t,|Z_2|\le t)
    &\ge\mathbb P(|Z_1|\le t)\mathbb P(|Z_2|\le t).
\end{align*}
Consequently,
\begin{align*}
    \Cov(\mathbf1\{|Z_1|>t\},\mathbf1\{|Z_2|>t\})
    &=\mathbb P(|Z_1|\le t,|Z_2|\le t)-\mathbb P(|Z_1|\le t)\mathbb P(|Z_2|\le t)
    \ge0.
\end{align*}

For the upper bound, let $f(z)=\mathbf1\{|z|>t\}-2\bar\Phi(t)$.
Since $f$ is bounded, centered, and even, the symmetry of the standard Gaussian distribution implies
\begin{align*}
    \mathbb E[f(Z_1)]
    &=\mathbb E[Z_1f(Z_1)]=0,
    \qquad
    \mathbb E[f(Z_1)^2]<\infty.
\end{align*}
Thus, \cref{lem:gebelein} gives
\begin{align*}
    &\left|\Cov(\mathbf1\{|Z_1|>t\},\mathbf1\{|Z_2|>t\})\right|\le r^2\Var(\mathbf1\{|Z_1|>t\})=r^2 2\bar\Phi(t)\{1-2\bar\Phi(t)\}.
\end{align*}
Combining the lower and upper bounds proves \cref{lem:local_gaussian_indicator}(i).
\end{proof}

\begin{proof}[Proof of \cref{lem:local_gaussian_indicator}(ii)]
Since $\mathbb P(|Z_1|>t)=\mathbb P(|Z_2|>t)=2\bar\Phi(t)$, we have
\begin{align*}
    &\Corr\left(\mathbf1\{|Z_1|>t\},\mathbf1\{|Z_2|>t\}\right)=\frac{\Cov(\mathbf1\{|Z_1|>t\},\mathbf1\{|Z_2|>t\})}{2\bar\Phi(t)\{1-2\bar\Phi(t)\}}.
\end{align*}
Since
\begin{align*}
    &\Cov\left(\mathbf1\{|Z_1|>t\},\mathbf1\{|Z_2|>t\}\right)\\
    & =\mathbb P(|Z_1|>t,|Z_2|>t)-\{2\bar\Phi(t)\}^2\le\mathbb P(|Z_1|>t,|Z_2|>t),
\end{align*}
we first bound the joint exceedance probability $\mathbb P(|Z_1|>t,|Z_2|>t)$.

Suppose that $t\ge2$ and $|r|<1$.
By symmetry of the centered Gaussian distribution,
\begin{align*}
    \mathbb P(|Z_1|>t,|Z_2|>t)
    &=2\mathbb P(Z_1>t,Z_2>t)+2\mathbb P(Z_1>t,Z_2<-t).
\end{align*}
Since $\{Z_1>t,Z_2>t\}\subseteq\{Z_1+Z_2>2t\}$ and $\{Z_1>t,Z_2<-t\}\subseteq\{Z_1-Z_2>2t\}$, and since $Z_1+Z_2\sim N(0,2(1+r))$ and $Z_1-Z_2\sim N(0,2(1-r))$, we obtain
\begin{align*}
    \mathbb P(|Z_1|>t,|Z_2|>t)
    &\le2\bar\Phi\left(t\sqrt{\frac2{1+r}}\right)
    +2\bar\Phi\left(t\sqrt{\frac2{1-r}}\right)\le4\bar\Phi\left(t\sqrt{\frac2{1+|r|}}\right).
\end{align*}
To bound $\bar\Phi(at)/\bar\Phi(t)$, we use Proposition 2.1.2 in \cite{vershynin2018high}, which gives 
\begin{align*}
    \left(\frac1u-\frac1{u^3}\right)\phi(u)
    &\le\bar\Phi(u)\le\frac{\phi(u)}u,
    \qquad u>0.
\end{align*}
For every $a\ge1$ and $t\ge2$, these inequalities imply
\begin{align*}
    \frac{\bar\Phi(at)}{\bar\Phi(t)}
    &\le\frac{1}{a(1-t^{-2})}\exp\left\{-\frac{(a^2-1)t^2}{2}\right\}\le\frac43\exp\left\{-\frac{(a^2-1)t^2}{2}\right\}.
\end{align*}
Taking $a=\sqrt{2/(1+|r|)}$ and using $1-2\bar\Phi(t)\ge1-2\bar\Phi(2)>0$, we obtain
\begin{align*}
    &\frac{\Cov(\mathbf1\{|Z_1|>t\},\mathbf1\{|Z_2|>t\})}{2\bar\Phi(t)\{1-2\bar\Phi(t)\}}\\
    &\le\frac{\mathbb P(|Z_1|>t,|Z_2|>t)}{2\bar\Phi(t)\{1-2\bar\Phi(t)\}}\\
    &\le\frac{2}{1-2\bar\Phi(t)}\frac{\bar\Phi(at)}{\bar\Phi(t)}\le\frac{8}{3\{1-2\bar\Phi(2)\}}
    \exp\left\{-\frac{t^2(1-|r|)}{2(1+|r|)}\right\}\le C\exp\left\{-\frac{(1+t)^2(1-|r|)}{16}\right\},
\end{align*}
where the last inequality follows from $1+|r|\le2$ and $(1+t)^2\le4t^2$.

For $0<t<2$, the same bound follows by increasing $C$, since a correlation is at most one and $(1+t)^2(1-|r|)\le9$.
Finally, if $|r|=1$, then $|Z_1|=|Z_2|$ almost surely, so the correlation of the two indicators equals one.
Thus, the desired bound holds for every $t>0$ and $r\in[-1,1]$ with $c_1=1/16$ and a universal constant $C>0$.
\end{proof}

\subsection{Proof of the equivalence in \cref{thm:local_coverage_failure}} \label{subsec:rmk_equiv}
\begin{proof}[Proof of the equivalence in \cref{thm:local_coverage_failure}]
We prove both implications.

\paragraph{Step 1:}
First, suppose that $p_nh_n/\sqrt{\log(h_n^{-1})}\to0$.
We distinguish two cases according to the size of $p_n$ relative to $h_n^{-1/2}$.
If $p_n<h_n^{-1/2}$, then $p_n\ge2$ implies 
\begin{align*}
    \frac{p_nh_n}{\sqrt{\log(2p_n)}}\le\frac{\sqrt{h_n}}{\sqrt{\log4}}.
\end{align*}
If $p_n\ge h_n^{-1/2}$, then $\log(2p_n)\ge 2^{-1}\log(h_n^{-1})$
and hence
\begin{align*}
    \frac{p_nh_n}{\sqrt{\log(2p_n)}}
    &\le\sqrt2\frac{p_nh_n}{\sqrt{\log(h_n^{-1})}}.
\end{align*}
Combining these two cases, we obtain
\begin{align*}
    0\le\frac{p_nh_n}{\sqrt{\log(2p_n)}}
    &\le\frac{\sqrt{h_n}}{\sqrt{\log4}}
    +\sqrt2\frac{p_nh_n}{\sqrt{\log(h_n^{-1})}}
    \to0,
\end{align*}
where the convergence follows from $h_n\to0$ and the assumed condition.

\paragraph{Step 2:}
Conversely, suppose that $ p_nh_n/\sqrt{\log(2p_n)}\to0$. 
Then $(p_nh_n)^2=o(\log(2p_n))$.
Since $\log(2z)\le z^2$ for every $z>0$, we obtain
\begin{align*}
    \log(2p_n)
    &=\log(h_n^{-1})+\log(2p_nh_n)\le\log(h_n^{-1})+(p_nh_n)^2=\log(h_n^{-1})+o(\log(2p_n)).
\end{align*}
Consequently,
\begin{align*}
    \frac{\log(2p_n)}{\log(h_n^{-1})}
    &\le1+o(1).
\end{align*}
Therefore,
\begin{align*}
    \frac{p_nh_n}{\sqrt{\log(h_n^{-1})}}
    &=\sqrt{\frac{\log(2p_n)}{\log(h_n^{-1})}}
    \frac{p_nh_n}{\sqrt{\log(2p_n)}}\le(1+o(1))\frac{p_nh_n}{\sqrt{\log(2p_n)}}\to0.
\end{align*}
This completes the proof.
\end{proof}

\subsection{Proof of \cref{lem:local_1_disc}}

\begin{proof}[Proof of \cref{lem:local_1_disc}(i)]
By the definition of $d_n$,
\begin{align*}
    d_n(x,y) &= \arccos\{\mathbb E[G_n(x)G_n(y)]\}\\
    &= \arccos\left[\Corr\{K_{1,h_n}(x),K_{1,h_n}(y)\}\right].
\end{align*}
For $x<y$ in $\mathcal X$, set $u=(y-x)/h_n$.
Then,
\begin{align*}
    \Corr\{K_{1,h_n}(x),K_{1,h_n}(y)\} &= \Corr\left\{g(Y_1)K\left(\frac{X_1-x}{h_n}\right),g(Y_1)K\left(\frac{X_1-x}{h_n}-u\right)\right\}\\
    &= \Corr\{K_{1,h_n}(x),K_{1,h_n}(x+h_nu)\}.
\end{align*}
For adjacent evaluation points, $(p-1)h_n\to\infty$ implies
\begin{align*}
    u &= \frac{x_{j+1}-x_j}{h_n}=\frac{\ol x-\ul x}{h_n(p-1)}\to0.
\end{align*}
Since $\arccos(1-z)=2\arcsin(\sqrt{z/2})=\sqrt{2z}\{1+O(z)\}$ as $z\to 0$, it suffices to show that
\begin{align}
    \Corr\{K_{1,h_n}(x),K_{1,h_n}(x+h_nu)\} &= 1-\frac{u^2}{2}\left\{\frac{\int_{\mathbb R}\{K'(v)\}^2dv}{\int_{\mathbb R}K(v)^2dv}+O(h_n+u)\right\}
    \label{eq:local_correlation_expansion}
\end{align}
uniformly over $x,x+h_nu\in\mathcal X$ as $u \to 0$.

By the definition of correlation,
\begin{align*}
    &\Corr\{K_{1,h_n}(x),K_{1,h_n}(x+h_nu)\}
    =\frac{h_n\Cov\{K_{1,h_n}(x),K_{1,h_n}(x+h_nu)\}}{\{h_n\Var[K_{1,h_n}(x)]h_n\Var[K_{1,h_n}(x+h_nu)]\}^{1/2}}.
\end{align*}
We first expand the covariance in the numerator.
Write $b(x)=f(x)\mathbb E[g(Y_1)^2\mid X_1=x]$.
Conditioning on $X_1$, changing variables, and Taylor expansion with respect to $u$ give
\begin{align}
    &h_n\Cov\{K_{1,h_n}(x),K_{1,h_n}(x+h_nu)\} \nonumber\\
    &=\int_{\mathbb R}b(x+h_nv)K(v)K(v-u)dv-h_n\mathbb E[K_{1,h_n}(x)]\mathbb E[K_{1,h_n}(x+h_nu)] \nonumber\\
    &=h_n\Var[K_{1,h_n}(x)]+uh_n\Cov\{K_{1,h_n}(x),h_n\partial_xK_{1,h_n}(x)\} \nonumber\\
    &\quad+\frac{u^2}{2}h_n\Cov\{K_{1,h_n}(x),h_n^2\partial_x^2K_{1,h_n}(x)\}+O(u^3).
    \label{eq:local_covariance_taylor}
\end{align}
Similarly,
\begin{align}
    &h_n\Var[K_{1,h_n}(x+h_nu)] \nonumber\\
    &=\int_{\mathbb R}b(x+h_nv)K(v-u)^2dv-h_n\{\mathbb E[K_{1,h_n}(x+h_nu)]\}^2 \nonumber\\
    &=h_n\Var[K_{1,h_n}(x)]+2uh_n\Cov\{K_{1,h_n}(x),h_n\partial_xK_{1,h_n}(x)\} \nonumber\\
    &\quad+u^2h_n\left\{\Var[h_n\partial_xK_{1,h_n}(x)]+\Cov\{K_{1,h_n}(x),h_n^2\partial_x^2K_{1,h_n}(x)\}\right\}+O(u^3).
    \label{eq:local_variance_taylor}
\end{align}

From Taylor expansion and the Lipschitz continuity of $K''$, for $0\le u\le1$,
\begin{align*}
    \sup_{v\in\mathbb R}\left|K(v-u)-K(v)+uK'(v)-\frac{u^2}{2}K''(v)\right| &\le Cu^3,\\
    \sup_{v\in\mathbb R}\left|K(v-u)^2-K(v)^2+2uK(v)K'(v)-u^2\left[\{K'(v)\}^2+K(v)K''(v)\right]\right| &\le Cu^3.
\end{align*}
Since $K$, $K'$, and $K''$ are compactly supported, both differences inside the absolute values vanish outside a fixed compact interval.
Hence, boundedness of $b$ gives uniform $O(u^3)$ remainder bounds for the integrals in \eqref{eq:local_covariance_taylor} and \eqref{eq:local_variance_taylor}.
The first inequality also gives
\begin{align*}
    \mathbb E[K_{1,h_n}(x+h_nu)] &= \mathbb E[K_{1,h_n}(x)]+u\mathbb E[h_n\partial_xK_{1,h_n}(x)]+\frac{u^2}{2}\mathbb E[h_n^2\partial_x^2K_{1,h_n}(x)]+O(u^3),
\end{align*}
uniformly over $x,x+h_nu\in\mathcal X$, by boundedness of $g$ and $f$.
Since the expectations in the preceding expansion are uniformly bounded, we have
\begin{align*}
    &h_n\mathbb E[K_{1,h_n}(x)]\mathbb E[K_{1,h_n}(x+h_nu)]\\
    &=h_n\{\mathbb E[K_{1,h_n}(x)]\}^2+uh_n\mathbb E[K_{1,h_n}(x)]\mathbb E[h_n\partial_xK_{1,h_n}(x)]\\
    &\quad+\frac{u^2h_n}{2}\mathbb E[K_{1,h_n}(x)]\mathbb E[h_n^2\partial_x^2K_{1,h_n}(x)]+O(h_nu^3),
\end{align*}
and
\begin{align*}
    &h_n\{\mathbb E[K_{1,h_n}(x+h_nu)]\}^2\\
    &=h_n\{\mathbb E[K_{1,h_n}(x)]\}^2+2uh_n\mathbb E[K_{1,h_n}(x)]\mathbb E[h_n\partial_xK_{1,h_n}(x)]\\
    &\quad+u^2h_n\left\{\{\mathbb E[h_n\partial_xK_{1,h_n}(x)]\}^2+\mathbb E[K_{1,h_n}(x)]\mathbb E[h_n^2\partial_x^2K_{1,h_n}(x)]\right\}+O(h_nu^3).
\end{align*}
Together with the uniform $O(u^3)$ remainder bounds for the scaled second moments, these expansions establish the uniform $O(u^3)$ remainder bounds in \eqref{eq:local_covariance_taylor} and \eqref{eq:local_variance_taylor}.

To substitute these expansions into the correlation formula, we also need a positive lower bound on the scaled variance.
The same change of variables gives
\begin{align}
    h_n\Var[K_{1,h_n}(x)] &= \int_{\mathbb R}b(x+h_nv)K(v)^2dv-h_n\{\mathbb E[K_{1,h_n}(x)]\}^2 = b(x)\int_{\mathbb R}K(v)^2dv+O(h_n).
    \label{eq:local_kernel_variance}
\end{align}
Since $\inf_{x\in\mathcal X}b(x)>0$, the scaled variance is asymptotically uniformly bounded away from zero.

Substituting \eqref{eq:local_covariance_taylor} and \eqref{eq:local_variance_taylor} into the correlation formula and using $(1+z)^{-1/2}=1-z/2+3z^2/8+O(z^3)$, we obtain
\begin{align}
    &\Corr\{K_{1,h_n}(x),K_{1,h_n}(x+h_nu)\} \nonumber\\
    &=\left\{1+u\frac{\Cov\{K_{1,h_n}(x),h_n\partial_xK_{1,h_n}(x)\}}{\Var[K_{1,h_n}(x)]}+\frac{u^2}{2}\frac{\Cov\{K_{1,h_n}(x),h_n^2\partial_x^2K_{1,h_n}(x)\}}{\Var[K_{1,h_n}(x)]}+O(u^3)\right\} \nonumber\\
    &\quad\times\left\{1-u\frac{\Cov\{K_{1,h_n}(x),h_n\partial_xK_{1,h_n}(x)\}}{\Var[K_{1,h_n}(x)]}-\frac{u^2}{2}\frac{\Var[h_n\partial_xK_{1,h_n}(x)]+\Cov\{K_{1,h_n}(x),h_n^2\partial_x^2K_{1,h_n}(x)\}}{\Var[K_{1,h_n}(x)]}\right. \nonumber\\
    &\qquad\left.+\frac{3u^2}{2}\frac{\Cov\{K_{1,h_n}(x),h_n\partial_xK_{1,h_n}(x)\}^2}{\Var[K_{1,h_n}(x)]^2}+O(u^3)\right\} \nonumber\\
    &=1-\frac{u^2}{2}\left\{\frac{\Var[h_n\partial_xK_{1,h_n}(x)]}{\Var[K_{1,h_n}(x)]}-\frac{\Cov\{K_{1,h_n}(x),h_n\partial_xK_{1,h_n}(x)\}^2}{\Var[K_{1,h_n}(x)]^2}\right\}+O(u^3).
    \label{eq:local_correlation_taylor}
\end{align}
It remains to evaluate the two ratios in \eqref{eq:local_correlation_taylor}.
Since
\begin{align*}
    h_n\partial_xK_{1,h_n}(x) &= -\frac{g(Y_1)}{h_n}K'\left(\frac{X_1-x}{h_n}\right),
\end{align*}
conditioning on $X_1$ and using the uniform bounds on the means yield
\begin{align}
    & h_n\Var[h_n\partial_xK_{1,h_n}(x)] = \int_{\mathbb R}b(x+h_nv)\{K'(v)\}^2dv+O(h_n)= b(x)\int_{\mathbb R}\{K'(v)\}^2dv+O(h_n), \nonumber\\
    & h_n\Cov\{K_{1,h_n}(x),h_n\partial_xK_{1,h_n}(x)\} = -\int_{\mathbb R}b(x+h_nv)K(v)K'(v)dv+O(h_n) = O(h_n).
    \label{eq:local_kernel_derivative_moments}
\end{align}
These estimates use the bounded derivative of $b$, compact support of $K$ and $K'$, and $\int_{\mathbb R}K(v)K'(v)dv=0$.
Together with \eqref{eq:local_kernel_variance}, they imply
\begin{align*}
    \frac{\Var[h_n\partial_xK_{1,h_n}(x)]}{\Var[K_{1,h_n}(x)]} = \frac{\int_{\mathbb R}\{K'(v)\}^2dv}{\int_{\mathbb R}K(v)^2dv}+O(h_n), ~~
    \frac{\Cov\{K_{1,h_n}(x),h_n\partial_xK_{1,h_n}(x)\}^2}{\Var[K_{1,h_n}(x)]^2} = O(h_n^2).
\end{align*}
Substitution into \eqref{eq:local_correlation_taylor} proves \eqref{eq:local_correlation_expansion}.
Since $K$ is continuously differentiable, compactly supported, and $\int_{\mathbb R} K(u)du =1$, we can see $\int_{\mathbb R}\{K'(v)\}^2dv>0$.
The expansion of $\arccos(1-z)$ therefore gives
\begin{align*}
    d_n(x,x+h_nu) &= u\left\{\left(\frac{\int_{\mathbb R}\{K'(v)\}^2dv}{\int_{\mathbb R}K(v)^2dv}\right)^{1/2}+O(h_n+u)\right\}.
\end{align*}
Taking $x=x_j$ and $u=(\ol x-\ul x)/\{h_n(p-1)\}$, we conclude that
\begin{align*}
    d_n(x_j,x_{j+1}) &= \frac{\ol x-\ul x}{h_n(p-1)}\left\{\left(\frac{\int_{\mathbb R}\{K'(v)\}^2dv}{\int_{\mathbb R}K(v)^2dv}\right)^{1/2}+o(1)\right\},
\end{align*}
uniformly over $j\in[p-1]$.
\end{proof}

\begin{proof}[Proof of \cref{lem:local_1_disc}(ii)]
To evaluate $\mathcal L_{n,j}-d_n(x_j,x_{j+1})$, for each $j\in[p-1]$, define
\begin{align*}
    c_{n,j} &:= \left\{\mathbb E[\{\tilde\psi_{1,h_n}(x_{j+1})-\tilde\psi_{1,h_n}(x_j)\}^2]\right\}^{1/2}.
\end{align*}
Then, we have
\begin{align}
    \mathcal L_{n,j}-d_n(x_j,x_{j+1}) &= \{\mathcal L_{n,j}-c_{n,j}\}-\{d_n(x_j,x_{j+1})-c_{n,j}\} \nonumber\\
    &= \frac{\mathcal L_{n,j}^2-c_{n,j}^2}{\mathcal L_{n,j}+c_{n,j}}-\{d_n(x_j,x_{j+1})-c_{n,j}\}.
    \label{eq:local_length_excess_decomposition}
\end{align}

We first express $c_{n,j}$ in terms of $d_n(x_j,x_{j+1})$.
From the definition of $d_n$ and the covariance function of $G_n$,
\begin{align*}
    \mathbb E[\tilde\psi_{1,h_n}(x)\tilde\psi_{1,h_n}(y)] &= \mathbb E[G_n(x)G_n(y)] = \cos\{d_n(x,y)\}.
\end{align*}
Since $\mathbb E[\tilde\psi_{1,h_n}(x)^2]=1$ for every $x\in\mathcal X$, we obtain
\begin{align*}
    \mathbb E[\{\tilde\psi_{1,h_n}(y)-\tilde\psi_{1,h_n}(x)\}^2] &= \mathbb E[\tilde\psi_{1,h_n}(y)^2]+\mathbb E[\tilde\psi_{1,h_n}(x)^2]-2\mathbb E[\tilde\psi_{1,h_n}(x)\tilde\psi_{1,h_n}(y)]\\
    &= 2\{1-\cos d_n(x,y)\}.
\end{align*}
Since $1-\cos t=2\sin^2(t/2)$ and $d_n(x,y)\in[0,\pi]$, we have
\begin{align}
    c_{n,j} &= \left[4\sin^2\{d_n(x_j,x_{j+1})/2\}\right]^{1/2} = 2\sin\{d_n(x_j,x_{j+1})/2\}.
    \label{eq:local_spherical_chord}
\end{align}
From \cref{lem:local_1_disc}(i) and $t-2\sin(t/2)=t^3/24+O(t^5)$ as $t\to0$,
\begin{align}
    d_n(x_j,x_{j+1})-c_{n,j} &= \frac{d_n(x_j,x_{j+1})^3}{24}+O\{d_n(x_j,x_{j+1})^5\} \nonumber\\
    &= \left(\frac{x_{j+1}-x_j}{h_n}\right)^3\left\{\frac1{24}\left(\frac{\int_{\mathbb R}\{K'(v)\}^2dv}{\int_{\mathbb R}K(v)^2dv}\right)^{3/2}+o(1)\right\},
    \label{eq:local_angle_chord_difference}
\end{align}
uniformly over $j\in[p-1]$.
Thus, it remains to evaluate $\mathcal L_{n,j}+c_{n,j}$ and $\mathcal L_{n,j}^2-c_{n,j}^2$.
In Steps 1 and 2, we evaluate $\mathcal L_{n,j}+c_{n,j}$ and $\mathcal L_{n,j}^2-c_{n,j}^2$, respectively.
In Step 3, we combine these expansions to obtain the asymptotic expansion of $\mathcal L_{n,j}-d_n(x_j,x_{j+1})$.

\paragraph{Step 1:}
To evaluate $\mathcal L_{n,j}+c_{n,j}$, in view of \cref{lem:local_1_disc}(i) and \eqref{eq:local_spherical_chord}, it remains to obtain an expansion of $\mathcal L_{n,j}$.

\subparagraph{Step 1-1:}
Recall that
\begin{align*}
    \mathcal L_{n,j} &= \sup_{\substack{r\ge1\\x_j=z_0<\cdots<z_r=x_{j+1}}}\sum_{a=1}^r d_n(z_{a-1},z_a).
\end{align*}
To avoid directly evaluating the supremum over partitions, we show that
\begin{align}
    \mathcal L_{n,j} &= \int_{x_j}^{x_{j+1}}v_n(x)dx,\qquad v_n(x):=\left\{\mathbb E[\{\partial_x\tilde\psi_{1,h_n}(x)\}^2]\right\}^{1/2}.
    \label{eq:local_spherical_length_identity}
\end{align}
For this representation, it suffices to establish, for each fixed $n$,
\begin{align}
    \lim_{\delta\downarrow0}\sup_{\substack{x,x+\varepsilon\in\mathcal X\\0<\varepsilon\le\delta}}\left|\frac{d_n(x,x+\varepsilon)}{\varepsilon}-v_n(x)\right| &= 0.
    \label{eq:local_distance_derivative_target}
\end{align}
Indeed, suppose that \eqref{eq:local_distance_derivative_target} holds.
For any partition $x_j=z_0<\cdots<z_r=x_{j+1}$, write $\delta_r=\max_{a\in[r]}(z_a-z_{a-1})$.
Then,
\begin{align*}
    &\left|\sum_{a=1}^r d_n(z_{a-1},z_a)-\sum_{a=1}^r v_n(z_{a-1})(z_a-z_{a-1})\right|\\
    &\le (x_{j+1}-x_j)\sup_{\substack{x,x+\varepsilon\in\mathcal X\\0<\varepsilon\le\delta_r}}\left|\frac{d_n(x,x+\varepsilon)}{\varepsilon}-v_n(x)\right|\to0
\end{align*}
as $\delta_r\downarrow0$, with $n$ fixed.
Since $v_n$ is continuous on $\mathcal X$, the Riemann sums converge, and hence
\begin{align*}
    \sum_{a=1}^r d_n(z_{a-1},z_a) &\to\int_{x_j}^{x_{j+1}}v_n(x)dx.
\end{align*}
Moreover, the triangle inequality for $d_n$ implies that refining a partition cannot decrease its sum of $d_n$.
Thus, for any partition $x_j=z_0<\cdots<z_r=x_{j+1}$ and a sequence of its refinements $x_j=z_0^{(m)}<\cdots<z_{r_m}^{(m)}=x_{j+1}$ with maximum interval length tending to zero,
\begin{align*}
    \sum_{a=1}^r d_n(z_{a-1},z_a) &\le \lim_{m\to\infty}\sum_{a=1}^{r_m}d_n(z_{a-1}^{(m)},z_a^{(m)})=\int_{x_j}^{x_{j+1}}v_n(x)dx
    \le\mathcal L_{n,j},
\end{align*}
where the last inequality follows from the definition of $\mathcal L_{n,j}$.
Taking the supremum over all partitions in the first expression gives
\begin{align*}
    \mathcal L_{n,j} &\le\int_{x_j}^{x_{j+1}}v_n(x)dx\le\mathcal L_{n,j},
\end{align*}
which proves \eqref{eq:local_spherical_length_identity}.

\subparagraph{Step 1-2:}
It remains to establish \eqref{eq:local_distance_derivative_target}.
For $\varepsilon>0$ with $x,x+\varepsilon\in\mathcal X$, the triangle inequality gives
\begin{align*}
    &\left|\frac{d_n(x,x+\varepsilon)}{\varepsilon}-v_n(x)\right|\\
    &\le\frac1{\varepsilon}\left|d_n(x,x+\varepsilon)-\|\tilde\psi_{1,h_n}(x+\varepsilon)-\tilde\psi_{1,h_n}(x)\|_{L^2(P)}\right|\\
    &\quad+\frac1{\varepsilon}\left|\|\tilde\psi_{1,h_n}(x+\varepsilon)-\tilde\psi_{1,h_n}(x)\|_{L^2(P)}-\varepsilon v_n(x)\right|.
\end{align*}
Thus, to establish \eqref{eq:local_distance_derivative_target}, it suffices to show that, for each fixed $n$,
\begin{align}
    &\sup_{x,x+\varepsilon\in\mathcal X}\left|\|\tilde\psi_{1,h_n}(x+\varepsilon)-\tilde\psi_{1,h_n}(x)\|_{L^2(P)}-\varepsilon v_n(x)\right|=o(\varepsilon),
    \label{eq:local_increment_norm_approx}\\
    &\sup_{x,x+\varepsilon\in\mathcal X}\left|d_n(x,x+\varepsilon)-\|\tilde\psi_{1,h_n}(x+\varepsilon)-\tilde\psi_{1,h_n}(x)\|_{L^2(P)}\right|=O(\varepsilon^3),
    \label{eq:local_distance_chord_approx}
\end{align}
as $\varepsilon \to 0$.
We first establish the first estimate and then establish the second estimate.

We first establish \cref{eq:local_increment_norm_approx}.
Fix $n$ sufficiently large.
For $\varepsilon>0$ with $x,x+\varepsilon\in\mathcal X$, the fundamental theorem of calculus gives
\begin{align*}
    &\tilde\psi_{1,h_n}(x+\varepsilon)-\tilde\psi_{1,h_n}(x)-\varepsilon\partial_x\tilde\psi_{1,h_n}(x)\\
    &=\int_x^{x+\varepsilon}\partial_x\tilde\psi_{1,h_n}(s)ds -  \varepsilon\partial_x\tilde\psi_{1,h_n}(x) =\int_x^{x+\varepsilon}\{\partial_x\tilde\psi_{1,h_n}(s)-\partial_x\tilde\psi_{1,h_n}(x)\}ds.
\end{align*}
By the reverse triangle inequality and Minkowski's integral inequality,
\begin{align*}
    &\left|\|\tilde\psi_{1,h_n}(x+\varepsilon)-\tilde\psi_{1,h_n}(x)\|_{L^2(P)}-\varepsilon v_n(x)\right|\\
    &\le\|\tilde\psi_{1,h_n}(x+\varepsilon)-\tilde\psi_{1,h_n}(x)-\varepsilon\partial_x\tilde\psi_{1,h_n}(x)\|_{L^2(P)}\\
    &\le\int_x^{x+\varepsilon}\|\partial_x\tilde\psi_{1,h_n}(s)-\partial_x\tilde\psi_{1,h_n}(x)\|_{L^2(P)}ds\\
    & \le\varepsilon  \sup_{s\in[x,x+\varepsilon]} \|\partial_x\tilde\psi_{1,h_n}(s)-\partial_x\tilde\psi_{1,h_n}(x)\|_{L^2(P)}\\
    &\le\varepsilon\sup_{\substack{s,t\in\mathcal X\\|s-t|\le\varepsilon}}\|\partial_x\tilde\psi_{1,h_n}(s)-\partial_x\tilde\psi_{1,h_n}(t)\|_{L^2(P)}.
\end{align*}
Since $\partial_x\tilde\psi_{1,h_n}$ is continuous on the compact interval $\mathcal X$, it is uniformly continuous.
Thus, taking the supremum over $x,x+\varepsilon\in\mathcal X$ gives
\begin{align*}
    &\sup_{x,x+\varepsilon\in\mathcal X}\left|\|\tilde\psi_{1,h_n}(x+\varepsilon)-\tilde\psi_{1,h_n}(x)\|_{L^2(P)}-\varepsilon v_n(x)\right|\\
    &\le\varepsilon\sup_{\substack{s,t\in\mathcal X\\|s-t|\le\varepsilon}}\|\partial_x\tilde\psi_{1,h_n}(s)-\partial_x\tilde\psi_{1,h_n}(t)\|_{L^2(P)}=o(\varepsilon)
\end{align*}
as $\varepsilon\to 0$, with $n$ fixed.
This establishes \cref{eq:local_increment_norm_approx}.

We next establish \eqref{eq:local_distance_chord_approx}.
By the fundamental theorem of calculus and Minkowski's integral inequality,
\begin{align*}
    \|\tilde\psi_{1,h_n}(x+\varepsilon)-\tilde\psi_{1,h_n}(x)\|_{L^2(P)} &\le\int_x^{x+\varepsilon}v_n(s)ds\le\varepsilon\sup_{s\in\mathcal X}v_n(s)=O(\varepsilon),
\end{align*}
where the last equality follows from the boundedness of $v_n$ on $\mathcal X$ for fixed $n$.

By the definition of $d_n$ and the unit second moments of the standardized summands,
\begin{align*}
    \|\tilde\psi_{1,h_n}(x+\varepsilon)-\tilde\psi_{1,h_n}(x)\|_{L^2(P)}^2 &= 2\{1-\cos d_n(x,x+\varepsilon)\}\\
    &=4\sin^2\{d_n(x,x+\varepsilon)/2\}.
\end{align*}
Since $d_n(x,x+\varepsilon)\in[0,\pi]$, taking square roots and then applying \(\arcsin\) to both sides yields
\begin{align*}
    d_n(x,x+\varepsilon) &= 2\arcsin\left(\frac{\|\tilde\psi_{1,h_n}(x+\varepsilon)-\tilde\psi_{1,h_n}(x)\|_{L^2(P)}}{2}\right).
\end{align*}
Also, $2\arcsin(z/2)=z+O(z^3)$ as $z\to0$ gives
\begin{align*}
    d_n(x,x+\varepsilon) &= 2\arcsin\left(\frac{\|\tilde\psi_{1,h_n}(x+\varepsilon)-\tilde\psi_{1,h_n}(x)\|_{L^2(P)}}{2}\right)\\
    &=\|\tilde\psi_{1,h_n}(x+\varepsilon)-\tilde\psi_{1,h_n}(x)\|_{L^2(P)}+O(\varepsilon^3).
\end{align*}
Consequently, for fixed $n$,
\begin{align*}
    \sup_{x,x+\varepsilon\in\mathcal X}\left|d_n(x,x+\varepsilon)-\|\tilde\psi_{1,h_n}(x+\varepsilon)-\tilde\psi_{1,h_n}(x)\|_{L^2(P)}\right| &\le C_n\varepsilon^3,
\end{align*}
where $C_n<\infty$ is independent of $x$ and $\varepsilon$.
This proves \eqref{eq:local_distance_chord_approx}.

\subparagraph{Step 1-3:}
We next evaluate $v_n(x)$.
By its definition, differentiating the standardized summand gives
\begin{align}
    h_n\partial_x\tilde\psi_{1,h_n}(x) &= \frac{h_n\partial_xK_{1,h_n}(x)-\mathbb E[h_n\partial_xK_{1,h_n}(x)]}{\{\Var[K_{1,h_n}(x)]\}^{1/2}} -\frac{\Cov\{K_{1,h_n}(x),h_n\partial_xK_{1,h_n}(x)\}}{\Var[K_{1,h_n}(x)]}\tilde\psi_{1,h_n}(x).
    \label{eq:local_standardized_first_derivative}
\end{align}
Taking the expectation of the square of \eqref{eq:local_standardized_first_derivative}, we obtain
\begin{align*}
    h_n^2v_n(x)^2 &= \frac{\Var[h_n\partial_xK_{1,h_n}(x)]}{\Var[K_{1,h_n}(x)]}-\frac{\Cov\{K_{1,h_n}(x),h_n\partial_xK_{1,h_n}(x)\}^2}{\Var[K_{1,h_n}(x)]^2}= \frac{\int_{\mathbb R}\{K'(v)\}^2dv}{\int_{\mathbb R}K(v)^2dv}+O(h_n),
\end{align*}
where the second equality follows from \eqref{eq:local_kernel_variance} and \eqref{eq:local_kernel_derivative_moments}.
Since $\int_{\mathbb R}\{K'(v)\}^2dv>0$, we have
\begin{align*}
    h_nv_n(x) &= \left(\frac{\int_{\mathbb R}\{K'(v)\}^2dv}{\int_{\mathbb R}K(v)^2dv}\right)^{1/2}+O(h_n),
\end{align*}
uniformly over $x\in\mathcal X$.
Substituting this expansion into \eqref{eq:local_spherical_length_identity} gives
\begin{align}
    \mathcal L_{n,j} &= \int_{x_j}^{x_{j+1}}v_n(x)dx \nonumber\\
    &= \frac1{h_n}\int_{x_j}^{x_{j+1}}\left\{\left(\frac{\int_{\mathbb R}\{K'(v)\}^2dv}{\int_{\mathbb R}K(v)^2dv}\right)^{1/2}+O(h_n)\right\}dx \nonumber\\
    &= \frac{x_{j+1}-x_j}{h_n}\left\{\left(\frac{\int_{\mathbb R}\{K'(v)\}^2dv}{\int_{\mathbb R}K(v)^2dv}\right)^{1/2}+O(h_n)\right\},
    \label{eq:local_interval_variation_expansion}
\end{align}
where the last equality uses the uniform remainder bound for $h_nv_n(x)$.

On the other hand, \cref{lem:local_1_disc}(i) and \eqref{eq:local_spherical_chord} imply
\begin{align*}
    c_{n,j} &= 2\sin\{d_n(x_j,x_{j+1})/2\}\\
    &= d_n(x_j,x_{j+1})+O\{d_n(x_j,x_{j+1})^3\}= \frac{x_{j+1}-x_j}{h_n}\left\{\left(\frac{\int_{\mathbb R}\{K'(v)\}^2dv}{\int_{\mathbb R}K(v)^2dv}\right)^{1/2}+o(1)\right\}.
\end{align*}
Combining this expansion with \eqref{eq:local_interval_variation_expansion}, we conclude that
\begin{align}
    \mathcal L_{n,j}+c_{n,j} &= \frac{2(x_{j+1}-x_j)}{h_n}\left\{\left(\frac{\int_{\mathbb R}\{K'(v)\}^2dv}{\int_{\mathbb R}K(v)^2dv}\right)^{1/2}+o(1)\right\},
    \label{eq:local_length_chord_sum}
\end{align}
uniformly over $j\in[p-1]$.

\paragraph{Step 2: Expansion of $\mathcal L_{n,j}^2-c_{n,j}^2$.}
By the fundamental theorem of calculus in $L^2(P)$ and Fubini's theorem,
\begin{align*}
    c_{n,j}^2 &= \mathbb E\left[\left\{\int_{x_j}^{x_{j+1}}\partial_x\tilde\psi_{1,h_n}(s)ds\right\}^2\right]= \int_{x_j}^{x_{j+1}}\int_{x_j}^{x_{j+1}}\mathbb E[\partial_x\tilde\psi_{1,h_n}(s)\partial_x\tilde\psi_{1,h_n}(t)]\,ds\,dt.
\end{align*}
Combining this identity with \eqref{eq:local_spherical_length_identity} and expanding the squared differences gives
\begin{align}
    \mathcal L_{n,j}^2-c_{n,j}^2 &= \int_{x_j}^{x_{j+1}}\int_{x_j}^{x_{j+1}}\left\{v_n(s)v_n(t)-\mathbb E[\partial_x\tilde\psi_{1,h_n}(s)\partial_x\tilde\psi_{1,h_n}(t)]\right\}\,ds\,dt \nonumber\\
    &= \frac12\int_{x_j}^{x_{j+1}}\int_{x_j}^{x_{j+1}}\left[\mathbb E[\{\partial_x\tilde\psi_{1,h_n}(t)-\partial_x\tilde\psi_{1,h_n}(s)\}^2]-\{v_n(t)-v_n(s)\}^2\right]\,ds\,dt,
    \label{eq:local_squared_length_excess}
\end{align}
where we used $v_n(x)^2=\mathbb E[\{\partial_x\tilde\psi_{1,h_n}(x)\}^2]$.
It therefore suffices to establish
\begin{align}
    \mathbb E[\{\partial_x\tilde\psi_{1,h_n}(t)-\partial_x\tilde\psi_{1,h_n}(s)\}^2] &= \frac{(t-s)^2}{h_n^4}\left\{\frac{\int_{\mathbb R}\{K''(v)\}^2dv}{\int_{\mathbb R}K(v)^2dv}+o(1)\right\},
    \label{eq:local_first_increment_target}\\
    \{v_n(t)-v_n(s)\}^2 &= o\left(\frac{(t-s)^2}{h_n^4}\right),
    \label{eq:local_speed_increment_target}
\end{align}
uniformly over distinct $s,t\in[x_j,x_{j+1}]$ and $j\in[p-1]$.

\subparagraph{Step 2-1:}
We first establish \eqref{eq:local_first_increment_target}.
For $s<t$, the fundamental theorem of calculus and Fubini's theorem give
\begin{align}
    \mathbb E[\{\partial_x\tilde\psi_{1,h_n}(t)-\partial_x\tilde\psi_{1,h_n}(s)\}^2] &= \mathbb E\left[\left\{\int_s^t\partial_x^2\tilde\psi_{1,h_n}(y)dy\right\}^2\right] \nonumber\\
    &= \int_s^t\int_s^t\mathbb E[\partial_x^2\tilde\psi_{1,h_n}(y)\partial_x^2\tilde\psi_{1,h_n}(z)]\,dy\,dz.
    \label{eq:local_derivative_increment_integral}
\end{align}
Thus, we evaluate the expectation in \eqref{eq:local_derivative_increment_integral} directly.

For $r\in\{0,1,2\}$, define
\begin{align*}
    \zeta_{n,r}(x) &:= \frac{h_n^r\partial_x^rK_{1,h_n}(x)-\mathbb E[h_n^r\partial_x^rK_{1,h_n}(x)]}{\{\Var[K_{1,h_n}(x)]\}^{1/2}}.
\end{align*}
In particular, $\zeta_{n,0}(x)=\tilde\psi_{1,h_n}(x)$.
Differentiating the definition of $\tilde\psi_{1,h_n}(x)$ with respect to $x$ gives
\begin{align*}
    h_n\partial_x\tilde\psi_{1,h_n}(x) &= \zeta_{n,1}(x)-\mathbb E[\zeta_{n,0}(x)\zeta_{n,1}(x)]\zeta_{n,0}(x).
\end{align*}
Differentiating once more yields
\begin{align}
    h_n^2\partial_x^2\tilde\psi_{1,h_n}(x) &= \zeta_{n,2}(x)-2\mathbb E[\zeta_{n,0}(x)\zeta_{n,1}(x)]h_n\partial_x\tilde\psi_{1,h_n}(x) \nonumber\\
    &\quad-\left\{\mathbb E[\zeta_{n,1}(x)^2]+\mathbb E[\zeta_{n,0}(x)\zeta_{n,2}(x)]-\{\mathbb E[\zeta_{n,0}(x)\zeta_{n,1}(x)]\}^2\right\}\zeta_{n,0}(x).
    \label{eq:local_standardized_second_derivative}
\end{align}

We now express the expectation in \eqref{eq:local_derivative_increment_integral} using $\zeta_{n,2}$.
Adding and subtracting $\zeta_{n,2}(y)$ and $\zeta_{n,2}(z)$ gives
\begin{align*}
    &h_n^4\mathbb E[\partial_x^2\tilde\psi_{1,h_n}(y)\partial_x^2\tilde\psi_{1,h_n}(z)]\\
    &=\mathbb E[\zeta_{n,2}(y)\zeta_{n,2}(z)]+\mathbb E[\{h_n^2\partial_x^2\tilde\psi_{1,h_n}(y)-\zeta_{n,2}(y)\}h_n^2\partial_x^2\tilde\psi_{1,h_n}(z)]+\mathbb E[\zeta_{n,2}(y)\{h_n^2\partial_x^2\tilde\psi_{1,h_n}(z)-\zeta_{n,2}(z)\}].
\end{align*}
By \eqref{eq:local_standardized_second_derivative} and the triangle inequality,
\begin{align*}
    &\|h_n^2\partial_x^2\tilde\psi_{1,h_n}(x)-\zeta_{n,2}(x)\|_{L^2(P)}\\
    &\le 2|\mathbb E[\zeta_{n,0}(x)\zeta_{n,1}(x)]|h_nv_n(x)|+\left|\mathbb E[\zeta_{n,1}(x)^2]+\mathbb E[\zeta_{n,0}(x)\zeta_{n,2}(x)]-\{\mathbb E[\zeta_{n,0}(x)\zeta_{n,1}(x)]\}^2\right|.
\end{align*}
Since $h_nv_n(x)=O(1)$ uniformly over $x\in\mathcal X$ by Step 1, it suffices to establish
\begin{align*}
    \mathbb E[\zeta_{n,0}(x)\zeta_{n,1}(x)] = O(h_n), ~~
    \mathbb E[\zeta_{n,1}(x)^2]+\mathbb E[\zeta_{n,0}(x)\zeta_{n,2}(x)] = O(h_n), ~~
    \max_{0\le r\le2}\|\zeta_{n,r}(x)\|_{L^2(P)} \le C,
\end{align*}
uniformly over $x\in\mathcal X$, and
\begin{align*}
    \mathbb E[\zeta_{n,2}(y)\zeta_{n,2}(z)] &= \frac{\int_{\mathbb R}\{K''(v)\}^2dv}{\int_{\mathbb R}K(v)^2dv}+O\left(h_n+\frac{|y-z|}{h_n}\right)
\end{align*}
uniformly over $y,z\in\mathcal X$ with $|y-z|\le h_n$.
These estimates, together with the Cauchy--Schwarz inequality, imply
\begin{align*}
    h_n^4\mathbb E[\partial_x^2\tilde\psi_{1,h_n}(y)\partial_x^2\tilde\psi_{1,h_n}(z)] &= \frac{\int_{\mathbb R}\{K''(v)\}^2dv}{\int_{\mathbb R}K(v)^2dv}+O\left(h_n+\frac{|y-z|}{h_n}\right).
\end{align*}

We establish the required expectations of $\zeta_{n,r}$ by conditioning on $X_1$ and changing variables.
As in the proof of \cref{lem:local_1_disc}(i), write $b(x)=f(x)\mathbb E[g(Y_1)^2\mid X_1=x]$.
For $r,q\in\{0,1,2\}$ and $|t-s|\le h_n$, conditioning on $X_1$ and changing variables give
\begin{align*}
    &\{h_n\Var[K_{1,h_n}(s)]h_n\Var[K_{1,h_n}(t)]\}^{1/2}\mathbb E[\zeta_{n,r}(s)\zeta_{n,q}(t)]\\
    &=(-1)^{r+q}\int_{\mathbb R}b(s+h_nv)K^{(r)}(v)K^{(q)}\left(v-\frac{t-s}{h_n}\right)dv-h_n\mathbb E[h_n^r\partial_x^rK_{1,h_n}(s)]\mathbb E[h_n^q\partial_x^qK_{1,h_n}(t)].
\end{align*}
The expectations in the last term are uniformly bounded because
\begin{align*}
    |\mathbb E[h_n^r\partial_x^rK_{1,h_n}(x)]| &\le \|g\|_\infty\|f\|_\infty\int_{\mathbb R}|K^{(r)}(v)|dv,\qquad r=0,1,2.
\end{align*}
Moreover, the bounded derivative of $b$, compact support of $K^{(r)}$, and Lipschitz continuity of $K^{(q)}$ imply
\begin{align*}
    &\left|\int_{\mathbb R}b(s+h_nv)K^{(r)}(v)K^{(q)}\left(v-\frac{t-s}{h_n}\right)dv-b(s)\int_{\mathbb R}K^{(r)}(v)K^{(q)}(v)dv\right|
    \le C\left(h_n+\frac{|t-s|}{h_n}\right).
\end{align*}
Also, \eqref{eq:local_kernel_variance} and $|b(t)-b(s)|\le C|t-s|\le Ch_n$ give
\begin{align*}
    \{h_n\Var[K_{1,h_n}(s)]h_n\Var[K_{1,h_n}(t)]\}^{1/2} &= b(s)\int_{\mathbb R}K(v)^2dv+O(h_n).
\end{align*}
Since $\inf_{x\in\mathcal X}b(x)>0$, division yields
\begin{align}
    \mathbb E[\zeta_{n,r}(s)\zeta_{n,q}(t)] &= (-1)^{r+q}\frac{\int_{\mathbb R}K^{(r)}(v)K^{(q)}(v)dv}{\int_{\mathbb R}K(v)^2dv}+O\left(h_n+\frac{|t-s|}{h_n}\right),
    \label{eq:local_normalized_kernel_cross_moments}
\end{align}
uniformly over $r,q\in\{0,1,2\}$ and $s,t\in\mathcal X$ with $|t-s|\le h_n$.

Taking $s=t=x$ in \eqref{eq:local_normalized_kernel_cross_moments} and using integration by parts, we have
\begin{align*}
    \int_{\mathbb R}K(v)K'(v)dv = 0, ~~
    \int_{\mathbb R}K(v)K''(v)dv = -\int_{\mathbb R}\{K'(v)\}^2dv, ~~
    \int_{\mathbb R}K'(v)K''(v)dv = 0.
\end{align*}
Consequently,
\begin{align}
    & \mathbb E[\zeta_{n,0}(x)\zeta_{n,1}(x)] = O(h_n), ~~
    \mathbb E[\zeta_{n,1}(x)^2]+\mathbb E[\zeta_{n,0}(x)\zeta_{n,2}(x)] = O(h_n), \nonumber\\
    & \mathbb E[\zeta_{n,1}(x)\zeta_{n,2}(x)] = O(h_n), ~~
    \max_{0\le r\le2}\|\zeta_{n,r}(x)\|_{L^2(P)} \le C.
    \label{eq:local_normalized_kernel_moments}
\end{align}
These estimates hold uniformly over $x\in\mathcal X$.
Furthermore, taking $(r,q)=(2,2)$ in \eqref{eq:local_normalized_kernel_cross_moments} gives
\begin{align*}
    \mathbb E[\zeta_{n,2}(y)\zeta_{n,2}(z)] &= \frac{\int_{\mathbb R}\{K''(v)\}^2dv}{\int_{\mathbb R}K(v)^2dv}+O\left(h_n+\frac{|y-z|}{h_n}\right),
\end{align*}
uniformly over $y,z\in\mathcal X$ with $|y-z|\le h_n$.

We now return to \eqref{eq:local_standardized_second_derivative}.
By \eqref{eq:local_normalized_kernel_moments}, the first derivative formula, and $h_nv_n(x)=O(1)$ from Step 1, we obtain
\begin{align}
    \|h_n\partial_x\tilde\psi_{1,h_n}(x)-\zeta_{n,1}(x)\|_{L^2(P)} &\le Ch_n, \nonumber\\
    \|h_n^2\partial_x^2\tilde\psi_{1,h_n}(x)-\zeta_{n,2}(x)\|_{L^2(P)} &\le Ch_n,
    \label{eq:local_normalized_derivative_approximation}
\end{align}
uniformly over $x\in\mathcal X$.
In particular, the triangle inequality gives
\begin{align*}
    h_n^2\|\partial_x^2\tilde\psi_{1,h_n}(x)\|_{L^2(P)} &\le \|\zeta_{n,2}(x)\|_{L^2(P)}+Ch_n\le C.
\end{align*}
Applying the Cauchy--Schwarz inequality to the two expectations obtained by adding and subtracting $\zeta_{n,2}(y)$ and $\zeta_{n,2}(z)$, we obtain
\begin{align*}
    &\left|h_n^4\mathbb E[\partial_x^2\tilde\psi_{1,h_n}(y)\partial_x^2\tilde\psi_{1,h_n}(z)]-\mathbb E[\zeta_{n,2}(y)\zeta_{n,2}(z)]\right| \nonumber\\
    &\le \|h_n^2\partial_x^2\tilde\psi_{1,h_n}(y)-\zeta_{n,2}(y)\|_{L^2(P)}h_n^2\|\partial_x^2\tilde\psi_{1,h_n}(z)\|_{L^2(P)}+\|\zeta_{n,2}(y)\|_{L^2(P)}\|h_n^2\partial_x^2\tilde\psi_{1,h_n}(z)-\zeta_{n,2}(z)\|_{L^2(P)}\\
    &\le Ch_n.
\end{align*}
Combining this inequality with \eqref{eq:local_normalized_kernel_cross_moments} for $(r,q)=(2,2)$, we conclude that
\begin{align}
    & h_n^4\mathbb E[\partial_x^2\tilde\psi_{1,h_n}(y)\partial_x^2\tilde\psi_{1,h_n}(z)] \\
    &= \mathbb E[\zeta_{n,2}(y)\zeta_{n,2}(z)]+O(h_n) = \frac{\int_{\mathbb R}\{K''(v)\}^2dv}{\int_{\mathbb R}K(v)^2dv}+O\left(h_n+\frac{|y-z|}{h_n}\right),
    \label{eq:local_second_derivative_cross_moment}
\end{align}
uniformly over $y,z\in\mathcal X$ with $|y-z|\le h_n$.

For $0<t-s\le h_n$, substituting \eqref{eq:local_second_derivative_cross_moment} into \eqref{eq:local_derivative_increment_integral} yields
\begin{align}
    \mathbb E[\{\partial_x\tilde\psi_{1,h_n}(t)-\partial_x\tilde\psi_{1,h_n}(s)\}^2] &= \frac1{h_n^4}\int_s^t\int_s^t\left\{\frac{\int_{\mathbb R}\{K''(v)\}^2dv}{\int_{\mathbb R}K(v)^2dv}+O\left(h_n+\frac{|y-z|}{h_n}\right)\right\}dy\,dz \nonumber\\
    &= \frac{(t-s)^2}{h_n^4}\left\{\frac{\int_{\mathbb R}\{K''(v)\}^2dv}{\int_{\mathbb R}K(v)^2dv}+O\left(h_n+\frac{t-s}{h_n}\right)\right\},
    \label{eq:local_first_derivative_increment}
\end{align}
where the last equality uses $|y-z|\le t-s$ on $[s,t]^2$.
By symmetry, the same expansion holds when $t<s$, with $|t-s|$ in the remainder.
Since
\begin{align*}
    \max_{j\in[p-1]}\sup_{s,t\in[x_j,x_{j+1}]}\frac{|t-s|}{h_n} &\le \frac{\ol x-\ul x}{h_n(p-1)}\to0,
\end{align*}
\eqref{eq:local_first_derivative_increment} establishes \eqref{eq:local_first_increment_target} uniformly over distinct $s,t\in[x_j,x_{j+1}]$ and $j\in[p-1]$.

\subparagraph{Step 2-2:}
We next establish \eqref{eq:local_speed_increment_target}.
By \eqref{eq:local_normalized_derivative_approximation}, the Cauchy--Schwarz inequality, and \eqref{eq:local_normalized_kernel_moments},
\begin{align}
    h_n^3\mathbb E[\partial_x\tilde\psi_{1,h_n}(x)\partial_x^2\tilde\psi_{1,h_n}(x)] &= \mathbb E[\zeta_{n,1}(x)\zeta_{n,2}(x)]+O(h_n)=O(h_n),
    \label{eq:local_standardized_derivative_moments}
\end{align}
uniformly over $x\in\mathcal X$.
Differentiating $v_n(x)^2=\mathbb E[\{\partial_x\tilde\psi_{1,h_n}(x)\}^2]$ and using the uniform positive lower bound on $h_nv_n(x)$ from Step 1 give
\begin{align*}
    |v_n'(x)| &= \frac{|\mathbb E[\partial_x\tilde\psi_{1,h_n}(x)\partial_x^2\tilde\psi_{1,h_n}(x)]|}{v_n(x)}\le \frac{C/h_n^2}{c/h_n}\le \frac{C}{h_n}.
\end{align*}
Therefore, by the fundamental theorem of calculus,
\begin{align}
    \{v_n(t)-v_n(s)\}^2 &= \left\{\int_s^t v_n'(x)dx\right\}^2\le \frac{C(t-s)^2}{h_n^2}=\frac{(t-s)^2}{h_n^4}O(h_n^2),
    \label{eq:local_speed_increment}
\end{align}
which proves \eqref{eq:local_speed_increment_target} uniformly over distinct $s,t\in[x_j,x_{j+1}]$ and $j\in[p-1]$.

\subparagraph{Step 2-3:}
Substituting \eqref{eq:local_first_derivative_increment} and \eqref{eq:local_speed_increment} into \eqref{eq:local_squared_length_excess}, we obtain
\begin{align*}
    \mathcal L_{n,j}^2-c_{n,j}^2 &= \frac1{2h_n^4}\int_{x_j}^{x_{j+1}}\int_{x_j}^{x_{j+1}}(s-t)^2\left\{\frac{\int_{\mathbb R}\{K''(v)\}^2dv}{\int_{\mathbb R}K(v)^2dv}+O\left(h_n+\frac{x_{j+1}-x_j}{h_n}\right)\right\}ds\,dt\\
    &= \frac{(x_{j+1}-x_j)^4}{12h_n^4}\left\{\frac{\int_{\mathbb R}\{K''(v)\}^2dv}{\int_{\mathbb R}K(v)^2dv}+O\left(h_n+\frac{x_{j+1}-x_j}{h_n}\right)\right\},
\end{align*}
where the second equality follows from
\begin{align*}
    \int_{x_j}^{x_{j+1}}\int_{x_j}^{x_{j+1}}(s-t)^2ds\,dt &= \frac{(x_{j+1}-x_j)^4}{6}.
\end{align*}
Since $h_n\to0$ and $(p-1)h_n\to\infty$, we conclude that
\begin{align}
    \mathcal L_{n,j}^2-c_{n,j}^2 &= \frac{(x_{j+1}-x_j)^4}{12h_n^4}\left\{\frac{\int_{\mathbb R}\{K''(v)\}^2dv}{\int_{\mathbb R}K(v)^2dv}+o(1)\right\},
    \label{eq:local_squared_length_excess_expansion}
\end{align}
uniformly over $j\in[p-1]$.
\paragraph{Step 3:}
Combining the expansions of $\mathcal L_{n,j}+c_{n,j}$ and $\mathcal L_{n,j}^2-c_{n,j}^2$ obtained in Steps 1 and 2, respectively, we obtain
\begin{align}
    \mathcal L_{n,j}-c_{n,j} &= \frac{\mathcal L_{n,j}^2-c_{n,j}^2}{\mathcal L_{n,j}+c_{n,j}} = \left(\frac{x_{j+1}-x_j}{h_n}\right)^3\left\{\frac{\displaystyle\frac{\int_{\mathbb R}\{K''(v)\}^2dv}{\int_{\mathbb R}K(v)^2dv}}{\displaystyle24\left(\frac{\int_{\mathbb R}\{K'(v)\}^2dv}{\int_{\mathbb R}K(v)^2dv}\right)^{1/2}}+o(1)\right\}.
    \label{eq:local_chord_excess}
\end{align}
Substituting \eqref{eq:local_chord_excess} and \eqref{eq:local_angle_chord_difference} into \eqref{eq:local_length_excess_decomposition} gives
\begin{align*}
    \mathcal L_{n,j}-d_n(x_j,x_{j+1})
    &= \left(\frac{x_{j+1}-x_j}{h_n}\right)^3\left\{\frac{\displaystyle\frac{\int_{\mathbb R}\{K''(v)\}^2dv}{\int_{\mathbb R}K(v)^2dv}-\left(\frac{\int_{\mathbb R}\{K'(v)\}^2dv}{\int_{\mathbb R}K(v)^2dv}\right)^2}{\displaystyle24\left(\frac{\int_{\mathbb R}\{K'(v)\}^2dv}{\int_{\mathbb R}K(v)^2dv}\right)^{1/2}}+o(1)\right\}\\
    &= \left(\frac{\ol x-\ul x}{h_n(p-1)}\right)^3\left\{\frac{\displaystyle\frac{\int_{\mathbb R}\{K''(v)\}^2dv}{\int_{\mathbb R}K(v)^2dv}-\left(\frac{\int_{\mathbb R}\{K'(v)\}^2dv}{\int_{\mathbb R}K(v)^2dv}\right)^2}{\displaystyle24\left(\frac{\int_{\mathbb R}\{K'(v)\}^2dv}{\int_{\mathbb R}K(v)^2dv}\right)^{1/2}}+o(1)\right\},
\end{align*}
uniformly over $j\in[p-1]$.
\end{proof}

\begin{proof}[Proof of \cref{lem:local_1_disc}(iii)]
By the additivity of the variation induced by $d_n$ over adjacent intervals,
\begin{align*}
    \mathcal L_n &= \sum_{j=1}^{p-1}\mathcal L_{n,j}.
\end{align*}
By \eqref{eq:local_interval_variation_expansion}, established in the proof of \cref{lem:local_1_disc}(ii), we have
\begin{align*}
    \mathcal L_{n,j} &= \frac{x_{j+1}-x_j}{h_n}\left\{\left(\frac{\int_{\mathbb R}\{K'(u)\}^2du}{\int_{\mathbb R}K(u)^2du}\right)^{1/2}+O(h_n)\right\},
\end{align*}
where the $O(h_n)$ term is uniform over $j\in[p-1]$.
Consequently,
\begin{align*}
    \mathcal L_n &= \frac1{h_n}\sum_{j=1}^{p-1}(x_{j+1}-x_j)\left\{\left(\frac{\int_{\mathbb R}\{K'(u)\}^2du}{\int_{\mathbb R}K(u)^2du}\right)^{1/2}+O(h_n)\right\}\\
    &= \frac{\ol x-\ul x}{h_n}\left\{\left(\frac{\int_{\mathbb R}\{K'(u)\}^2du}{\int_{\mathbb R}K(u)^2du}\right)^{1/2}+O(h_n)\right\},
\end{align*}
where the second equality follows from $\sum_{j=1}^{p-1}(x_{j+1}-x_j)=\ol x-\ul x$ and the uniformity of the $O(h_n)$ term.
Since $h_n\to0$, this proves \cref{lem:local_1_disc}(iii).
\end{proof}

\subsection{Proof of \cref{thm:local_grid}}

\begin{proof}[Proof of \cref{thm:local_grid}]

By \cref{lem:local_1_disc}(i), uniformly over $j\in[p-1]$,
\begin{align*}
    d_n(x_j,x_{j+1}) &= \frac{\ol x-\ul x}{h_n(p-1)}\left\{\left(\frac{\int_{\mathbb R}\{K'(u)\}^2du}{\int_{\mathbb R}K(u)^2du}\right)^{1/2}+o(1)\right\}.
\end{align*}
By the definition of $H_t$ and Taylor expansion, as $\theta\to0$ with $t\theta\to0$,
\begin{align*}
    H_t(\theta) &= \int_0^\theta\left[1-\exp\left\{-\frac{t^2}{2}\tan^2\left(\frac{s}{2}\right)\right\}\right]ds\\
    &= \int_0^\theta\left\{\frac{t^2s^2}{8}+O\bigl(t^2(1+t^2)s^4\bigr)\right\}ds= \frac{t^2\theta^3}{24}+O\bigl(t^2(1+t^2)\theta^5\bigr).
\end{align*}
Since $(p-1)h_n\to\infty$ and $t_n/\{(p-1)h_n\}\to0$, we have
\begin{align*}
    \max_{j\in[p-1]}(1+t_n^2)d_n(x_j,x_{j+1})^2 &= O\left(\frac{1+t_n^2}{h_n^2(p-1)^2}\right)=o(1).
\end{align*}
Consequently,
\begin{align*}
    H_{t_n}(d_n(x_j,x_{j+1})) &= \left(\frac{\ol x-\ul x}{h_n(p-1)}\right)^3\left\{\frac{t_n^2}{24}\left(\frac{\int_{\mathbb R}\{K'(u)\}^2du}{\int_{\mathbb R}K(u)^2du}\right)^{3/2}+o(1+t_n^2)\right\},
\end{align*}
uniformly over $j\in[p-1]$.

Combining this expansion with \cref{lem:local_1_disc}(ii) and summing over $j\in[p-1]$, we obtain
\begin{align*}
    &\sum_{j=1}^{p-1}\left\{\mathcal L_{n,j}-d_n(x_j,x_{j+1})+H_{t_n}(d_n(x_j,x_{j+1}))\right\}\\
    &=\frac{(\ol x-\ul x)^3}{24h_n^3(p-1)^2}\left\{\frac{\displaystyle\frac{\int_{\mathbb R}\{K''(u)\}^2du}{\int_{\mathbb R}K(u)^2du}-\left(\frac{\int_{\mathbb R}\{K'(u)\}^2du}{\int_{\mathbb R}K(u)^2du}\right)^2}{\displaystyle\left(\frac{\int_{\mathbb R}\{K'(u)\}^2du}{\int_{\mathbb R}K(u)^2du}\right)^{1/2}}+t_n^2\left(\frac{\int_{\mathbb R}\{K'(u)\}^2du}{\int_{\mathbb R}K(u)^2du}\right)^{3/2}+o(1+t_n^2)\right\}\\
    &=\frac{(\ol x-\ul x)^3}{24h_n^3(p-1)^2}\left\{\frac{\left(\int_{\mathbb R}K(u)^2du\right)\left(\int_{\mathbb R}\{K''(u)\}^2du\right)+(t_n^2-1)\left(\int_{\mathbb R}\{K'(u)\}^2du\right)^2}{\left(\int_{\mathbb R}K(u)^2du\right)^{3/2}\left(\int_{\mathbb R}\{K'(u)\}^2du\right)^{1/2}}+o(1+t_n^2)\right\}.
\end{align*}
Multiplying this expansion by $e^{-t_n^2/2}/\pi$ gives the discretization term in the asserted bound.

Also, by \cref{lem:local_1}(ii), \cref{lem:local_1_disc}(iii), and the uniform positive lower bound on $h_n\Var[K_{1,h_n}(x)]$, we may take $B_n^2=O(h_n^{-1})$ and have $\mathcal L_n=O(h_n^{-1})$.
Since $(p-1)h_n\to\infty$, we have $h_n^{-1}=o(p)$, and hence
\begin{align*}
    \left\{\frac{B_n^2\log^5\{2np(1+\mathcal L_n)\}}{n}\right\}^{1/4} &\le C\left\{\frac{\log^5(2np)}{nh_n}\right\}^{1/4}.
\end{align*}
The conclusion follows from \cref{thm:upper_general}.
\end{proof}

\subsection{Proof of  \cref{lem:local_grid_validity}}

\begin{proof}[Proof of \cref{lem:local_grid_validity}]
Fix $\alpha\in(0,1)$.
To obtain the asserted coverage error bound, it suffices to establish
\begin{align}
    t_n^2 &= 2\log(h_n^{-1})+O(1).
    \label{eq:local_threshold_expansion}
\end{align}
Indeed, this implies
\begin{align*}
    e^{-t_n^2/2} &= O(h_n),\qquad 1+t_n^2=O(\log(h_n^{-1})).
\end{align*}
Since $ph_n/\sqrt{\log(h_n^{-1})}\to\infty$, we also have $(p-1)h_n\to\infty$ and
\begin{align*}
    \frac{t_n}{(p-1)h_n} &= O\left(\frac{\sqrt{\log(h_n^{-1})}}{ph_n}\right)\to0.
\end{align*}
Thus, \cref{thm:local_grid} gives
\begin{align*}
    &\left|\mathbb P\left(\sup_{x\in\mathcal X}|\hat T_n(x)|\le\hat c^\star_{n,1-\alpha}\right)-(1-\alpha)\right|\\
    &\le C\frac{e^{-t_n^2/2}(1+t_n^2)}{h_n^3(p-1)^2}+C\left\{\frac{\log^5(2np)}{nh_n}\right\}^{1/4}=O\left(\frac{\log(h_n^{-1})}{(ph_n)^2}\right)+O\left(\left\{\frac{\log^5(2np)}{nh_n}\right\}^{1/4}\right)=o(1).
\end{align*}
It remains to establish \eqref{eq:local_threshold_expansion}.
In Step 1, we bound $t_n$ by two Gaussian quantiles at fixed probability levels.
In Steps 2 and 3, we bound these quantiles from above and below, respectively.
In Step 4, we combine the bounds.
Throughout, $C,c>0$ denote constants independent of $n$, $p$, and $t$, whose values may change between displays.

\paragraph{Step 1:}
Fix $0<a<1-\alpha<b<1$.
By \cref{lem:local_1}(ii) and \eqref{eq:local_kernel_variance}, we may take $B_n^2=O(h_n^{-1})$.
Since $\log^5(2np)/(nh_n)\to0$, for all sufficiently large $n$,
\begin{align*}
    a &\le \left[1-\alpha-C\left\{\frac{B_n^2\log^5(2np)}n\right\}^{1/4}\right]_+\le b.
\end{align*}
The definition of $t_n$ and the monotonicity of $q_n$ therefore give
\begin{align}
    q_n(a) &\le t_n\le q_n(b).
    \label{eq:local_threshold_quantile_bounds}
\end{align}
For a constant $A>0$ to be chosen, define $t_{n,\pm}:=\{2\log(h_n^{-1})\pm A\}^{1/2}$ for all sufficiently large $n$.
By the definition of $q_n$, it suffices to establish
\begin{align}
    \mathbb P\left(\max_{j\in[p]}|G_n(x_j)|>t_{n,+}\right) &<1-b,
    \label{eq:local_quantile_upper_target}\\
    \mathbb P\left(\max_{j\in[p]}|G_n(x_j)|\le t_{n,-}\right) &<a,
    \label{eq:local_quantile_lower_target}
\end{align}
because these inequalities imply $q_n(b)\le t_{n,+}$ and $q_n(a)\ge t_{n,-}$.

\paragraph{Step 2:}
We prove \eqref{eq:local_quantile_upper_target}.
Write $\bar\Phi(t)=1-\Phi(t)$.
If $\max_{j\in[p]}|G_n(x_j)|>t$ and $|G_n(x_1)|\le t$, there exists $j\in[p-1]$ such that $|G_n(x_j)|\le t<|G_n(x_{j+1})|$.
The union bound, symmetry, and \cref{lem:gaussian_two_point_crossing} therefore yield
\begin{align}
    \mathbb P\left(\max_{j\in[p]}|G_n(x_j)|>t\right) &\le 2\bar\Phi(t)+2\sum_{j=1}^{p-1}\mathbb P\{G_n(x_j)\le t<G_n(x_{j+1})\} \nonumber\\
    &\le 2\bar\Phi(t)+\frac{e^{-t^2/2}}{\pi}\sum_{j=1}^{p-1}d_n(x_j,x_{j+1}) \nonumber\\
    &\le Ch_n^{-1}e^{-t^2/2},
    \label{eq:local_quantile_upper_tail}
\end{align}
for all sufficiently large $n$ and every $t\ge1$.
The last inequality uses $\bar\Phi(t)\le Ce^{-t^2/2}$ and $\sum_{j=1}^{p-1}d_n(x_j,x_{j+1})\le\mathcal L_n=O(h_n^{-1})$ from \cref{lem:local_1_disc}(iii).
Taking $t=t_{n,+}$ gives
\begin{align*}
    \mathbb P\left(\max_{j\in[p]}|G_n(x_j)|>t_{n,+}\right) &\le Ch_n^{-1}e^{-t_{n,+}^2/2}=Ce^{-A/2}.
\end{align*}
Hence, \eqref{eq:local_quantile_upper_target} holds whenever $Ce^{-A/2}<1-b$.

\paragraph{Step 3:}
We prove \eqref{eq:local_quantile_lower_target}.
For $t\ge1$, define $N_n(t):=\sum_{j=1}^p\mathbf1\{|G_n(x_j)|>t\}$.
Since each $G_n(x_j)$ is standard normal, $\mathbb E[N_n(t)]=2p\bar\Phi(t)$.
Chebyshev's inequality gives
\begin{align*}
    \mathbb P\left(\max_{j\in[p]}|G_n(x_j)|\le t\right) &= \mathbb P\{N_n(t)=0\}\le\frac{\Var(N_n(t))}{\{2p\bar\Phi(t)\}^2}.
\end{align*}
We therefore bound $\Var(N_n(t))$.
Since $\Var(\mathbf1\{|G_n(x_j)|>t\})\le2\bar\Phi(t)$, \cref{lem:local_gaussian_indicator} and \eqref{eq:local_M_near_correlation}--\eqref{eq:local_M_far_correlation} imply, for $j\ne k$,
\begin{align*}
    &\Cov\left(\mathbf1\{|G_n(x_j)|>t\},\mathbf1\{|G_n(x_k)|>t\}\right)\\
    &\le C\bar\Phi(t)
    \begin{cases}
        \exp\{-ct^2|x_j-x_k|^2/h_n^2\}, & 0<|x_j-x_k|\le2h_n,\\
        h_n^2, & |x_j-x_k|>2h_n.
    \end{cases}
\end{align*}
Write $\Delta_n=(\ol x-\ul x)/(p-1)$.
There are $p-\ell\le p$ pairs with $k-j=\ell$, and $|x_j-x_k|=\ell\Delta_n$ for each such pair.
Consequently,
\begin{align*}
    \Var(N_n(t)) &\le Cp\bar\Phi(t)\left\{1+\sum_{\ell=1}^{p-1}\exp\left(-ct^2\ell^2\Delta_n^2/h_n^2\right)+ph_n^2\right\}\\
    &\le Cp\bar\Phi(t)\left(1+\frac{ph_n}{t}+ph_n^2\right),
\end{align*}
where the second inequality follows from
\begin{align*}
    \sum_{\ell=1}^{p-1}\exp\left(-ct^2\ell^2\Delta_n^2/h_n^2\right) &\le\int_0^\infty\exp\left(-ct^2u^2\Delta_n^2/h_n^2\right)du\le\frac{Ch_n}{t\Delta_n}\le\frac{Cph_n}{t}.
\end{align*}
Substituting the bound for $\Var(N_n(t))$ into Chebyshev's inequality and using $\bar\Phi(t)\ge ct^{-1}e^{-t^2/2}$ for $t\ge1$, we obtain
\begin{align}
    \mathbb P\left(\max_{j\in[p]}|G_n(x_j)|\le t\right) &\le\frac{C}{\bar\Phi(t)}\left(\frac1p+\frac{h_n}{t}+h_n^2\right) \le Ce^{t^2/2}\left(\frac{t}{p}+h_n+h_n^2t\right).
    \label{eq:local_quantile_lower_tail}
\end{align}
Taking $t=t_{n,-}$ gives
\begin{align*}
    \mathbb P\left(\max_{j\in[p]}|G_n(x_j)|\le t_{n,-}\right) &\le Ce^{-A/2}\left(\frac{t_{n,-}}{ph_n}+1+h_nt_{n,-}\right)=Ce^{-A/2}\{1+o(1)\},
\end{align*}
because $t_{n,-}=O(\sqrt{\log(h_n^{-1})})$, $ph_n/\sqrt{\log(h_n^{-1})}\to\infty$, and $h_n\sqrt{\log(h_n^{-1})}\to0$.
Hence, \eqref{eq:local_quantile_lower_target} holds for all sufficiently large $n$ whenever $Ce^{-A/2}<a/2$.

\paragraph{Step 4:}
Choose a fixed $A>0$ sufficiently large that \eqref{eq:local_quantile_upper_target} and \eqref{eq:local_quantile_lower_target} hold for all sufficiently large $n$.
By \eqref{eq:local_threshold_quantile_bounds} and the definition of $q_n$,
\begin{align*}
    2\log(h_n^{-1})-A &\le q_n(a)^2\le t_n^2\le q_n(b)^2\le2\log(h_n^{-1})+A.
\end{align*}
This proves \eqref{eq:local_threshold_expansion}, and the coverage error bound follows from the argument at the beginning of the proof.
\end{proof}

\subsection{Derivation of \cref{rmk:student_t_grid_selection}}

We derive the grid selection rule by applying \cref{prop:student_t_discretization_bound} to the local empirical process, following the argument in the proof of \cref{thm:local_grid}.
In Step 1, we expand $H_{t_n,\nu_n,\mu_n}$ around zero.
In Step 2, we combine this expansion with \cref{lem:local_1_disc}(i) and (ii) to derive the grid selection rule.

\paragraph{Step 1:}
Allow $t_n\ge0$, $\mu_n>0$, and $\nu_n>2$ to depend on $n$.
Our goal is to obtain a cubic expansion of $H_{t_n,\nu_n,\mu_n}(\theta)$ as $\theta\downarrow0$, analogous to the expansion used in the proof of \cref{thm:local_grid}.
Define
\begin{align*}
    \xi_n&\coloneqq\frac{\mu_n t_n^2(\nu_n^2-1)}{\nu_n(\nu_n-2)}.
\end{align*}
The coefficient $\xi_n$ replaces $t_n^2$ in the proof of \cref{thm:local_grid}.
Indeed, for $z\ge0$, Taylor's theorem gives
\begin{align*}
    1-\left(1+\frac{\xi_n z}{\nu_n+1}\right)^{-(\nu_n+1)/2}&=\frac{\xi_n}{2}z+O(\xi_n^2z^2).
\end{align*}
Substituting $z=\tan^2(s/2)=s^2/4+O(s^4)$ into the definition of $H_{t_n,\nu_n,\mu_n}$ and integrating from zero to $\theta$, we obtain
\begin{align}
    H_{t_n,\nu_n,\mu_n}(\theta)&=\frac{\xi_n\theta^3}{24}+O\{\xi_n(1+\xi_n)\theta^5\}.
    \label{eq:student_t_small_angle}
\end{align}
This expansion holds as $\theta\downarrow0$, uniformly over $t_n\ge0$, $\mu_n>0$, and $\nu_n>2$.

\paragraph{Step 2:}
Suppose that \cref{as:DGP_local} holds and the evaluation grid is equally spaced, with $\Delta_n=(\ol x-\ul x)/(p-1)$.
Our goal is to expand the two terms $\mathcal L_{n,j}-d_n(x_j,x_{j+1})$ and $H_{t_n,\nu_n,\mu_n}(d_n(x_j,x_{j+1}))$ separately, sum their leading terms over the evaluation grid, and derive the grid selection rule.

We first consider $H_{t_n,\nu_n,\mu_n}(d_n(x_j,x_{j+1}))$.
To make the remainder in \eqref{eq:student_t_small_angle} negligible relative to the cubic term, assume that
\begin{align}
    (1+\xi_n)(\Delta_n/h_n)^2&\to0.
    \label{eq:student_t_spacing_condition}
\end{align}
By \cref{lem:local_1_disc}(i),
\begin{align*}
    d_n(x_j,x_{j+1})&=\left(\frac{\int_{\mathbb R}\{K'(u)\}^2du}{\int_{\mathbb R}K(u)^2du}\right)^{1/2}\frac{\Delta_n}{h_n}\{1+o(1)\},
\end{align*}
uniformly over $j\in[p-1]$.
Combining this expansion with \eqref{eq:student_t_small_angle} and \eqref{eq:student_t_spacing_condition}, we obtain
\begin{align*}
    H_{t_n,\nu_n,\mu_n}\bigl(d_n(x_j,x_{j+1})\bigr)&=\frac{\xi_n}{24}\left(\frac{\int_{\mathbb R}\{K'(u)\}^2du}{\int_{\mathbb R}K(u)^2du}\right)^{3/2}\left(\frac{\Delta_n}{h_n}\right)^3\{1+o(1)\},
\end{align*}
uniformly over $j\in[p-1]$.

We next consider $\mathcal L_{n,j}-d_n(x_j,x_{j+1})$.
By \cref{lem:local_1_disc}(ii),
\begin{align*}
    &\mathcal L_{n,j}-d_n(x_j,x_{j+1})\\
    &\quad=\frac{1}{24}\left(\frac{\Delta_n}{h_n}\right)^3\left\{\frac{\int_{\mathbb R}\{K''(u)\}^2du}{\left(\int_{\mathbb R}K(u)^2du\right)^{1/2}\left(\int_{\mathbb R}\{K'(u)\}^2du\right)^{1/2}}-\left(\frac{\int_{\mathbb R}\{K'(u)\}^2du}{\int_{\mathbb R}K(u)^2du}\right)^{3/2}+o(1)\right\},
\end{align*}
uniformly over $j\in[p-1]$.

Adding these two expansions and using
\begin{align*}
    (p-1)\left(\frac{\Delta_n}{h_n}\right)^3&=\frac{(\ol x-\ul x)^3}{h_n^3(p-1)^2},
\end{align*}
we obtain
\begin{align*}
    &\sum_{j=1}^{p-1}\left\{\mathcal L_{n,j}-d_n(x_j,x_{j+1})+H_{t_n,\nu_n,\mu_n}\bigl(d_n(x_j,x_{j+1})\bigr)\right\}\\
    &\quad=\frac{(\ol x-\ul x)^3}{24h_n^3(p-1)^2}\Bigg\{\frac{\int_{\mathbb R}\{K''(u)\}^2du}{\left(\int_{\mathbb R}K(u)^2du\right)^{1/2}\left(\int_{\mathbb R}\{K'(u)\}^2du\right)^{1/2}}\\
    &\hspace{95pt}+(\xi_n-1)\left(\frac{\int_{\mathbb R}\{K'(u)\}^2du}{\int_{\mathbb R}K(u)^2du}\right)^{3/2}+o(1+\xi_n)\Bigg\}.
\end{align*}

We now verify that the spacing condition in \eqref{eq:student_t_spacing_condition} follows from the same rate conditions as in the Gaussian case.
For $\mu_n=(n-1)/n$ and $\nu_n\ge3$,
\begin{align*}
    \xi_n&=\frac{\mu_n t_n^2(\nu_n^2-1)}{\nu_n(\nu_n-2)}\le\frac83t_n^2.
\end{align*}
Consequently,
\begin{align*}
    (1+\xi_n)(\Delta_n/h_n)^2&\le(\Delta_n/h_n)^2+\frac83(t_n\Delta_n/h_n)^2\to0
\end{align*}
whenever $\Delta_n/h_n\to0$ and $t_n\Delta_n/h_n\to0$, as in the proof of \cref{thm:local_grid}.

Finally, substituting the expansion above into \cref{prop:student_t_discretization_bound}, multiplying by $(1+\mu_n t_n^2/\nu_n)^{-(\nu_n-1)/2}/\pi$ and setting $t_n=q_n(1-\alpha)$ give the leading term in \cref{rmk:student_t_grid_selection}.

\bibliography{ref_FINE}

@article{chernozhukov2013intersection,
  author  = {Chernozhukov, Victor and Lee, Sokbae and Rosen, Adam M.},
  title   = {Intersection Bounds: Estimation and Inference},
  journal = {Econometrica},
  year    = {2013},
  volume  = {81},
  number  = {2},
  pages   = {667--737},
  doi     = {10.3982/ECTA8718}
}

@article{andrews2014nonparametric,
  author  = {Andrews, Donald W. K. and Shi, Xiaoxia},
  title   = {Nonparametric Inference Based on Conditional Moment Inequalities},
  journal = {Journal of Econometrics},
  year    = {2014},
  volume  = {179},
  number  = {1},
  pages   = {31--45},
  doi     = {10.1016/j.jeconom.2013.10.005}
}

@book{wasserman2006all,
  title={All of nonparametric statistics},
  author={Wasserman, Larry},
  year={2006},
  publisher={Springer Science \& Business Media}
}

@article{bickel1973some,
  author  = {Bickel, Peter J. and Rosenblatt, Murray},
  title   = {On Some Global Measures of the Deviations of Density Function Estimates},
  journal = {The Annals of Statistics},
  year    = {1973},
  volume  = {1},
  number  = {6},
  pages   = {1071--1095},
  doi     = {10.1214/aos/1176342558}
}

@book{gine2016mathematical,
  author    = {Gin{\'e}, Evarist and Nickl, Richard},
  title     = {Mathematical Foundations of Infinite-Dimensional Statistical Models},
  publisher = {Cambridge University Press},
  year      = {2016}
}

@article{piterbarg2004discrete,
  author  = {Piterbarg, Vladimir I.},
  title   = {Discrete and Continuous Time Extremes of Gaussian Processes},
  journal = {Extremes},
  volume  = {7},
  number  = {2},
  pages   = {161--177},
  year    = {2004},
  doi     = {10.1007/s10687-005-6198-8}
}

@article{tan2014multivariate,
  author  = {Tan, Zhongquan and Hashorva, Enkelejd},
  title   = {On Piterbarg's max-discretisation theorem for multivariate stationary Gaussian processes},
  journal = {Journal of Mathematical Analysis and Applications},
  volume  = {409},
  number  = {1},
  pages   = {299--314},
  year    = {2014},
  doi     = {10.1016/j.jmaa.2013.07.022}
}

@article{hashorva2015grids,
  author  = {Hashorva, Enkelejd and Tan, Zhongquan},
  title   = {Piterbarg's max-discretization theorem for stationary vector Gaussian processes observed on different grids},
  journal = {Statistics},
  volume  = {49},
  number  = {2},
  pages   = {338--360},
  year    = {2015},
  doi     = {10.1080/02331888.2014.982653}
}

@article{tan2015homogeneous,
  author  = {Tan, Zhongquan and Wang, Kaiyong},
  title   = {On Piterbarg's max-discretisation theorem for homogeneous Gaussian random fields},
  journal = {Journal of Mathematical Analysis and Applications},
  volume  = {429},
  number  = {2},
  pages   = {969--994},
  year    = {2015},
  doi     = {10.1016/j.jmaa.2015.04.048}
}

@article{azais2020discretization,
  author  = {Azaïs, Jean-Marc and Chassan, Malika},
  title   = {Discretization error for the maximum of a Gaussian field},
  journal = {Stochastic Processes and their Applications},
  volume  = {130},
  number  = {2},
  pages   = {545--559},
  year    = {2020},
  doi     = {10.1016/j.spa.2019.02.002}
}

@article{kato2019eiv,
  author  = {Kato, Kengo and Sasaki, Yuya},
  title   = {Uniform Confidence Bands for Nonparametric Errors-in-Variables Regression},
  journal = {Journal of Econometrics},
  year    = {2019},
  volume  = {213},
  number  = {2},
  pages   = {516--555},
  doi     = {10.1016/j.jeconom.2019.05.021}
}

@article{imai2026catt,
  author  = {Imai, Shunsuke and Qin, Lei and Yanagi, Takahide},
  title   = {Doubly Robust Uniform Confidence Bands for Group-Time Conditional Average Treatment Effects in Difference-in-Differences},
  journal = {Journal of Business \& Economic Statistics},
  year    = {2026},
  volume  = {44},
  number  = {2},
  pages   = {511--523},
  doi     = {10.1080/07350015.2025.2541719}
}

@article{colangelo2026double,
  title   = {Double Debiased Machine Learning Nonparametric Inference with Continuous Treatments},
  author  = {Colangelo, Kyle and Lee, Ying-Ying},
  journal = {Journal of Business \& Economic Statistics},
  year    = {2026},
  volume  = {44},
  number  = {1},
  pages   = {67--79},
  doi     = {10.1080/07350015.2025.2505487}
}

@article{BellMcCaffrey2002,
  author  = {Bell, Robert M. and McCaffrey, Daniel F.},
  title   = {Bias Reduction in Standard Errors for Linear Regression with Multi-Stage Samples},
  journal = {Survey Methodology},
  year    = {2002},
  volume  = {28},
  number  = {2},
  pages   = {169--181},
  url     = {https://www150.statcan.gc.ca/n1/pub/12-001-x/2002002/article/9058-eng.pdf}
}

@article{ImbensKolesar2016,
  author  = {Imbens, Guido W. and Koles{\'a}r, Michal},
  title   = {Robust Standard Errors in Small Samples: Some Practical Advice},
  journal = {The Review of Economics and Statistics},
  year    = {2016},
  volume  = {98},
  number  = {4},
  pages   = {701--712},
  doi     = {10.1162/REST_a_00552}
}

@article{Welch1951,
  author  = {Welch, B. L.},
  title   = {On the Comparison of Several Mean Values: An Alternative Approach},
  journal = {Biometrika},
  year    = {1951},
  volume  = {38},
  number  = {3--4},
  pages   = {330--336},
  doi     = {10.1093/biomet/38.3-4.330}
}

@book{doob1953stochastic,
  author    = {Doob, J. L.},
  title     = {Stochastic Processes},
  publisher = {John Wiley \& Sons},
  address   = {New York},
  year      = {1953}
}

@article{gine2008uniform,
  author  = {Gin{\'e}, Evarist and Nickl, Richard},
  title   = {Uniform central limit theorems for kernel density estimators},
  journal = {Probability Theory and Related Fields},
  year    = {2008},
  volume  = {141},
  pages   = {333--387},
  doi     = {10.1007/s00440-007-0087-9}
}

@article{hall1988confidence,
  author  = {Hall, Peter and Titterington, D. M.},
  title   = {On Confidence Bands in Nonparametric Density Estimation and Regression},
  journal = {Journal of Multivariate Analysis},
  year    = {1988},
  volume  = {27},
  number  = {1},
  pages   = {228--254},
  doi     = {10.1016/0047-259X(88)90127-3}
}

@article{CCKK26,
  author        = {Chernozhukov, Victor and Chetverikov, Denis and Kato, Kengo and Koike, Yuta},
  title         = {Sharp High-Dimensional Gaussian Approximations over Rectangles without Covariance Restrictions},
  year          = {2026},
  journal = {arXiv 2611.21474},
  url           = {https://arxiv.org/abs/2611.21474}
}

@article{ma2019inference,
  title   = {Inference for First-Price Auctions with Guerre, Perrigne, and Vuong's Estimator},
  author  = {Ma, Jun and Marmer, Vadim and Shneyerov, Artyom},
  journal = {Journal of Econometrics},
  volume  = {211},
  number  = {2},
  pages   = {507--538},
  year    = {2019},
  doi     = {10.1016/j.jeconom.2019.02.006}
}

@article{babii2020honest,
  title   = {Honest Confidence Sets in Nonparametric {IV} Regression and Other Ill-Posed Models},
  author  = {Babii, Andrii},
  journal = {Econometric Theory},
  volume  = {36},
  number  = {4},
  pages   = {658--706},
  year    = {2020},
  doi     = {10.1017/S0266466619000380}
}

@article{tsuda2025uniform,
  title         = {A Uniform Confidence Band for the Marginal Treatment Effect Function},
  author        = {Tsuda, Toshiki and Jin, Yanchun and Okui, Ryo},
  year          = {2025},
  journal = {arXiv 2501.17455},
  url           = {https://arxiv.org/abs/2501.17455}
}

@article{sasaki2022unconditional,
  title   = {Unconditional Quantile Regression with High-Dimensional Data},
  author  = {Sasaki, Yuya and Ura, Takuya and Zhang, Yichong},
  journal = {Quantitative Economics},
  volume  = {13},
  number  = {3},
  pages   = {955--978},
  year    = {2022},
  doi     = {10.3982/QE1896}
}

@article{hou2026quantile,
  author  = {Hou, Haoze and Huang, Wei and Zhang, Zheng},
  title   = {Nonparametric Quantile Regression and Uniform Inference with Unknown Error Distribution},
  journal = {Journal of Business \& Economic Statistics},
  year    = {2026},
  volume  = {44},
  number  = {1},
  pages   = {1--12},
  doi     = {10.1080/07350015.2025.2486009}
}

@article{ma2023ite,
  author  = {Ma, Jun and Marmer, Vadim and Yu, Zhengfei},
  title   = {Inference on Individual Treatment Effects in Nonseparable Triangular Models},
  journal = {Journal of Econometrics},
  year    = {2023},
  volume  = {235},
  number  = {2},
  pages   = {2096--2124},
  doi     = {10.1016/j.jeconom.2023.02.011}
}

@article{cattaneo2026sharp,
  author  = {Cattaneo, Matias D. and Masini, Ricardo P. and Underwood, William G.},
  title   = {Sharp Anti-Concentration Inequalities for Extremum Statistics via Copulas},
  journal = {Bernoulli},
  year    = {2026},
  volume  = {32},
  number  = {3},
  pages   = {2022--2045},
  doi     = {10.3150/25-BEJ1943}
}

@article{husler2004dependence,
  author  = {H{\"u}sler, J{\"u}rg},
  title   = {Dependence Between Extreme Values of Discrete and Continuous Time Locally Stationary Gaussian Processes},
  journal = {Extremes},
  volume  = {7},
  number  = {2},
  pages   = {179--190},
  year    = {2004},
  doi     = {10.1007/s10687-005-6199-7}
}

@article{Plackett54,
  author  = {Plackett, R. L.},
  title   = {A reduction formula for normal multivariate integrals},
  journal = {Biometrika},
  year    = {1954},
  volume  = {41},
  number  = {3--4},
  pages   = {351--360},
  doi     = {10.1093/biomet/41.3-4.351}
}

@article{horowitz2012uniform,
  author  = {Horowitz, Joel L. and Lee, Sokbae},
  title   = {Uniform Confidence Bands for Functions Estimated Nonparametrically with Instrumental Variables},
  journal = {Journal of Econometrics},
  year    = {2012},
  volume  = {168},
  number  = {2},
  pages   = {175--188},
  doi     = {10.1016/j.jeconom.2011.12.001}
}

@article{bentkus2003dependence,
  author  = {Bentkus, V.},
  title   = {On the Dependence of the {Berry--Esseen} Bound on Dimension},
  journal = {Journal of Statistical Planning and Inference},
  year    = {2003},
  volume  = {113},
  number  = {2},
  pages   = {385--402},
  doi     = {10.1016/S0378-3758(02)00094-0}
}

@article{nourdin2019gebel,
  author  = {Nourdin, Ivan and Peccati, Giovanni and Yang, Xiaochuan},
  title   = {Berry--Esseen Bounds in the Breuer--Major {CLT} and {Gebelein's} Inequality},
  journal = {Electronic Communications in Probability},
  volume  = {24},
  number  = {34},
  pages   = {1--12},
  year    = {2019},
  doi     = {10.1214/19-ECP241}
}

@article{lee2017doubly,
  title={Doubly robust uniform confidence band for the conditional average treatment effect function},
  author={Lee, Sokbae and Okui, Ryo and Whang, Yoon-Jae},
  journal={Journal of Applied Econometrics},
  volume={32},
  number={7},
  pages={1207--1225},
  year={2017},
  publisher={Wiley Online Library}
}

@article{fan2022estimation,
  title={Estimation of conditional average treatment effects with high-dimensional data},
  author={Fan, Qingliang and Hsu, Yu-Chin and Lieli, Robert P and Zhang, Yichong},
  journal={Journal of Business \& Economic Statistics},
  volume={40},
  number={1},
  pages={313--327},
  year={2022},
  publisher={Taylor \& Francis}
}

@book{van1996weak,
  title={Weak Convergence and Empirical Processes: With Applications to Statistics},
  author={van der Vaart, Aad and Wellner, Jon A},
  year={1996},
  publisher={Springer}
}

@article{chernozhukov2014anti,
  title={Anti-concentration and honest, adaptive confidence bands},
  author={Chernozhukov, Victor and Chetverikov, Denis and Kato, Kengo},
  journal={The Annals of Statistics},
  volume={42},
  number={5},
  pages={1787},
  year={2014},
  publisher={Institute of Mathematical Statistics}
}

@article{chernozhukov2014gaussian,
  title={Gaussian approximation of suprema of empirical processes},
  author={Chernozhukov, Victor and Chetverikov, Denis and Kato, Kengo},
  journal={The Annals of Statistics},
  volume={42},
  number={4},
  pages={1564--1597},
  year={2014}
}

@article{belloni2015some,
  title={Some new asymptotic theory for least squares series: Pointwise and uniform results},
  author={Belloni, Alexandre and Chernozhukov, Victor and Chetverikov, Denis and Kato, Kengo},
  journal={Journal of Econometrics},
  volume={186},
  number={2},
  pages={345--366},
  year={2015},
  publisher={Elsevier}
}

@article{chernozhukov2016empirical,
  title={Empirical and multiplier bootstraps for suprema of empirical processes of increasing complexity, and related Gaussian couplings},
  author={Chernozhukov, Victor and Chetverikov, Denis and Kato, Kengo},
  journal={Stochastic Processes and their Applications},
  volume={126},
  number={12},
  pages={3632--3651},
  year={2016},
  publisher={Elsevier}
}

@article{le1986asymptotic,
  title={Asymptotic Methods in Statistical Decision Theory},
  author={Le Cam, Lucien},
  journal={Springer Series in Statistics},
  year={1986},
  publisher={Springer New York}
}

@article{kato2018uniform,
  title={Uniform confidence bands in deconvolution with unknown error distribution},
  author={Kato, Kengo and Sasaki, Yuya},
  journal={Journal of Econometrics},
  volume={207},
  number={1},
  pages={129--161},
  year={2018},
  publisher={Elsevier}
}

@article{cattaneo2025strong,
  author  = {Cattaneo, Matias D. and Yu, Ruiqi Rae},
  title   = {Strong approximations for empirical processes indexed by {Lipschitz} functions},
  journal = {The Annals of Statistics},
  year    = {2025},
  volume  = {53},
  number  = {3},
  pages   = {1203--1229},
  doi     = {10.1214/25-AOS2500}
}

@article{rio1994local,
  author  = {Rio, Emmanuel},
  title   = {Local invariance principles and their application to density estimation},
  journal = {Probability Theory and Related Fields},
  year    = {1994},
  volume  = {98},
  number  = {1},
  pages   = {21--45},
  doi     = {10.1007/BF01311347}
}

@article{koltchinskii1994komlos,
  author  = {Koltchinskii, Vladimir I.},
  title   = {{Koml{\'o}s--Major--Tusn{\'a}dy} approximation for the general empirical process and {Haar} expansions of classes of functions},
  journal = {Journal of Theoretical Probability},
  year    = {1994},
  volume  = {7},
  number  = {1},
  pages   = {73--118},
  doi     = {10.1007/BF02213361}
}

@article{gine2010confidence,
  author  = {Gin{\'e}, Evarist and Nickl, Richard},
  title   = {Confidence bands in density estimation},
  journal = {The Annals of Statistics},
  year    = {2010},
  volume  = {38},
  number  = {2},
  pages   = {1122--1170},
  doi     = {10.1214/09-AOS738}
}

@article{ChKa20,
  title={Jackknife multiplier bootstrap: finite sample approximations to the {$U$}-process supremum with applications},
  author={Chen, Xiaohui and Kato, Kengo},
  journal={Probability Theory and Related Fields},
  volume={176},
  pages={1097--1163},
  year={2020},
  publisher={Springer}
}

@article{CCK13,
	author = {Chernozhukov, Victor and Chetverikov, Denis and Kato, Kengo},
	journal = {Annals of Statistics},
	number = {6},
	pages = {2786--2819},
	title = {GAUSSIAN APPROXIMATIONS AND MULTIPLIER BOOTSTRAP FOR MAXIMA OF SUMS OF HIGH-DIMENSIONAL RANDOM VECTORS},
	volume = {41},
	year = {2013}}

@article{CCKK22,
	author = {Chernozhukov, Victor and Chetverikov, Denis and Kato, Kengo and Koike, Yuta},
	journal = {Annals of Statistics},
	number = {5},
	pages = {2562--2586},
	title = {Improved central limit theorem and bootstrap approximation in high dimensions},
	volume = {50},
	year = {2022}}

@article{chernozhukov2023high,
  title={High-dimensional data bootstrap},
  author={Chernozhukov, Victor and Chetverikov, Denis and Kato, Kengo and Koike, Yuta},
  journal={Annual Review of Statistics and Its Application},
  volume={10},
  number={1},
  pages={427--449},
  year={2023},
  publisher={Annual Reviews}
}

@book{vershynin2018high,
  title     = {High-Dimensional Probability: An Introduction with Applications in Data Science},
  author    = {Vershynin, Roman},
  year      = {2018},
  publisher = {Cambridge University Press},
  series    = {Cambridge Series in Statistical and Probabilistic Mathematics},
  volume    = {47},
  doi       = {10.1017/9781108231596}
}
\bibliographystyle{apalike}
\end{document}